\documentclass[noinfoline]{imsart}

\RequirePackage[OT1]{fontenc}
\RequirePackage{amsthm,amsmath}
\RequirePackage{natbib}
\RequirePackage[colorlinks,citecolor=blue,urlcolor=blue]{hyperref}

\usepackage{verbatim}
\usepackage{comment}

\allowdisplaybreaks

\usepackage{fullpage}

\newcommand{\E}{\mathbb{E}}
\def\Ez#1{\mathbb{E} \left( #1 \right)}

\newcommand {\cov}{\textrm{Cov}}
\newcommand {\var}{\textrm{Var}}

\newcommand {\QBartlett}{L}

\newcommand {\Nh}{\mathcal{N}_h}
\newcommand {\Op}{O_{\mathbb{P}}}
\newcommand {\op}{o_{\mathbb{P}}}
\DeclareMathOperator{\I}{i}
\DeclareMathOperator{\cum}{cum}
\DeclareMathOperator*{\argmin}{arg\,min}

\DeclareMathOperator{\diag}{diag}
\DeclareMathOperator{\tr}{tr}
\DeclareMathOperator{\Prob}{\mathbb{P}}
\newcommand{\D}[1]{\ensuremath{\operatorname{d}\!{#1}}}
\newcommand{\transpose}{^{\top}}
\newcommand{\boldsymbolmu}{\mathbb{M}}
\newcommand{\mathbbSigma}{\mathbb{S}}

\usepackage{subfigure}
\usepackage{layout}

\usepackage{dsfont}
\usepackage[mathscr]{eucal}
\usepackage[toc,page]{appendix}
\usepackage{mathrsfs}
\usepackage{color}
\usepackage{pifont}
\usepackage{bm}
\usepackage{latexsym}
\usepackage{amsfonts}
\usepackage{amssymb}
\usepackage{epsfig}
\usepackage{graphicx}
\usepackage{multirow}
\usepackage{enumitem}
\usepackage{mathbbol}
\usepackage[dvipsnames]{xcolor}
\usepackage{stackrel}

\newtheorem{theorem}{Theorem}
\newtheorem{proposition}{Proposition}
\newtheorem{corollary}{Corollary}
\newtheorem{lemma}{Lemma}

\startlocaldefs
\numberwithin{equation}{section}
\theoremstyle{plain}
\endlocaldefs

\begin{document}

\begin{frontmatter}

\title{{\large Sparsely Observed Functional Time Series: Estimation and Prediction}}

\runtitle{Sparsely Observed Functional Time Series: Estimation and Prediction}

\begin{aug}
%\author{\fnms{Marie-H\'el\`ene} \snm{Descary}\ead[label=e1]{marie-helene.descary@epfl.ch}} \and
\author{\fnms{Tom{\'a}{\v s}} \snm{Rub{\'i}n}\ead[label=e1]{tomas.rubin@epfl.ch}} \and
\author{\fnms{Victor M.} \snm{Panaretos}\ead[label=e2]{victor.panaretos@epfl.ch}}

%\thankstext{t1}{Research supported by a Swiss National Science Foundation grant.}

\runauthor{T. Rub{\'i}n \& V.M. Panaretos}

\affiliation{Ecole Polytechnique F\'ed\'erale de Lausanne}

\address{Institut de Math\'ematiques\\
Ecole Polytechnique F\'ed\'erale de Lausanne\\
\printead{e1}, \printead*{e2}}

\end{aug}

\begin{abstract}
Functional time series analysis, whether based on time or frequency domain methodology, has traditionally been carried out under the assumption of complete observation of the constituent series of curves, assumed stationary. Nevertheless, as is often the case with independent functional data, it may well happen that the data available to the analyst are not the actual sequence of curves, but relatively few and noisy measurements per curve, potentially at different locations in each curve's domain. Under this sparse sampling regime, neither the established estimators of the time series' dynamics nor their corresponding theoretical analysis will apply. The subject of this paper is to tackle the problem of estimating the dynamics and of recovering the latent process of smooth curves in the sparse regime. Assuming smoothness of the latent curves, we construct a consistent nonparametric estimator of the series' spectral density operator and use it to develop a frequency-domain recovery approach, that predicts the latent curve at a given time by borrowing strength from the (estimated) dynamic correlations in the series across time.  {This new methodlogy is seen to comprehensively outperform a naive recovery approach that would ignore temporal dependence and use only methodology employed in the i.i.d. setting and hinging on the lag zero covariance.}  Further to predicting the latent curves from their noisy point samples, the method fills in gaps in the sequence (curves nowhere sampled), denoises the data, and serves as a basis for forecasting. Means of providing corresponding confidence bands are also investigated. A simulation study interestingly suggests that sparse observation for a longer time period may provide better performance than dense observation for a shorter period, in the presence of smoothness. The methodology is further illustrated by application to an environmental data set on fair-weather atmospheric electricity, which naturally leads to a sparse functional time series.
\end{abstract}


\end{frontmatter}

\tableofcontents

\newpage
\section{Introduction}

Functional data analysis constitutes a collection of statistical methods to analyse data comprised of ensembles of random functions: multiple occurrences of random processes evolving continuously in time and/or space, typically over a bounded rectangular domain \citep{ramsay2007applied,ferraty2006nonparametric,hsing2015theoretical,wang2016functional}. The challenges arising in functional data, on the one hand, arise from their infinite-dimensional nature: this calls upon tools and techniques from functional analysis, while standard inference problems may become ill-posed. On the other hand, the data, though continuous in nature, are seldom observed as such. Instead, finitely sampled versions are available to the statistician. If the sampling is sufficiently dense, the data can often be treated as genuinely functional data, possibly after a pre-smoothing step. The statistical estimators and procedures may be then based on the intrinsically infinite dimensional inputs and techniques. This approach was popularised by \citet{ramsay2007applied}.

It can very well happen, though, that the data are recorded only at some intermediate locations of their domain, possibly corrupted by measurement error.  In this case, it is necessary to regard the underlying functional nature of the data only as a latent process, and additional effort is required to construct adequate statistical methodology.  This scenario is often referred to as \emph{sparsely observed functional data} and usually occurs when the independent realisations of the latent functional process is a longitudinal trajectory. In a key paper,  \citet{yao2005functional} demonstrated how to estimate the covariance operator of the latent functional process using kernel regression and how to estimate the principal components of the latent process through conditional expectations. See also  \citet{yao2005functional_linear_regression} for an application of the proposed methodology in functional linear regression. The rate of convergence of the kernel smoother of  \citet{yao2005functional} was later strengthened by \citet{hall2006properties} and  \citet{li2010uniform}. Other methods to deal with sparsely observed functional data make use of minimizing a specific convex criterion function and expressing the estimator within a reproducing kernel Hilbert space, see \citet{cai2010nonparametric}, and \citet{wong2017nonparametric}.

Still, there are many applications where independence of the underlying curves cannot be assumed, for instance when the functional data are naturally ordered into a temporal sequence indexed by discrete time. We then speak of functional time series, and these are usually analysed by assuming stationarity and weak dependence across the time index. Historically, the research has been focused mostly into generalizing linear processes into functional spaces, see  \citet{bosq2012linear} and \citet{blanke2007inference} for overview publications. More recently, the research has moved beyond the linear structure. \citet{hormann2010weakly} considered the effect of weak dependence on principal component analysis and studied the estimation of the long-run covariance operator. \citet{horvath2013estimation} provided a central limit theorem for the mean of a stationary weak dependent sequence and considered the estimation of the long-run covariance operator.

A step further from the estimation of isolated characteristics such as the mean function and the said long-run covariance operator is to estimate the entire second-order structure of the process, without assuming linearity. To this aim, \citet{panaretos2013cramer} introduced the notation of spectral density operators and harmonic principal components, capturing the complete second-order dynamics in the frequency domain, whereas \citet{panaretos2013fourier} showed how to estimate the said spectral density operators by smoothing the operator-valued analogue of the periodogram. They formalised weak dependence by cumulant-type mixing conditions, \`a la \citet{brillinger1981time}. In parallel work, \citet{hormann2015dynamic} introduced the notation of dynamic principal components, closely related to the harmonic principal components of \citet{panaretos2013cramer}, and estimated the spectral density operators by the operator version of Bartlett's estimate \citep{bartlett1950periodogram}.

Despite the long tradition of functional time series as a driving force behind theoretical and methodological progress in functional data analysis more generally, a surprising fact is that the focus has been almost exclusively ``densely" observed functional time series, where it is assumed that the full functional data are available. Indeed discrete sampling appears to be a nearly absent consideration, with the exceptions (to our knowledge) being: \citet{panaretos2013fourier}, who show the stability of their asymptotics under \emph{dense} discrete observation but with measurement error of decaying magnitude; and, more recently, \citet{kowal2017functional} who studied functional autoregressive models by means of Bayesian hierarchical Gaussian models. They derived a Gibbs sampler for inference and forecasting but the paper does not examine the asymptotic behaviour of the method. In particular, in one of their considered sampling regimes, which they call \textit{sparse-fixed} design, posterior Bayesian concentration would be intangible. The Bayesian modelling framework was also extended to multivariate dynamic linear models by \citet{kowal2017bayesian} and to dynamic function-on-scalar regression by  \citet{kowal2018dynamic}. A related problem was studied by \citet{paul2011principal}, who considered correlated sparsely observed functional data with separable covariance structure, but the focus was not on dynamics.

In this article we address this gap (or, rather, chasm) and consider the problem of estimating the complete dynamics, and recovering the latent curves, in a stationary functional time series that is observed \emph{sparsely}, \emph{irregularly}, and with \emph{measurement errors}. The number of observations per curve is assumed to be random, almost surely finite, and \emph{not} increasing to infinity. Therefore we speak of genuine sparsity, much in the same vein as \citet{yao2005functional}. As a first step, we show how to estimate the full second-order dynamics of the functional time series based on sparse noisy data using kernel regression methods. We construct estimators of individual characteristics such as the mean function and the lag autocovariance operators, as an aside, but the main contribution is the kernel-based generalization of Bartlett's estimate of the spectral density operators. By integrating back the spectral density into the time domain we construct a consistent estimator of the entire space-time covariance structure.

Our methodology can also be interpreted in a design context: in certain applications, it might be possible for the scientist to choose how to distribute a given fixed budget of measurements over individual curves and over time. In this case, one might ask how to better estimate the underlying dynamics: whether it is better to sample a functional time series more densely over shorter time-span, or to record fewer observations per curve but over a longer time-space. In Section \ref{sec:simulation} we perform a simulation study to examine this tradeoff, and find that under sufficient smoothness, the sparse sampling regime over a longer period seems preferable.

The second contribution of the article is the establishment of a functional data recovery framework. We show how to predict the unobserved functional data once the space-time dynamics have been estimated. The recovery of the functional data is done by conditioning on all observed data, borrowing strength from the complete dynamics of the process (rather than just the marginal covariance).  {Simulations show that this approach comprehensively outperfoms a naive approach that would ignore dependence and simply employ the methodology of \citet{yao2005functional} using only the lag zero covariance, as one would do in the i.i.d. case (see e.g. Table \ref{table:recovery_relative_gain_MA4})}. When the functional time series is Gaussian, we furthermore show how to construct confidence bands for the latent functional data, with both pointwise and simultaneous coverage. In addition, we show how the functional recovery methodology naturally leads to forecasting.

Functional time series methodology is often useful in analysing continuously measured scalar time series, that can be subdivided into segments of an obvious periodicity, usually days. A key benefit of this technique is the separation of the intra-day variability and the temporal dependence among the consecutive days. The approach is especially fruitful in the analysis of environmental or meteorological phenomena, for example, particulate matter atmospheric pollution
\citep{hormann2010weakly,hormann2015dynamic,hormann2016detection,aue2015prediction}.
Nonetheless, some meteorological variables cannot be measured continuously and uninterruptedly. A practical motivation of this article comes from the data on atmospheric electricity \citep{tammet2009joint}. The peculiarity of this data is that the atmospheric electricity can be reliably measured only in fair-weather conditions. Otherwise, the physical-chemical processes behind the atmospheric electricity are altered and thus a different kind of process is measured. Details of this mechanism are reported in the data analysis in Section \ref{sec:data analysis}. Because of this censoring protocol, the considered functional time series is genuinely sparsely observed. We analyse such a dataset using our proposed methods, as a means of illustration.

The rest of the article is organised as follows. In Section \ref{sec:model} we define the functional time-series framework we work with and introduce the estimation and prediction methodology. In Section \ref{sec:asymptotics} we formulate the asymptotic theory for the suggested estimators under two different sets of assumptions: the cumulant mixing conditions leading to suboptimal rates, and a stronger set of assumptions including the strong mixing conditions resulting in optimal rates. Section \ref{sec:simulation} contains the results of numerical experiments designed to probe the finite-sample performance of our methodology. Section \ref{sec:data analysis} illustrates the proposed methodology on the fair-weather atmospheric electricity time series. In Appendix \ref{appendix:practical_implementation_concerns} we comment on some implementation concerns, the formal proofs are included in Appendix \ref{sec:proofs}, and some additional results of the numerical experiments are presented in Appendix \ref{appendix:supplementary results}.

%%%%%%%%%%%%%%%%%%%%%%%%%%%%%%%%%%%%%%%%%%%%%%%%%%%%%%%%%%%%%%%%%%%%%%%%%%%%%%%%%%%%%%%%%%%%%%%%%%%%%
%%%%%%%%%%%%%%%%%%%%%%%%%%%%%%%%%%%%%%%%%%%%%%%%%%%%%%%%%%%%%%%%%%%%%%%%%%%%%%%%%%%%%%%%%%%%%%%%%%%%%

\section{Model and Estimation Methodology} \label{sec:model}
\subsection{Functional Time Series Framework}
Functional time series is a sequence of random function defined on the interval $[0,1]$ and is denoted as $\{X_t\}_{t\in\mathbb{Z}} = \{ X_t(x), x\in [0,1] \}_{t\in\mathbb{Z}}$. We assume that $X_t \in \mathcal{H} = L^2([0,1])$ and $\Ez{ \|X_t\|^2 }<\infty$. Moreover we assume that the realisations (paths) of $X_t$ are smooth functions (concrete smoothness assumptions will be introduced in Section \ref{sec:asymptotics}). This space-time process will be referred to as a functional time series. Assuming second-order stationarity in the time variable $t$, we may define the (common) mean function of $X_t(\cdot)$ by
$$ \Ez{ X_t(x) } = \mu(x), \qquad x\in[0,1],$$
and capture the second-order dynamics of the functional time series by its lag-$h$ autocovariance kernels,
$$ R_h(x,y) = \E \left\{ (X_h(x)-\mu(x)(X_0(y)-\mu(y)) \right\}, \qquad x,y\in[0,1], \quad h\in\mathbb{Z} .$$
Each kernel $R_h(\cdot,\cdot)$ introduces a corresponding operator $\mathscr{R}_h: L^2([0,1])\to L^2([0,1])$ defined by right integration
$$ (\mathscr{R}_h g)(x) = \int_0^1 R_h(x,y)g(y)\D{y}, \qquad g\in L^2([0,1]). $$
In addition to the stationarity, we assume weak dependence, in that the autocovariance kernels are summable in the supremum norm (denoted by $\|\cdot\|_\infty$) and the autocovariance operators summable in the nuclear norm (denoted by $\|\cdot\|_1$)
\begin{equation}
\label{eq:assumption_stationarity_summability_of_autocovariance}
\sum_{h\in\mathbb{Z}} \|R_h\|_\infty =  \sum_{h\in\mathbb{Z}} \sup_{x,y\in[0,1]} |R_h(x,y)| < \infty ,
\qquad
\sum_{h\in\mathbb{Z}} \|\mathscr{R}_h\|_1 < \infty.
\end{equation}
Under these conditions, \citet{panaretos2013fourier} showed that for each $\omega\in(-\pi,\pi)$, the following series converge in the supremum norm and the nuclear norm, respectively
\begin{equation}
\label{eq:spectral_density_kernel_operator}
f_\omega(\cdot,\cdot) = \frac{1}{2\pi} \sum_{h\in\mathbb{Z}} R_h(\cdot,\cdot) \exp(- \I\omega h),
\qquad
\mathscr{F}_\omega = \frac{1}{2\pi} \sum_{h\in\mathbb{Z}} \mathscr{R}_h \exp(- \I\omega h).
\end{equation}
The kernel $f_\omega(\cdot,\cdot)$ and the operator $\mathscr{F}_\omega$ are called the spectral density kernel at frequency $\omega$ and the spectral density operator at frequency $\omega$ respectively.
The lagged autocovariance kernels and operators can be recovered by the inversion formula \citep{panaretos2013fourier} that holds in the supremum and the nuclear norm, respectively:
\begin{equation}
\label{eq:inversion_formula_operators_kernels}
R_h(\cdot,\cdot) = \int_{-\pi}^\pi f_\omega(\cdot,\cdot) \exp(\I \omega h) \D\omega,
\qquad
\mathscr{R}_h = \int_{-\pi}^\pi \mathscr{F}_\omega \exp(\I \omega h) \D\omega.
\end{equation}
In particular, the spectral density operator $\mathscr{F}_\omega$ is a non-negative, self-adjoint trace-class operator for all $\omega$.

%Hence it admits the spectral representation
%$$\mathscr{F}_\omega = \sum_{m=1}^\infty \lambda_m^\omega \varphi_m^\omega \otimes \varphi_m^\omega$$
%where $\otimes$ is the tensor product in $L^2([0,1])$. The elements of the sequence $\lambda_1^\omega \geq \lambda_2^\omega \geq \dots \geq 0$ are called harmonic eigenvalues and the corresponding functions $\varphi_m^\omega$ are called harmonic eigenfunctions.

\subsection{Observation Scheme}
\label{subsec:observation scheme}

We consider a sparse observation scheme with additive independent measurement errors. Let $Y_{tj}$ be the $j$-th measurement on the $t$-th curve at spatial position $x_{tj}\in [0,1]$, where $j=1,\dots,N_t$ and $N_t$ is the number of measurements on the curve $X_t$ for $t=1,\dots,T$. The additive measurement errors are denoted by $\epsilon_{tj}$ and are assumed to be independent identically distributed realisations of a mean $0$ and variance $\sigma^2>0$ random variable for $j=1,\dots,N_t$ and $t=1,\dots,T$. Furthermore, the measurement errors are assumed to be independent of  $\{X_t\}_{t\in\mathbb{Z}}$ as well as the measurement locations $\{x_{tj}\}$.
The observation model can be then written as
\begin{equation}\label{eq:observation_scheme}
Y_{tj} = X_t(x_{tj}) + \epsilon_{tj}, \qquad j=1,\dots,N_t,\quad t=1,...,T.
\end{equation}

\noindent The spatial positions $x_{tj}$ as well as their number $N_t$ are considered random and concrete conditions for the asymptotic results are given in Section \ref{sec:asymptotics}.

\subsection{Nonparametric Estimation of the Model Dynamics}
\label{subsec:nonparam_estimation}

Given the sparsely observed data $\{Y_{tj}\}$ generated by the observation scheme \eqref{eq:observation_scheme}, we wish to estimate the mean function $\mu$ and the lag autocovariance kernels $R_h(\cdot,\cdot)$. Thanks to the formulae \eqref{eq:spectral_density_kernel_operator} and \eqref{eq:inversion_formula_operators_kernels}, the estimation of the lag autocovariance operators is equivalent to the estimation of the spectral density $f_\omega(\cdot,\cdot)$.

In the first step, we estimate the common mean function $\mu$ by a local linear smoother, see, for example, \citet{FanJianqing1996Lpma}. Let $K(\cdot)$ be a one-dimensional symmetric probability density function. Throughout this paper we work with the Epanechnikov kernel $K(v)=\frac{3}{4}(1-v^2)$ for $ v\in[-1,1],$ and $0$ otherwise, but any other usual smoothing kernel would be appropriate. Let $B_\mu>0$ be the bandwidth parameter. We define the estimator of $\mu(x)$ as $\hat\mu(x) = \hat a_0$ by minimizing the weighted sum of squares:
\begin{equation}\label{eq:local_LS_for_mu}
(\hat a_0, \hat a_1) = \argmin_{a_0,a_1} \sum_{t=1}^T \sum_{j=1}^{N_t} K\left(\frac{x_{tj}-x}{B_\mu}\right) \left\{ Y_{tj} - a_0-a_1(x_{tj}-x) \right\}^2 .
\end{equation}

Then, in a second step, we show how to estimate the second order characteristics of the functional time series, namely the lag-$0$ covariance and the lag-$h$ autocovariance kernels. Since the measurement errors $\epsilon_{tj}$ contribute only to the diagonal of the lag-0 autocovariance kernel, $\cov( Y_{t+h,j}, Y_{tk}) = R_h( x_{t+h,j}, x_{tk} ) + \sigma^2 1_{[h=0, j=k]}$ where $1_{[h=0, j=k]}=1$ if the condition in the subscript is satisfied and zero otherwise. Therefore we consider the ``raw'' covariances
\begin{equation}\label{eq:raw_autocovariances}
G_{h,t}( x_{t+h,j}, x_{tk} ) = (Y_{t+h,j} - \hat\mu(x_{t+h,j}))(Y_{tk} - \hat\mu(x_{tk}))
\end{equation}
where $h=0,\dots,T-1$, $t=1,\dots,T-h$, $j=1,\dots,N_{t+h}$, and $k=1,\dots,N_t$. We anticipate that $\Ez{ G_{h,t}( x_{t+h,j}, x_{tk} )} \approx R_h( x_{t+h,j}, x_{tk} ) + \sigma^2 1_{[h=0, j=k]}$. Hence, the diagonal of the raw lag-$0$ covariances must be removed when estimating the lag-$0$ covariance kernel.

Specifically, to estimate the lag-$0$ covariance kernel, we employ a local-linear surface-smoother on $[0,1]^2$ applied to the raw covariances $\{G_{0,t}( x_{tj}, x_{tk} ), t=1,\dots,T,j\neq k\}$. Precisely, we let $\hat R_0(x,y) = \hat b_0$ where $\hat b_0$ is obtained by minimizing the following weighted sum of squares:
\begin{equation}
\label{eq:local_LS_cov_lag_0}
(\hat b_0, \hat b_1, \hat b_2) = \argmin_{b_0,b_1,b_2} \sum_{t=1}^T \sum_{j\neq k}
K\left( \frac{x_{tj}-x }{B_R} \right)
K\left( \frac{x_{tk}-y }{B_R} \right)
\left\{ G_{0,t}( x_{tj}, x_{tk} ) - b_0 - b_1(x_{tj}-x) - b_2(x_{tk}-y) \right\}^2
\end{equation}
and $B_R>0$ is the bandwidth parameter.

We estimate the measurement error variance $\sigma^2$ using the approach of \citet{yao2005functional}. That is, we first estimate $V(x) = R_0(x,x) + \sigma^2$ by smoothing the variance on the diagonal. We assign $\hat{V}(x) = \hat{c}_0$ where:
\begin{equation} \label{eq:local_LS_diagonal}
(\hat{c}_0,\hat{c}_1)
= \argmin_{c_0, c_1} \sum_{t=1}^T \sum_{j=1}^{N_t} K\left( \frac{x_{tj}-x}{B_V} \right)
\left\{ Y_{tj} - c_0 - c_1(x_{tj}-x) \right\}^2
.
\end{equation}
Instead of using $\{\hat{R}_0(x,x):x\in[0,1]\}$ as the estimator of the diagonal of the lag-$0$ covariance kernel (without the ridge contamination), \citet{yao2003shrinkage,yao2005functional} opted for a local-quadratic smoother -- arguing that the covariance kernel is maximal along the diagonal, and so a local-quadratic smoother is expected to outperform a local linear smoother. This heuristic was also confirmed by our own simulations. Therefore, following \citet{yao2005functional}, we fit a local-quadratic smoother along the direction perpendicular to the diagonal. Concretely, the estimator is defined as $\bar{R}_0(x) = \bar{c}_0$ where $\bar{c}_0$ is the minimizer of the following weighted sum of squares:
\begin{multline}
\label{eq:local_LS_est_sigma2_minimisation_V}
(\bar{c}_0,\bar{c}_1,\bar{c}_2) = \argmin_{c_0,c_1,c_2}
\sum_{t=1}^T \sum_{j\neq k} 
K\left( \frac{x_{tj}-x }{B_R} \right)
K\left( \frac{x_{tk}-x }{B_R} \right)
\times\\\times
\left\{ G_{0,t}( x_{tj}, x_{tk} ) - c_0 - c_1(P(x_{tj},x_{tk})-x) - c_2(P(x_{tj},x_{tk})-x)^2 \right\}^2
\end{multline}
where $P(x_{tj},x_{tk})$ is the first coordinate (which is the same as the second one) of the projection of the point $(x_{tj},x_{tk})$ onto the diagonal of $[0,1]^2$.
The measurement error variance is then estimated by
\begin{equation} \label{eq:local_LS_est_sigma2}
\hat{\sigma}^2 = \int_0^1 \left( \hat{V}(x) - \bar{R}_0(x) \right) \D x.
\end{equation}
Since the estimator \eqref{eq:local_LS_est_sigma2} is based on smoothers, it is not guaranteed to be a positive number. This problem was already commented on by \citet{yao2005functional}. In the theoretical part of their paper, the negative estimate is replaced by zero and, in their code, it is replaced by a small positive number. The replacement by a positive number can be seen as a form of regularization.

%As an alternative solution could be to try using the diagonal $\hat{R}_0(x,x),x\in[0,1],$ of the local linear smoother \eqref{eq:local_LS_cov_lag_0} in the place of $\bar{R}_0(\cdot)$ in the formula \eqref{eq:local_LS_est_sigma2}. Because of the above cited arguments, the local linear smoother \eqref{eq:local_LS_cov_lag_0} is expected to underestimate the diagonal $\hat{R}_0(x,x),x\in[0,1],$ and, therefore, there is a higher chance for the integrated difference $\int_0^1 \left( \hat{V}(x) - \hat{R}_0(x,x) \right) \D x$ to be positive.

Next, we proceed with the estimation of the lag-$h$ autocovariance kernels for $h>0$.
We define the estimator $\hat R_h(x,y) = \hat b_0^{(h)}$ for $h=1,\dots,T-1$ by minimizing
\begin{multline}\label{eq:local_LS_cov_lag_nu}
(\hat b_0^{(h)}, \hat b_1^{(h)}, \hat b_2^{(h)}) = \argmin_{b_0^{(h)},b_1^{(h)},b_2^{(h)}}
\sum_{t=1}^{T-h}
\sum_{j=1}^{N_{t+h}} \sum_{k=1}^{N_t}
K\left( \frac{x_{t+h,j}-x }{B_R} \right)
K\left( \frac{x_{tk}-y }{B_R} \right)
\times\\\times
\left\{ G_{h,t}( x_{t+h,j}, x_{tk} ) - b_0^{(h)} - b_1^{(h)}(x_{t+h,j}-x) - b_2^{(h)}(x_{tk}-y) \right\}^2
\end{multline}
For $h<0$ we set $\hat R_h = \hat R_{-h}\transpose$.
Observe that we did not need to remove the diagonal as in \eqref{eq:local_LS_cov_lag_0}.
Denote the corresponding estimated covariance operators as $\hat{\mathscr{R}}_h$.

\subsection{Spectral Density Kernel Estimation}
\label{subsec:Spectral Density Kernels Estimation}

To estimate the spectral density kernels $f_\omega$ one has to resort to smoothing or a different sort of regularization at some point. \citet{panaretos2013fourier} performed kernel smoothing of the periodogram in the spectral domain whereas \citet{hormann2015dynamic} made use of Barlett's estimate.
Bartlett's estimate involves a weighted average of the lagged autocovariances, with a choice of weights that downweighs higher order lags. From the theoretical perspective, this approach is equivalent to kernel smoothing of the periodogram \citep[\S 6.2.3]{PriestleyMauriceB1981Saat}. In fact, the Bartlett's weights correspond to the Fourier coefficients of the smoothing kernel, assumed compactly supported.

In this paper, we opt for the \textit{Barlett's weights} (or \textit{triangular window}) defined as $W_h = (1-|h|/\QBartlett)$ for $|h|<\QBartlett$ and $0$ otherwise for the Barlett's span parameter $\QBartlett\in\mathbb{N}$ as it seems to be a popular choice \citep{hormann2010weakly,hormann2015dynamic}. It should be noted that other choices of weights are possible \citep{rice2017plug} and the so-called local quadratic windows (Parzen, Bartlett-Pristley, etc.) improve the asymptotic bias. See \citet[\S 7.5]{PriestleyMauriceB1981Saat} for the detailed discussion in one-dimensional case. The statement seems to be also true for functional time series \citep{vandelft2019anote}.

If the full functional observations were available, the spectral density would be estimated by the formula (cf. \citet{hormann2015dynamic})
\begin{equation}\label{eq:bartletts_estimate_full_observations}
\hat{\mathscr{F}}_\omega =\frac{1}{2\pi} \sum_{h=-\QBartlett}^{\QBartlett} W_h \hat{\mathscr{R}}_h e^{-\I h\omega}
\end{equation}
where $\hat{\mathscr{R}}_h$ are the standard empirical autocovariance operators.
We could use the formula \eqref{eq:bartletts_estimate_full_observations} and plug-in the smoothed autocovariance kernels obtained in Section \ref{subsec:nonparam_estimation} but instead, we opt to show how to directly construct a smoother-based estimator of the spectral density kernels.
Specifically, we estimate the spectral density kernel at frequency $\omega\in(-\pi,\pi)$ by the local-linear surface-smoother applied to the raw covariances multiplied by complex exponentials. The weights for the smoother are based both on the spatial distance from the raw covariances as well as the time lag.
Specifically, we estimate the spectral density kernel as
\begin{equation}
\label{eq:spectral_density_estimator_bartlett_smoother}
\hat{f}_\omega(x,y) =\frac{\QBartlett}{2\pi} \hat{d}_0 \in\mathbb{C}
\end{equation}
where $\hat{d}_0$ is obtained by minimizing the following weighted sum of squares
\begin{multline}
\label{eq:minimization-spectral-density}
(\hat{d}_0,\hat{d}_1,\hat{d}_2) = \argmin_{(d_0,d_1,d_2) \in \mathbb{C}^3}
\sum_{h=-\QBartlett}^{\QBartlett}
{\frac{1}{\mathcal{N}_h}}
\sum_{t=\max(1,1-h)}^{\min(T,T-h)}
\stackrel[j\neq k \text{ if } h=0]{}{\sum_{j=1}^{N_{t+h}} \sum_{k=1}^{N_t}}
\big|
G_{h,t}(x_{t+h,j}, x_{tk})
e^{-\I h\omega}
-\\- d_0-d_1(x_{t_{t+h,j}} - x) - d_2(x_{tk} -y)\big|^2
W_h
\frac{1}{B_R^2} K\left(\frac{x_{t_{t+h,j}} - x}{B_R}\right)K\left(\frac{x_{tk} -y}{B_R} \right)
\end{multline}
{where $\mathcal{N}_h = (T-|h|) (\bar{N})^2$ for $h\neq 0$, $\mathcal{N}_0 = T (\overline{N^2} - \bar{N})$, and where $\bar{N} = (1/n) \sum_{t=1}^T N_t$ and $\overline{N^2} = (1/n) \sum_{t=1}^T N_t^2$.
}

It turns out that the minimizer of this complex minimization problem can be expressed explicitly. Moreover, the minimizer depends only on a few quantities that are independent of $\omega$, and can be pre-calculated. The estimator can be thus constructed for a given $\omega$ by multiplying these quantities by complex exponentials and performing a handful of inexpensive arithmetic operations. Consequently, it is computationally feasible to evaluate the estimator \eqref{eq:spectral_density_estimator_bartlett_smoother} on a dense grid of frequencies. The explicit form is stated in Section \ref{subsec:proof_of_Thm_2}.

Denote the integral operator corresponding to $\hat{f}_\omega(\cdot,\cdot)$ as $\hat{\mathscr{F}}_\omega$. We can go back to the temporal domain by integrating the spectral density and reproduce the estimators of the
autocovariance kernels and operators by the formulae \eqref{eq:inversion_formula_operators_kernels}
\begin{equation}
\label{eq:tilde_R_tilde_mathscr_r}
\tilde R_h(\cdot,\cdot) = \int_{-\pi}^\pi \hat{f}_\omega(\cdot,\cdot)e^{\I h\omega}\D\omega,
\qquad
\tilde{\mathscr{R}}_h =\int_{-\pi}^\pi \hat{\mathcal{F}}_\omega e^{\I h\omega}\D\omega.
\end{equation}

The estimators of spectral density kernels $\hat{f}_\omega(\cdot,\cdot), \omega\in(-\pi,\pi),$ are achieved by kernel smoothing. Therefore, especially for smaller sample sizes, the operators $\hat{\mathscr{F}}_\omega, \omega\in(-\pi,\pi),$
might not be strictly non-negative, and may feature some tail negative eigenvalues of small modulus.
To ensure numerical stability of the method in the following section, it is recommended to truncate these negative eigenvalues of $\hat{\mathscr{F}}_\omega$ at each frequency $\omega\in(-\pi,\pi)$.

If dimensionality reduction is of interest, one can truncate the spectral density operators $\hat{\mathscr{F}}_\omega$ at each frequency $\omega\in(-\pi,\pi)$ to an appropriate rank. Such dimensionality reduction is based on the Cram{\'e}r-Karhunen-Lo{\`e}ve expansion and was proven optimal in preserving the functional time series dynamics by
\citet{panaretos2013cramer}, and independently by \citet{hormann2015dynamic}. Since dimension reduction is not necessary for our theory/methods in the next section, we do not pursue it further.

\subsection{Periodic Behaviour Identification}
\label{subsec:periodic behaviour identification}

As discussed at the beginning of Section \ref{subsec:Spectral Density Kernels Estimation}, the choice of Bartlett's span parameter $\QBartlett$ is related to the bandwidth for smoothing in the frequency domain. To achieve consistent spectral density estimation, the parameter $L$ needs to be kept quite small (cf. condition \ref{assumption:B.10} and Theorem \ref{thm:spectral_density}). However, for the purpose of exploratory data analysis, it is useful to explore the data for periodic behaviour in a similar way as a periodogram is used in the case of scalar time series.

When the periodicity examination is indeed of interest, we propose to evaluate the estimator \eqref{eq:spectral_density_estimator_bartlett_smoother} for a fairly large value of $L$. The selection of adequate value of $L$ is a question of computational power available because the computational time to evaluate \eqref{eq:spectral_density_estimator_bartlett_smoother} grows linearly in $L$. In the data analysis Section \ref{sec:data analysis} we work with $L=1000$ which is roughly half of the considered time series length.

Once the estimator \eqref{eq:spectral_density_estimator_bartlett_smoother} is evaluated for a given value of $L$ we propose to calculate the trace of the spectral density operator at frequency $\omega\in(0,\pi)$. Peaks in this plot indicate periodic behaviour of the functional time series. The existence of periodicity is not only a useful insight into the nature of the data but may us prompt into approaching the periodic behaviour in a different way, for example by modelling the periodicity in a deterministic way as we do it in the data analysis carried out in Section \ref{sec:data analysis}.

\subsection{Functional Data Recovery Framework and Confidence Bands}
\label{subsec:functional_data_recovery}

We now consider the problem of recovering the latent functional data $\{X_t(x):x\in[0,1]\}$ given the sparse noisy samples $\{Y_{tj}\}$, and provide corresponding confidence bands.

Consider the random element $\mathbb{X}_T = [X_1, \dots, X_T] \in \mathcal{H}^T$ composed of ``stacked'' functional data (formally, it is an element of the product Hilbert space $\mathcal{H}^T$). Note that
\begin{equation}
\label{eq:mathbb_X_mu}
\Ez{ \mathbb{X}_T } = \boldsymbolmu_T = [\mu,\dots,\mu] \in\mathcal{H}^T,
\end{equation}
\begin{equation}
\label{eq:mathbb_X_cov}
\var(\mathbb{X}_T) = \mathbbSigma_T =
\begin{bmatrix}
\mathscr{R}_0 & \mathscr{R}_1\transpose & \mathscr{R}_2\transpose & \dots & \mathscr{R}_{T-1}\transpose \\
\mathscr{R}_1 & \mathscr{R}_0 & \mathscr{R}_1\transpose & \dots & \mathscr{R}_{T-2}\transpose \\
\vdots & \vdots & \vdots & \ddots & \vdots \\
\mathscr{R}_{T-1} & \mathscr{R}_{T-2} & \mathscr{R}_{T-3} & \dots & \mathscr{R}_0 \\
\end{bmatrix}
\in L(\mathcal{H}^T).
\end{equation}
Now define the stacked observables as $\mathbb{Y}_T = (Y_{11},\dots,Y_{1N_1},\dots,Y_{t1},\dots,Y_{tN_t},\dots,Y_{TN_1},\dots,Y_{TN_T}) \in \mathbb{R}^{\mathcal{N}_1^T}$ where $\mathcal{N}_1^T = \sum_{t=1}^T N_t$ is the total number of observations up to time $T$.
By analogy to $\mathbb{Y}_T$, stack the measurement errors $\{\epsilon_{tj}\}$ and denote this vector  $\mathcal{E}_T\in\mathbb{R}^{\mathcal{N}_1^T}$. Note that $\var(\mathcal{E}_T) = \sigma^2 I_{\mathcal{N}_1^T}$.
Further define the evaluation operators
$H_t : \mathcal{H} \to \mathbb{R}^{N_t}, g\mapsto (g(x_{t1}),\dots,g(x_{tN_t})) $ for each $t=1,\dots,T$ and
the stacked censor operator $\mathbb{H}_T : \mathcal{H}^T \to \mathbb{R}^{\mathcal{N}_1^T}, [g_1,\dots,g_T]\mapsto [H_1 g_1, \dots, H_T g_T]$.
Finally define the projection operator $P_t : \mathcal{H}^T \to\mathcal{H}, [g_1,\dots,g_T] \mapsto g_t$ for $t=1,\dots,T$.

\noindent In this notation we can rewrite the observation scheme \eqref{eq:observation_scheme} as
$$ \mathbb{Y}_T = \mathbb{H}_T\mathbb{X}_T + \mathcal{E}_T.$$

%\subsubsection{Estimation of Functional Data}
%\label{subsubsec:Estimation of Functional Data}

The best linear unbiased predictor 
%\textcolor{red}{BETTER SAY PREDICTOR? ALSO IT IS THE BEST PREDICTOR AMONG ALL PREDICTORS IN THE GAUSSIAN CASE -- IT IS THE BEST LINEAR PREDICTOR WHEN IN THE NONGAUSSIAN CASE}
of $\mathbb{X}_T$ given $\mathbb{Y}_T$, which we denote by $\widehat{\mathbb{X}}_T(\mathbb{Y}_T)$, is given by the formula
%\E[\mathbb{X}_T|\mathbb{Y}_T=\mathbb{y}]
\begin{equation}\label{eq:BLUP_X_theoretical}
\widehat{\mathbb{X}}_T(\mathbb{Y}_T) = \boldsymbolmu_T + \mathbbSigma_T\mathbb{H}_T^* (\mathbb{H}_T\mathbbSigma_T\mathbb{H}_T^* + \sigma^2 I_{\mathcal{N}_1^T})^{-1}(\mathbb{Y}_T-\mathbb{H}_T\boldsymbolmu_T) \quad\in \mathcal{H}^T
\end{equation}
where $*$ denotes the adjoint operator. The term $\mathbb{H}_T\mathbbSigma_T\mathbb{H}_T^*$ is in fact a positive semi-definite matrix. Owing to the fact that $\sigma^2>0$, the matrix $\mathbb{H}_T\mathbbSigma_T\mathbb{H}_T^* + \sigma^2 I_{\mathcal{N}_1^T}$ is always invertible.

Now fix $s\in\{1,\dots,T\}$. The best linear unbiased predictor of the functional datum $X_s$, which we denote by $\widehat{X}_s(\mathbb{Y}_T)$, is given by
%\E[ X_s | \mathbb{Y}_T=\mathbb{y} ] =
\begin{equation}\label{eq:BLUP_Xs_theoretical}
\widehat{X}_s(\mathbb{Y}_T) =
P_s \widehat{\mathbb{X}}_T(\mathbb{Y}_T)
\quad\in\mathcal{H}.
\end{equation}
Hence the recovery of $X_s$ by the formula \eqref{eq:BLUP_Xs_theoretical} uses the observed data across all $t=1,\dots,T$, borrowing strength across all the observations.

In practice, however, we need to replace the unknown parameters involved in the construction of the predictor by their estimates. Define $\hat{\boldsymbolmu}_T$ and $\hat{\mathbbSigma}_T$ by substituting $\hat\mu$ and $\tilde{\mathscr{R}}_h$ for their theoretical counterparts in formulae \eqref{eq:mathbb_X_mu} and \eqref{eq:mathbb_X_cov} respectively.
Now replace $\boldsymbolmu_T$, $\mathbbSigma_T$, $\sigma^2$ by $\hat{\boldsymbolmu}_T$, $\hat{\mathbbSigma}_T$ and $\hat{\sigma}^2$, respectively, in formulae \eqref{eq:BLUP_X_theoretical} and \eqref{eq:BLUP_Xs_theoretical}. The resulting predictors are denoted by
\begin{equation}\label{eq:BLUP_X_empirical}
\tilde{\mathbb{X}}_T(\mathbb{Y}_T) = \hat{\boldsymbolmu}_T + \hat{\mathbbSigma}_T \mathbb{H}_T^* (\mathbb{H}_T\hat{\mathbbSigma}_T\mathbb{H}_T^* + \sigma^2 I_{\mathcal{N}_1^T})^{-1}(\mathbb{Y}_T-\mathbb{H}_T\hat{\boldsymbolmu}_T)
\end{equation}
and
\begin{equation}\label{eq:BLUP_Xs_empirical}
\tilde{X}_s(\mathbb{Y}_T) = P_s \tilde{\mathbb{X}}_T(\mathbb{Y}_T).
\end{equation}
%\textcolor{red}{SO FAR GAUSSIANITY HAS NOT BEEN NECESSARY --- DO I SUSPECT CORRECTLY THAT GAUSSIANITY ONLY COMES IN WHEN WE DO THE CONFIDENCE BANDS?}

%\subsubsection{Pointwise Confidence Bands for Functional Data}

In order to construct confidence bands for the unobservable paths, we work under the Gaussian assumption:
\begin{enumerate}[label=(A{\arabic*})]
\item \label{assumption:A.1}
The functional time series $\{X_t\}_t$ as well as the measurement errors $\{\epsilon_{tj}\}_{tj}$ are Gaussian processes.
\end{enumerate}
Thanks to the Gaussian assumption \ref{assumption:A.1}, the predictors of $\mathbb{X}_T$ and $X_s$ given by formulae \eqref{eq:BLUP_X_theoretical} and \eqref{eq:BLUP_Xs_theoretical} are in fact given by conditional expectations and are the best predictors among all predictors.
Furthermore, we can calculate the exact conditional distribution of $\mathbb{X}_T$ given $\mathbb{Y}_T$ by the formula
\begin{equation}\label{eq:conditional_distribution_X}
\mathbb{X}_T | \mathbb{Y}_T \sim N_{\mathcal{H}^T}( \boldsymbolmu_{\mathbb{X}_T|\mathbb{Y}_T} , \mathbbSigma_{\mathbb{X}_T|\mathbb{Y}_T} )
\end{equation}
where
\begin{align}
\label{eq:conditional_distribution_X_mu}
\boldsymbolmu_{\mathbb{X}_T|\mathbb{Y}_T} &=
\boldsymbolmu_T + \mathbbSigma_T\mathbb{H}_T^* (\mathbb{H}_T\mathbbSigma_T\mathbb{H}_T^* + \sigma^2 I_{\mathcal{N}_1^T})^{-1}(\mathbb{Y}_T-\mathbb{H}_T\boldsymbolmu_T)
,\\
\label{eq:conditional_distribution_X_Sigma}
\mathbbSigma_{\mathbb{X}_T|\mathbb{Y}_T} &=
\mathbbSigma_T-\mathbbSigma_T\mathbb{H}_T^* (\mathbb{H}_T\mathbbSigma_T\mathbb{H}_T^* + \sigma^2 I_{\mathcal{N}_1^T})^{-1} \mathbb{H}_T \mathbbSigma_T.
\end{align}
From \eqref{eq:conditional_distribution_X} we can access the conditional distribution of $X_s$ for fixed $s=1,\dots,T$, by writing
\begin{equation}\label{eq:conditional_distribution_Xs}
X_s | \mathbb{Y}_T \sim N_{\mathcal{H}^T}( \boldsymbolmu_{X_s|\mathbb{Y}_T} , \mathbbSigma_{X_s|\mathbb{Y}_T} )
\end{equation}
where 
\begin{equation}
\label{eq:conditional_distribution_Xs_mu_Sigma}
\boldsymbolmu_{X_s|\mathbb{Y}_T} =
P_s \boldsymbolmu_{\mathbb{X}_T|\mathbb{Y}_T},
\qquad
\mathbbSigma_{X_s|\mathbb{Y}_T} =
P_s \mathbbSigma_{\mathbb{X}_T|\mathbb{Y}_T} P_s^*.
\end{equation}
%\begin{align}
%\label{eq:conditional_distribution_Xs_mu}
%\boldsymbolmu_{X_s|\mathbb{Y}_T} &=
%P_s \boldsymbolmu_{\mathbb{X}_T|\mathbb{Y}_T}
%,\\
%\label{eq:conditional_distribution_Xs_Sigma}
%\mathbbSigma_{X_s|\mathbb{Y}_T} &=
%P_s \mathbbSigma_{\mathbb{X}_T|\mathbb{Y}_T} P_s^*
%.
%\end{align}

To construct a band for $X_s$ with pointwise coverge, we construct a confidence interval for $X_s(x)$ at each $x\in[0,1]$ --- as we will see, the endpoints of these intervals are continuous functions of $x$, and so automatically define a confidence band. In practice, one constructs bands for a dense collection of locations in $[0,1]$ and interpolates. Given the conditional distribution $X_s(x) | \mathbb{Y}_T  \sim N( \boldsymbolmu_{X_s|\mathbb{Y}_T}(x) , \mathbbSigma_{X_s|\mathbb{Y}_T}(x,x) )$, the $(1-\alpha)$-confidence interval for fixed $x\in [0,1]$ is constructed as
\begin{equation}\label{eq:confidence_band_Xs_population}
\boldsymbolmu_{X_s|\mathbb{Y}_T}(x) \pm \Phi^{-1}(1-\alpha/2) \sqrt{\mathbbSigma_{X_s|\mathbb{Y}_T}(x,x)}
\end{equation}
where $\Phi^{-1}(1-\alpha/2)$ is the $(1-\alpha/2)$-quantile of the standard normal distribution.

In practice, when we do not know the true dynamics of the functional time series, we have to use the estimates of $\mu(\cdot)$ and $R_h(\cdot,\cdot)$. We define $\hat{\boldsymbolmu}_{\mathbb{X}_T|\mathbb{Y}_T}, \hat{\mathbbSigma}_{\mathbb{X}_T|\mathbb{Y}_T}, \hat{\boldsymbolmu}_{X_s|\mathbb{Y}_T}$ and $\hat{\mathbbSigma}_{X_s|\mathbb{Y}_T}$ by replacing $\boldsymbolmu_T$ and $\mathbbSigma_T$ with $\hat{\boldsymbolmu}_T$ and $\hat{\mathbbSigma}_T$ in the formulae
\eqref{eq:conditional_distribution_X_mu},
\eqref{eq:conditional_distribution_X_Sigma},
\eqref{eq:conditional_distribution_Xs_mu_Sigma} respectively.
Therefore the asymptotic confidence interval for $X_s(x)$ is obtain by rewriting \eqref{eq:confidence_band_Xs_population} using the empirical counterparts
\begin{equation}\label{eq:confidence_band_Xs_empirical}
\hat{\boldsymbolmu}_{X_s|\mathbb{Y}_T}(x) \pm \Phi^{-1}(1-\alpha/2) \sqrt{\hat{\mathbbSigma}_{X_s|\mathbb{Y}_T}(x,x)}.
\end{equation}
%The pointwise confidence band for $X_s$ is then constructed by connecting the end points of the confidence intervals \eqref{eq:confidence_band_Xs_empirical} for $x\in[0,1]$.

%\subsubsection{Simultaneous Confidence Bands for Functional Data}
%\label{subsubsec:simultaneous bands}

For the construction of the simultaneous band we use the method introduced by \citet{degras2011simultaneous}. Fix $s=1,\dots,T$.  In the previous section we derived the conditional distribution of $X_s$ given $\mathbb{Y}_T$ in formula \eqref{eq:conditional_distribution_Xs}.
Define the conditional correlation kernel
\begin{equation}\label{eq:simul_bands_correlation_kernel}
\rho_{X_s|\mathbb{Y}_T} (x,y)
=
\begin{cases} 
      \frac{\mathbbSigma_{X_s|\mathbb{Y}_T} (x,y)}{\sqrt{ \mathbbSigma_{X_s|\mathbb{Y}_T} (x,x)\mathbbSigma_{X_s|\mathbb{Y}_T} (y,y) }},
      & \mathbbSigma_{X_s|\mathbb{Y}_T}(x,x)>0,\quad \mathbbSigma_{X_s|\mathbb{Y}_T} (y,y)>0, \\
      0 ,& \text{otherwise}.
\end{cases}
\end{equation}

Then, the collection of intervals 
\begin{equation}\label{eq:simul_confidence_band_population}
\left\{\boldsymbolmu_{X_s|\mathbb{Y}_T}(x) \pm z_{\alpha,\rho} \sqrt{\mathbbSigma_{\mathbb{X}_T|\mathbb{Y}_T} (x,x)}: x\in[0,1]\right\},
\end{equation}
forms a (continuous) confidence band with simultaneous coverage probability $(1-\alpha)$ over $x\in[0,1]$. Here $z_{\alpha,\rho}$ is the $(1-\alpha)$-quantile of the law of $\sup_{x\in[0,1]}|Z(x)|$ where $\{Z(x), x\in[0,1]\}$ is a zero mean Gaussian process with covariance kernel
$\rho_{\mathbb{X}_T|\mathbb{Y}_T}$. The definition of a quantile specifically requires that $P( \sup_{x\in[0,1]} |Z(x)| \leq z_{\alpha,\rho}) = 1-\alpha$.
\citet{degras2011simultaneous} explains how to calculate this quantile numerically.
% and shows that the computation simplifies when the functional data are represented in a finite basis which is a common practise in FDA applications \textcolor{red}{BUT IS THIS RELEVANT IN OUR CASE? IF NOT, THERE'S NO POINT OF MENTIONING IT}.

In practice, we replace the population level quantities in \eqref{eq:simul_confidence_band_population} by their estimated counterparts and define the asymptotic simultaneous confidence band as
\begin{equation}\label{eq:simul_confidence_band_empiric}
\left\{\hat{\boldsymbolmu}_{X_s|\mathbb{Y}_T}(x) \pm z_{\alpha,\hat{\rho}} \sqrt{\hat{\mathbbSigma}_{\mathbb{X}_T|\mathbb{Y}_T} (x,x)}: x\in[0,1]\right\}, 
\end{equation}
where $\hat{\boldsymbolmu}_{X_s|\mathbb{Y}_T}(x)$ and $\hat{\mathbbSigma}_{\mathbb{X}_T|\mathbb{Y}_T} (x,x)$ are as above and the quantile $z_{\alpha,\hat{\rho}}$ is calculated for the correlation structure $\hat{\rho}_{X_s|\mathbb{Y}_T} $ defined as the empirical counterpart to \eqref{eq:simul_bands_correlation_kernel}.

Note that $\Phi^{-1}(1-\alpha/2)<z_{\alpha,\rho}$ for any correlation kernel $\rho$  \citep{degras2011simultaneous}. Therefore, as expected, the pointwise confidence bands are enveloped by the simultaneous band. Once again, in practice, one evaluates the band limits defining \eqref{eq:simul_confidence_band_empiric} on a dense grid of $[0,1]$ and interpolates.

The functional recovery framework proposed in this section can be easily extended into forecasting, i.e. prediction of functional curves beyond the time horizon $T$. For the details, see Section \ref{subsec:forecasting}.

%%%%%%%%%%%%%%%%%%%%%%%%%%%%%%%%%%%%%%%%%%%%%%%%%%%%%%%%%%%%%%%%%%%%%%%%%%%%%%%%%%%%%%%%%%%%%%%%%%%%%
%%%%%%%%%%%%%%%%%%%%%%%%%%%%%%%%%%%%%%%%%%%%%%%%%%%%%%%%%%%%%%%%%%%%%%%%%%%%%%%%%%%%%%%%%%%%%%%%%%%%%

\section{Asymptotic Results} \label{sec:asymptotics}

\subsection{On the Choice of Mixing Conditions}

In Sections 
\ref{subsec:asymptotic_results_cumulant} and
\ref{subsec:asymptotic_results_strong}
we develop asymptotic theory for our methodology under two different sets of assumptions.

Firstly, in Section \ref{subsec:asymptotic_results_cumulant} we prove the asymptotic behaviour of the estimators under Brillinger-type cumulant mixing conditions. The corresponding Theorems \ref{thm:mean_and_autocov_function} and \ref{thm:spectral_density} are in a sense canonical, in that their proofs rely on generalisations of the techniques by \citet{yao2005functional}. Nevertheless, the yielded convergence rates for one dimensional smoothing and surface smoothing are
$\Op( 1/(\sqrt{T}B_\mu) )$ and $\Op( 1/(\sqrt{T}B_R^2) )$, respectively, which are not optimal.

The optimal rates for one dimensional smoothing and surface smoothing are known to be $\Op(\sqrt{\log T/(T B_\mu)})$ and $\Op(\sqrt{\log T/(T B_R^2)})$ respectively. Recovering such rates using local-regression methods for time-series data relies heavily on the employed measure of weak dependence, namely strong mixing conditions, \citep{hansen2008uniform,liebscher1996strong,masry1996multivariate}, \citep[Thm 6.5]{fan2008nonlinear},
geometric strong mixing conditions \citep[Thm. 2.2 and Cor. 2.2]{bosq2012nonparametric}, and
$\rho$-mixing conditions \citep{peligrad1992properties}. In Section \ref{subsec:asymptotic_results_strong} and Theorems \ref{thm:strong_mixing_mu_R}, \ref{thm:strong_mixing_F} we make use of techniques developed by \citet{hansen2008uniform} to obtain the optimal rates under strong mixing. 

Since these two sets of rates rest on qualitatively different conditions, we have chosen to include both results into the article.

\subsection{Asymptotic Results under Cumulant Mixing Conditions}
\label{subsec:asymptotic_results_cumulant}

In order to establish the consistency and the convergence rate of the estimators introduced in Section \ref{sec:model}, we will make use of the following further assumptions on the model \eqref{eq:observation_scheme}:

\begin{enumerate}[label=(B{\arabic*})]
\item \label{assumption:B.1}
The number of measurements $N_t$ in time $t$ are independent random variables with law $N_t \sim N$ where $N\geq 0$, $\Ez{ N } < \infty$ and $\Prob(N>1)>0$.
\item \label{assumption:B.2}
The measurement locations $x_{tj}, j=1,\dots,N_t, t=1,\dots,T$ are independent random variables generated from the density $g(\cdot)$ and are independent of the number of measurements $(N_t)_{t=1,\dots,T}$. The density $g(\cdot)$ is assumed to be twice continuously differentiable and strictly positive on $[0,1]$.
\end{enumerate}
We allow the event $\{N_t=0\}$ to potentially have positive probability. This corresponds to the situation where no measurements are available at time $t$, for example when we additionally have missing data at random.
We also need to impose smoothness conditions on the unknown functional parameters
\begin{enumerate}[resume, label=(B{\arabic*})]
\item \label{assumption:B.3}
The common mean function, $\mu(\cdot)$, is twice continuously differentiable on $[0,1]$.
\item \label{assumption:B.4}
The autocovariance kernels, $R_h(\cdot,\cdot)$, are twice continuously differentiable on $[0,1]^2$ for each $h\in\mathbb{Z}$. Moreover,
$$\sup_{x,y\in[0,1]} \left| \frac{\partial^2}{\partial y^{\alpha_1} \partial x^{\alpha_2}} R_h(y,x) \right|$$
is uniformly bounded in $h$ for all combinations of $\alpha_1,\alpha_2 \in\mathbb{N}_0$ where $\alpha_1+\alpha_2=2$.
\end{enumerate}
To prove the consistency of autocovariance kernels estimators $\hat{R}_h(\cdot,\cdot)$ we need to further assume some mixing conditions in the time domain. The smoothing estimators are essentially moment-based, therefore it is natural to consider cumulant-type summability conditions. For the introduction to the cumulants of real random variables see \citet{rosenblatt2012stationary}
and for the definitions and properties of the cumulant kernels and cumulant operators see \citet{panaretos2013fourier}.
\begin{enumerate}[resume, label=(B{\arabic*})]
\item \label{assumption:B.5}
Denote the 4-th order cumulant kernel of $\{X_t\}$ as $\cum(X_{t_1}, X_{t_2}, X_{t_3}, X_{t_4})(\cdot,\cdot,\cdot,\cdot)$. Assume the summability in the supremum norm
$$
\sum_{h_1,h_2,h_3 = -\infty}^\infty \sup_{x_1,x_2,x_3,x_4 \in [0,1]} \left|
\cum(X_{h_1}, X_{h_2}, X_{h_3}, X_0)(x_1,x_2,x_3,x_4)
\right| < \infty.
$$
\end{enumerate}
We will also need to strengthen the summability assumption \eqref{eq:assumption_stationarity_summability_of_autocovariance}.
\begin{enumerate}[resume, label=(B{\arabic*})]
\item \label{assumption:B.6}
Assume
$$ \sum_{h=-\infty}^\infty |h| \sup_{x,y\in[0,1]} \left| R_h(x,y) \right| <\infty.$$
\end{enumerate}
The last two conditions correspond to conditions \textbf{C$'$(1,2)} and \textbf{C$'$(0,4)} in  \citet{panaretos2013fourier}, respectively. 
Finally, we impose the following assumptions on the decay rate of the bandwidth parameters and the growing rate of the Bartlett's span parameter $\QBartlett$
\begin{enumerate}[resume, label=(B{\arabic*})]
\item \label{assumption:B.7}
$B_\mu \to 0$, $T B_\mu^4 \to \infty$,
\item \label{assumption:B.8}
$B_R \to 0$, $T B_R^6\to \infty$,
\item \label{assumption:B.9}
$B_V \to 0$, $T B_V^4 \to \infty$,
\item \label{assumption:B.10}
$\QBartlett \to \infty$, $\QBartlett = o(\sqrt{T} B_R^2), L = o(B_R^{-2})$.
\end{enumerate}

%\begin{enumerate}[resume, label=(B{\arabic*})]
%\item greg
%\item grerge
%\end{enumerate}

\noindent We may now state our asymptotic results on uniform consistency and convergence rates:

\begin{theorem}\label{thm:mean_and_autocov_function}
Under the assumptions \ref{assumption:B.1} --- \ref{assumption:B.3} and \ref{assumption:B.7}:
\begin{equation}\label{eq:thm:mean_function:mu}
\sup_{x\in[0,1]} |\hat\mu(x) - \mu(x)| = \Op\left( \frac{1}{\sqrt{T} B_\mu} + B_\mu^2 \right).
\end{equation}
Under the assumptions \ref{assumption:B.1} --- \ref{assumption:B.5} and \ref{assumption:B.7} --- \ref{assumption:B.9}, for for fixed lag $h\in\mathbb{Z}$:
\begin{equation}\label{eq:thm:autocov_function:R}
\sup_{x,y\in[0,1]} | \hat R_h(x,y) - R_h(x,y) | = \Op\left( \frac{1}{\sqrt{T} B^2_R} +B^2_R\right),
\end{equation}
\begin{equation}\label{eq:thm:autocov_function:sigma}
\hat{\sigma}^2 = \sigma^2 + \Op\left\{ \frac{1}{\sqrt{T}} \left( \frac{1}{B_V} + \frac{1}{B_R^2}\right)+B^2_\mu +B^2_R \right\}.
\end{equation}
\end{theorem}

\begin{theorem}\label{thm:spectral_density}
Under the assumptions \ref{assumption:B.1} --- \ref{assumption:B.5} and \ref{assumption:B.7} --- \ref{assumption:B.10}, the spectral density is estimated consistently:
\begin{equation}\label{eq:thm_spectral_density_rate1}
\sup_{\omega\in[-\pi,\pi]} \sup_{x,y\in[0,1]} \left| \hat f_\omega(x,y)-f_\omega(x,y) \right| = \op(1).
\end{equation}
If we further assume condition \ref{assumption:B.6}, we can additionally obtain the convergence rate:
$$ \sup_{\omega\in[-\pi,\pi]} \sup_{x,y\in[0,1]} \left| \hat f_\omega(x,y)-f_\omega(x,y) \right| = \Op\left(\QBartlett \frac{1}{\sqrt{T}}\frac{1}{B_R^2} + \QBartlett B^2_R\right) .$$
\end{theorem}

As a consequence of Theorem \ref{thm:spectral_density} we obtain the consistency and the convergence rate of the entire space-time covariance structure \eqref{eq:tilde_R_tilde_mathscr_r}, i.e. rates uniform in both time index and spatial argument:
\begin{corollary}\label{corollary:sup_consistency}
Under the assumptions \ref{assumption:B.1} --- \ref{assumption:B.5} and \ref{assumption:B.7} --- \ref{assumption:B.10}:
\begin{equation}\label{eq:corollary:sup_consistency:statement}
\sup_{h\in\mathbb{Z}} \sup_{x,y\in[0,1]}
| \tilde R_h(x,y) - R_h(x,y) | = \op(1)
\end{equation}
and assuming further \ref{assumption:B.6}:
\begin{equation}\label{eq:corollary:sup_consistency:statement2}
\sup_{h\in\mathbb{Z}} \sup_{x,y\in[0,1]}
| \tilde R_h(x,y) - R_h(x,y) | = \Op\left(\QBartlett \frac{1}{\sqrt{T}}\frac{1}{B_R^2} + \QBartlett B^2_R\right).
\end{equation}
\end{corollary}

\subsection{Asymptotic Results Under Strong Mixing Conditions}
\label{subsec:asymptotic_results_strong}

We begin by listing the assumptions leading to the optimal convergence rates. Besides imposing the key assumption of the strong mixing we need to strengthen some of the other assumptions as well. We require some additional regularity conditions on the smoothing kernel $K(\cdot)$ which until now was assumed only to be a bounded probability density function. The condition is formulated for a generic $k(\cdot)$ multivariate kernel because we will require more than just the smoothing kernel $K(\cdot)$ to satisfy this condition.
\begin{enumerate}[label=(C{\arabic*})]
\item \label{assumption:C}
The function $k: \mathbb{R}^d \to \mathbb{R}$ is bounded and integrable
$$ |k(u)| \leq \bar{k} < \infty,
\qquad \int_{\mathbb{R}^d} \left| k(u) \right| \D u <\infty,
 $$
and for some $\Lambda_1 < \infty$ and $L<\infty$, either $k(u)=0$ for $|u|>\tilde{L}$ and 
$$ \left| k(u) - k(u') \right| \leq \Lambda_1 \|u-u'\|, \qquad u,u'\in\mathbb{R},$$
or $k(\cdot)$ is differentiable, $|(\partial/\partial) u k(u)| \leq \Lambda_1$, and for some $\nu>1$, $|(\partial/\partial) u k(u)| \leq \Lambda_1 \|u\|^{-\nu}$ for $\|u\|>\tilde{L}$.
\end{enumerate}
The following conditions impose more conditions on the functional time series model.
\begin{enumerate}[label=(D{\arabic*})]
\item \label{assumption:D1}
The functional time series $\{X_t\}_{t\in\mathbb{Z}}$ is strictly stationary and strong mixing with mixing coefficients $\alpha_m$ that satisfy
$$ \alpha(m) \leq A m^{-\beta}, $$
for $A<\infty$ and for some $s>2$
$$ \E |X_t(x)|^s \leq B_1 < \infty, \qquad x\in[0,1]. $$
\item \label{assumption:D2}
The number of measurement locations $N_t$ in time $t$ are independent random variables with law $N_t\sim N$ where $N\in\{0,1,\dots,N^{max}\}$ for some $N^{max}\in\mathbb{N}$ and such that $\Prob(N>1)>0$.
\item \label{assumption:D3}
The measurement errors $\{\epsilon_{tj}\}$ are independent identically-distributed zero-mean random variables satisfying
$$ \E |\epsilon_{tj}|^s < \infty. $$
Moreover, $\{\epsilon_{tj}\}$ are independent of the functional time series $\{X_t(\cdot)\}$.
\item \label{assumption:D4}
The marginal density of the observation location $g(\cdot)$ satisfies
$$ 0 < B_2 \leq \inf_{x\in[0,1]} g(x) \leq \sup_{x\in[0,1]} g(x) \leq B_3 < \infty. $$
\end{enumerate}

\noindent For the estimation of the mean function $\mu(\cdot)$ the following assumptions are required:
\begin{enumerate}[resume, label=(D{\arabic*})]
\item \label{assumption:D5}
The functional time series $\{X_t(\cdot)\}$ satisfies
$$ \sup_{h\in\mathbb{Z}} \sup_{x,y\in[0,1]} \E |X_h(x)X_0(y)| < \infty. $$
\item \label{assumption:D6}
The smoothing kernel $K(\cdot)$ satisfies $\int |u|^4 K(u)\D u <\infty$ and the functions $u\mapsto K(u), u\mapsto u K(u), u\mapsto u^2 K(u)$ satisfy the assumption \ref{assumption:C}.
\item \label{assumption:D7}
The coefficients $\beta$ and $s$ from the Assumption \ref{assumption:D1} satisfy
$$ \beta > \frac{2s-1}{s-2}. $$
\item \label{assumption:D8}
The bandwidth parameter $B_\mu$ satisfies
$$ \frac{\log T}{T^{\theta_\mu} B_\mu} = o(1), \qquad T\to\infty, $$
with
$$ \theta_\mu = \frac{\beta-2-(1+\beta)/(s-1)}{\beta+1-(1+\beta)/(s-1)}.$$
\item \label{assumption:D9}
The functions $ g(\cdot) $ and $g(\cdot)\mu(\cdot)$ are twice continuously differentiable on $[0,1]$.
\end{enumerate}

\noindent The rates for the lag-$h$ autocovariance kernel estomator(s) will require the following set of assumptions:
\begin{enumerate}[resume, label=(D{\arabic*})]
\item \label{assumption:D10}
The functional time series $\{X_t(\cdot)\}$ satisfies
$$ \sup_{h\in\mathbb{Z}} \sup_{x,y\in[0,1]} \E \left| X_h(x)X_0(y) \right|^2 < \infty $$
and
$$ \sup_{h,h'\in\mathbb{Z}}
\sup_{x,y,x',y' \in[0,1]}
\E\left| X_h(x)X_0(y)X_{h'}(x')X_{h+h'}(y') \right| < \infty
.$$
\item \label{assumption:D11}
The smoothing kernel $K(\cdot)$ satisfies $\iint |uv|^4 K(u)K(v) \D u\D v<\infty$ and the functions
$(u,v) \mapsto u^p v^q K(u)K(v)$ satisfy the assumption \ref{assumption:C} for $p,q\in\mathbb{N}_0, 0\leq p+q\leq 2$.
\item \label{assumption:D12}
The bandwidth parameter $B_R^2$ satisfies
$$ \frac{\log T}{T^{\theta_R} B_R} = o(1),\qquad T\to\infty, $$
with
$$ \theta_R = \frac{\beta-3 - (1+\beta)/(s-1)}{\beta+1 - (1+\beta)/(s-1)}.$$
\item \label{assumption:D13}
The functions $g(x)g(y)$ and $g(x)g(y)R_h(x,y)$ are twice continuously differentiable and
$$ \sup_{x,y\in[0,1]} \left| \frac{\partial^2}{\partial x^{\alpha_1} \partial y^{\alpha_2}}  R_h(x,y) \right| $$
is uniformly bounded in $h$ for all combinations of $\alpha_1,\alpha_2 \in\mathbb{N}_0$ where $\alpha_1+\alpha_2=2$.
\end{enumerate}

\noindent The following conditions will be required for the rates concerning spectral density estimation.
\begin{enumerate}[resume, label=(D{\arabic*})]
\item \label{assumption:D14}
Assume the summability in the supremum norm of the 4-th order cumulant kernel of $\{X_t\}$, 
$$
\sum_{h_1,h_2,h_3 = -\infty}^\infty \sup_{x_1,x_2,x_3,x_4 \in [0,1]} \left|
\cum(X_{h_1}, X_{h_2}, X_{h_3}, X_0)(x_1,x_2,x_3,x_4)
\right| < \infty.
$$
\item \label{assumption:D15}
The Bartlett span parameter $\QBartlett$ satisfies
$$ \QBartlett = o\left( \left(\sqrt{\frac{\log T}{T B_R^2}}\right) ^{-\frac{s-2}{s-1}} \right) $$
\item \label{assumption:D16}
The bandwidth parameter $B_R^2$ satisfies
$$ \frac{\log T}{T^{\theta_F} B_R} = o(1),\qquad T\to\infty, $$
with
$$ \theta_F = \frac{\beta(s-2)-4s+4}{\beta(s-2)}$$
and
$$ \QBartlett B_R^2 = o(1). $$
\end{enumerate}

We can now state the main consistency and convergence results under the strong mixing conditions.
\begin{theorem}\label{thm:strong_mixing_mu_R}
Under the assumptions \ref{assumption:D1} --- \ref{assumption:D9},
$$ \sup_{x\in[0,1]} \left| \hat\mu(x) - \mu(x) \right| = \Op\left( \sqrt{\frac{\log T}{T B_\mu}}  + B_\mu^2\right).$$
For fixed $h\in\mathbb{Z}$, under the assumptions \ref{assumption:D1} --- \ref{assumption:D13},
$$ \sup_{x,y\in[0,1]} \left| \hat{R}_h(x,y) - R_h(x,y) \right| = \Op\left( \sqrt{\frac{\log T}{T B_R^2}}  + B_R^2 \right).$$
\end{theorem}

\begin{theorem}\label{thm:strong_mixing_F}
Under the assumption \ref{assumption:D1} --- \ref{assumption:D11} and \ref{assumption:D13} --- \ref{assumption:D16},
\begin{equation}\label{eq:thm:strong_mixing_F_rate1}
\sup_{\omega\in[-\pi,\pi]} \sup_{x,y\in[0,1]}
\left| \hat{f}_\omega(x,y) - f_\omega(x,y) \right| =
\op(1)
\end{equation}
and assuming further \ref{assumption:B.6},
\begin{equation}\label{eq:thm:strong_mixing_F_rate2}
\sup_{\omega\in[-\pi,\pi]} \sup_{x,y\in[0,1]}
\left| \hat{f}_\omega(x,y) - f_\omega(x,y) \right| =
\Op\left( \QBartlett \sqrt{\frac{\log T}{T B_R^2}}  + \QBartlett B_R^2 \right).
\end{equation}
\end{theorem}

\subsection{Functional Data Recovery and Confidence Bands}
\label{subsec:asymptotics_functional_data_recovery}

In this section we turn our attention to developing asymptotic theory for the recovered functional data and the associated confidence bands, in particular, the asymptotic behaviour of the plug-in estimator \eqref{eq:BLUP_Xs_empirical} vis-\`a-vis its theoretical counterpart \eqref{eq:BLUP_Xs_theoretical}.

First of all, we need to clarify what asymptotic result we can hope to accomplish.
Before venturing into functional time series, let us comment on the asymptotic results for independent identically distributed functional data \citep{yao2005functional}.
As the number of sparsely observed functional data grows to infinity, one can consistently estimate the second-order structure of the stochastic process (which in this case consists of the zero-lag autocovariance, due to independence). This is then used in the plug-in prediction of a given functional datum, say $X_s(\cdot)$, given the sparse measurements on this datum. In the limit, this prediction is as good as if we knew the true lag zero covariance of the stochastic process   \citep[Theorem 3]{yao2005functional}. Because the predictor uses the estimate of the lag zero covariance based on all the observed data, \citet{yao2005functional} call this trait as \textit{borrowing strength} from the entire sample.

In the time series setting of the current paper, one can expand the concept of borrowing strength from the entire sample. As the number of sparsely observed functional data (i.e. the time horizon $T$) grows to infinity, one can not only estimate the dynamics of the functional time series consistently (Theorem \ref{thm:spectral_density} and Corollary \ref{corollary:sup_consistency}), but also further exploit the fact that neighbouring data are correlated to further improve the recovery. Because of the weak dependence, the influence of the observations decreases as we part away from the time $s$. Therefore we fix a span of times $1,\dots,S$ where $s<S\in\mathbb{N}$ and we will be interested in the prediction of $X_s$ given the data in this span. To be precise, we are going to prove that the prediction of $X_s$ from the data in the local span and based on the estimated dynamics from complete data is, in the limit, as good as the prediction based on the true (unknown) dynamics. Therefore, in our case, we are \textit{borrowing strength} across the sample in a twofold sense -- firstly for the estimation of the functional time series dynamics, and then for prediction of the functional datum $X_s$.

The span $S$ can in principle be chosen to be as large as one wishes, but is held fixed with respect to $T$. This is justified by the weak dependence assumption. In practice,  one must also entertain numerical considerations and not choose $S$ to be exceedingly large, since the evaluation of the predictors \eqref{eq:BLUP_Xs_theoretical} and \eqref{eq:BLUP_Xs_empirical} based on longer spans requires the inversion of a big matrix.

We formulate Theorems \ref{thm:correctness_of_kriging} and \ref{thm:correctness_of_bands} under the cumulant mixing conditions required for Theorems \ref{thm:mean_and_autocov_function} and \ref{thm:spectral_density}. Nevertheless, the conclusions also hold also under the strong mixing condition regime of Theorems \ref{thm:strong_mixing_mu_R} and \ref{thm:strong_mixing_F} since, as is apparent from the proofs, the only requirement coming into play is the consistency of the spectral density operator estimators in the sense of \eqref{eq:thm_spectral_density_rate1} or \eqref{eq:thm:strong_mixing_F_rate1}.

\begin{theorem}\label{thm:correctness_of_kriging}
Under the assumptions \ref{assumption:B.1} --- \ref{assumption:B.5} and \ref{assumption:B.7} --- \ref{assumption:B.10},
for fixed $s\in\mathbb{N}, s<S$, 
$$
\sup_{x\in[0,1]} \left|  \tilde{X}_s(\mathbb{Y}_S)(x)  - \widehat{X}_s(\mathbb{Y}_S)(x) \right| = \op(1).$$
\end{theorem}

In the following theorem we verify the asymptotic coverage probability of the pointwise and simultaneous confidence bands \eqref{eq:confidence_band_Xs_empirical} and \eqref{eq:simul_confidence_band_empiric} under the Gaussian assumption \ref{assumption:A.1}.

\begin{theorem}\label{thm:correctness_of_bands}
Under the assumptions \ref{assumption:A.1}, \ref{assumption:B.1} --- \ref{assumption:B.5} and \ref{assumption:B.7} --- \ref{assumption:B.10}, for fixed $s\in\mathbb{N}, s\leq S$:
\begin{itemize}
\item Asymptotic coverage of the pointwise confidence band for fixed $x\in[0,1]$: $$ \lim_{T\to\infty} \Prob\left\{ \left| \tilde{X}_s(\mathbb{Y}_S)(x) - X_s(x) \right| \leq \Phi^{-1}\left(1-\alpha/2\right) \sqrt{\hat{\mathbbSigma}_{\mathbb{X}_T|\mathbb{Y}_T} (x,x)}  \right\} = 1-\alpha.
$$
\item
Asymptotic coverage of the simultaneous confidence band:
$$ \lim_{T\to\infty} \Prob\left\{ \forall{x\in[0,1]} : \left| \tilde{X}_s(\mathbb{Y}_S)(x) - X_s(x) \right| \leq z_{\alpha,\hat{\rho}} \sqrt{\hat{\mathbbSigma}_{\mathbb{X}_T|\mathbb{Y}_T} (x,x)}  \right\} = 1-\alpha . $$
\end{itemize}
\end{theorem}

%%%%%%%%%%%%%%%%%%%%%%%%%%%%%%%%%%%%%%%%%%%%%%%%%%%%%%%%%%%%%%%%%%%%%%%%%%%%%%%%%%%%%%%%%%%%%%%%%%%%%
%%%%%%%%%%%%%%%%%%%%%%%%%%%%%%%%%%%%%%%%%%%%%%%%%%%%%%%%%%%%%%%%%%%%%%%%%%%%%%%%%%%%%%%%%%%%%%%%%%%%%

\section{Numerical Experiments}
\label{sec:simulation}

\subsection{Simulation Setting}
\label{subsec:simulation_setting}

In this section, we present a simulation study in order to prove the finite-sample performance of our methodology. To this aim, we simulate realisations of functional linear processes, namely functional moving average processes and functional autoregressive processes. These provide a good framework to investigate our methods since their spectral density operators can be explicitly calculated in closed form. Specifically, we consider:

\begin{itemize}
\item \textbf{Functional moving average process}

The (Gaussian) functional moving average process of order $q$ is given by the formula \citep{bosq2012linear}
\begin{equation}
\label{eq:simulations_FMA}
X_t = \mu + E_t + \mathcal{B}_1 E_{t-1} + \mathcal{B}_2 E_{t-2} + \dots + \mathcal{B}_q E_{t-q}
\end{equation}
where  $\mu\in\mathcal{H}$ is the mean function, $\mathcal{B}_j, j=1,\dots,q$ are bounded linear operators in $\mathcal{H}$, and $\{E_t\}$ is zero-mean Gaussian noise with a trace-class covariance operator $\mathcal{S}$.
The functional moving average process is a stationary linear process \citep{bosq2012linear} and clearly satisfies the assumption \eqref{eq:assumption_stationarity_summability_of_autocovariance} in the nuclear norm and thus admits the spectral density in the operator sense.
Though the calculation of the spectral density of the functional moving average process is straightforward, we are not aware of it having been considered before in its functional form elsewhere.
\begin{proposition}
\label{prop:spec_density_FMA}
The functional moving average process defined above admits the spectral density
\begin{equation}
\label{eq:spectral_density_FMA}
\mathscr{F}_\omega = \frac{1}{2\pi}
\left( I + \mathcal{B}_1 e^{-\I\omega} + \dots + \mathcal{B}_q e^{-\I\omega q}\right)
\mathcal{S}
\left( I + \mathcal{B}_1^* e^{\I\omega} + \dots + \mathcal{B}_q^* e^{\I\omega q}\right),
\qquad \omega\in(-\pi,\pi),
\end{equation}
in the operator sense \eqref{eq:spectral_density_kernel_operator}. Moreover, if the kernels corresponding to the operators $\mathcal{B}_1,\dots,\mathcal{B}_q$ are smooth, the spectral density exists also in the kernel sense \eqref{eq:spectral_density_kernel_operator} and the process satisfies the assumptions \ref{assumption:B.4}, \ref{assumption:B.5}, \ref{assumption:B.6}. If the mean function $\mu(\cdot)$ is smooth, the process satisfies also \ref{assumption:B.3}.
\end{proposition}

We set again the mean function as $\mu(x) = 4 \sin( 1.5 \pi x )$. The covariance kernel $S(x,y)$ of the driving noise is set to be $S(x,y) = 1.4  \sin(2\pi x)\sin(2\pi y) + 0.6 \cos(2\pi x)\cos(2\pi y)$. Next we define $\mathcal{B}_1,\dots,\mathcal{B}_8$ as integral operators with kernels $B_1(x,y) = B_5(x,y) = 5\exp( -(x^2+y^2) )$, $B_2(x,y) = B_6(x,y) = 5\exp( -((1-x)^2+y^2) )$, $B_3(x,y) =B_7(x,y) = 5\exp( -(x^2+(1-y)^2) )$, and $B_4(x,y) = B_8(x,y) = 5\exp( -((1-x)^2+(1-y)^2) )$ respectively. We denote these functional moving average processes as $\mathbf{FMA(q)}$ for $q=2,4,8$.

\item \textbf{Functional autoregressive process}

The (Gaussian) functional autoregressive process of order 1, well reviewed in \citet{bosq2012linear}, is defined by the iteration
\begin{equation}
\label{eq:FAR_equation}
(X_{t+1} - \mu) = \mathcal{A} (X_t - \mu) + E_t
\end{equation} where $\{X_t\}$ is a functional time series in the Hilbert space $\mathcal{H} = L^2([0,1])$, $\mu\in\mathcal{H}$ is the mean function, $\mathcal{A}$ is a bounded linear operator on $\mathcal{H}$, and $\{E_t\}$ is zero-mean Gaussian noise with a trace-class covariance operator $\mathcal{S}$.
\citet{bosq2012linear} showed that if the transition operator $\mathcal{A}$ satisfies $\|\mathcal{A}\|<1$ (the operator norm on $\mathcal{H}$) then there exists a unique Gaussian stationary solution to the equation \eqref{eq:FAR_equation}.
The formula for the spectral density of the functional autoregressive process has a form analogous to the finite-dimensional vector autoregression case (cf. \citet[\S 9.4]{PriestleyMauriceB1981Saat}), but its extension to the functional case appears to be a novel contribution:

%\textcolor{red}{DID YOU ESTABLISH THIS YOURSELF OR DID YOU FIND IT SOMEWHERE? IN THE FORMER CASE, EVEN IF IT'S NOT TOO HARD TO PROVE, IT'S WORTH MENTIONING THAT THIS IS A RESULT FIRST SHOWN IN THIS PAPER (I CAN PROVIDE A MODEST WORDING)}

\begin{proposition}
\label{prop:spec_density_FAR}
The functional autoregressive process of order 1 solving the equation \eqref{eq:FAR_equation} with $\|\mathcal{A}\| < 1$ satisfies the assumption \eqref{eq:assumption_stationarity_summability_of_autocovariance} in the operator sense,
and admits the spectral density
\begin{equation}
\label{eq:spectral_density_FAR}
\mathscr{F}_\omega = \frac{1}{2\pi} (I - \mathcal{A}  e^{-\I\omega})^{-1} \mathcal{S} (I - \mathcal{A}^*  e^{\I\omega})^{-1} , \quad \omega\in (-\pi,\pi)
\end{equation}
in the operator sense \eqref{eq:spectral_density_kernel_operator}. Moreover, if the kernels corresponding to the operators $\mathcal{A}$ and $\mathcal{S}$ are smooth, the spectral density exists also in the kernel sense \eqref{eq:spectral_density_kernel_operator} and the process satisfies the assumptions \ref{assumption:B.4}, \ref{assumption:B.5}, \ref{assumption:B.6}. If the mean function $\mu(\cdot)$ is smooth, the process satisfies also \ref{assumption:B.3}.
\end{proposition}

For our simulations we choose $\mu(x) = 4 \sin(1.5 \pi x)$. The autoregressive operator $\mathcal{A} = \mathcal{A}_c$ is the integral operator with kernel $A_c(x,y) = \kappa_c \exp\left( -(x+2y)^2 \right)$
where the scaling constant $\kappa_c$ is chosen so that $\| \mathcal{A}_c \| = c$. We vary $c$ to control the degree of temporal dependence and let $c\in\{ 0.7, 0.9\}$. The covariance operator $\mathcal{S}$ is the integral operator with kernel $S(x,y) = 1.4  \sin(2\pi x)\sin(2\pi y) + 0.6 \cos(2\pi x)\cos(2\pi y)$. In the simulation results we denote the resulting two processes as $\mathbf{FAR(1)_{0.7}}$ and $\mathbf{FAR(1)_{0.9}}$ for $c=0.7$ and $c=0.9$ respectively.

\end{itemize}

We simulate the functional moving average processes $\mathbf{FMA(2)}$, $\mathbf{FMA(4)}$, $\mathbf{FMA(8)}$, and the functional autoregressive processes $\mathbf{FAR(1)_{0.7}}$, $\mathbf{FAR(1)_{0.9}}$, over temporal periods of varying length, specifically $T\in\{150,300,450,600,900,1200\}$. The simulation is started from the stationary distribution of the respective processes.

The simulations must be obviously performed in a finite dimension. We performed the simulation in the third-order B-spline basis created by equidistantly placing 20 knots on the interval $[0,1]$. Hence the basis admits $21$ elements. The B-spline basis is efficient in expressing smooth functions \citep{ramsay2007applied}.

The sparse observations are then obtained by the following process. We set a maximum number of locations to be sampled $N^{max} \in\{5,10,20,30,40\}$. For each $t=1,\dots,T$, a random integer $N_t$ is independently drawn from the uniform distribution on $0,1,\dots,N^{max}$. Next, for each $t=1,\dots,T$, we independently draw $N_t$ random locations $x_{tj}, j=1,\dots,N_t$ from the uniform distribution on $[0,1]$. At each location, an independent identically distributed Gaussian measurement error $\epsilon_{tj} \sim N(0,\sigma^2)$ is added and the ensemble $Y_{tj} = X_t(x_{tj}) + \epsilon_{tj}, j=1,\dots,N_t, t=1,\dots,T$ is used as the dataset for the estimation procedure. Therefore the observation protocol satisfies the assumptions \ref{assumption:B.1} and \ref{assumption:B.2}.

The measurement error variance is chosen in the way that the ratio $\tr( \mathscr{R}_0 ) / \sigma^2$, which we interpret as a basic signal-to-noise ratio metric, is $20$. The same signal-to-noise ratio was used in the simulation study by \citet{yao2005functional}.
Further simulation results of ours not reported here indicate that moderate variations of the signal-to-noise ratio do not change the conclusions of this simulation study.
%Specifically we set $\sigma$ as
%$0.274$, $0.316$, $0.37$, $0.224$, and $0.228$
%for the above defined functional moving average processes $\mathbf{FMA(2)}$, $\mathbf{FMA(4)}$, $\mathbf{FMA(8)}$, and functional autoregressive processes $\mathbf{FAR(1)_{0.7}}$, $\mathbf{FAR(1)_{0.9}}$ respectively.

\subsection{Estimation of the Spectral Density}
\label{subsec:simulations_estimation of spectral density}

In this subsection we quantify the estimation error of the spectral density estimator \eqref{eq:spectral_density_estimator_bartlett_smoother} in our simulation setting. In particular, we want to explore the dependence of the estimation error on the length $T$ of the time series and the number $N^{max}$ impacting the average number of measurements per curve. 

For each of the considered process and for each pair of the sample size parameters $T\in\{150,300,450,600,900,1200\}$ and $N^{max}\in\{5,10,20,30,40\}$ we simulated $100$ independent realisations. We have run the estimation procedure introduced in Sections \ref{subsec:nonparam_estimation} and \ref{subsec:Spectral Density Kernels Estimation}. In each case, the tuning parameters $B_\mu$, $B_R$, and $B_V$ are selected by the $K$-fold cross-validation as explained in Section \ref{subsec:selection of B_mu, B_R and B_V}.
The selection of Bartlett's span parameter are discussed in Section \ref{subsec:selection of Q_Bartlett}.
Based on the results of the simulation study, we introduce a simple selection rule that works well for spectral density estimation.
The optimal $\QBartlett$ depends clearly on the (unknown) dynamics of the functional time series. As a compromise across the simulated processes we propose to use the following selection rule
\begin{equation}
\label{eq:Q_Bartlett decision rule}
\QBartlett = \lfloor T^{1/3} \left(\bar{N}\right)^{1/4} \rfloor
\end{equation}
where $\bar{N}$ is the average number of measurements per curve and $\lfloor\cdot\rfloor$ is the integer part of a given real number.
The selection rule \eqref{eq:Q_Bartlett decision rule} was hand-picked for the considered range of variables $T$ and $N^{max}$ and should not be used for extrapolation, especially not for dense observation schemes.
%The Table \ref{table:Q_Barlett selection} presents the selected values of $\QBartlett$ for the considered sample sizes.

We measure the quality of the spectral density estimation by the relative mean square error defined as
\begin{equation}\label{eq:RMSE_definition}
RMSE =
\frac{
\int_{-\pi}^\pi \int_0^1\int_0^1 | \hat{f}_\omega(x,y) - f_\omega(x,y) |^2 \D x \D y \D\omega
}{
\int_{-\pi}^\pi \int_0^1\int_0^1 | f_\omega(x,y) |^2 \D x \D y \D\omega
 }
\end{equation}
where $\hat{f}_\omega(\cdot,\cdot)$ and $f_\omega(\cdot,\cdot)$ are respectively the estimated and the true spectral density kernels at the frequency $\omega\in(-\pi,\pi)$.
Due to space constraints, we present in Table \ref{table:estimation_of_spectral_density_FRO_MA4} the results only for the functional moving average process of order 4, $\mathbf{FMA(4)}$. The results for the remaining considered processes are reported in Section \ref{appendix:supplementary results}.

\begin{table}[]
\caption{Average relative mean square errors (defined in \eqref{eq:RMSE_definition}) of the spectral density estimators for the above defined functional moving average process of order 4 ($\mathbf{FMA(4)}$) and varying sample sizes. The numbers in parentheses are the standard deviations of the relative mean square error. Each cell of the table (each error and its standard deviation) is the result of 100 independent simulations. The Bartlett's span parameter $\QBartlett$ was selected by the rule \eqref{eq:Q_Bartlett decision rule} 
}
\label{table:estimation_of_spectral_density_FRO_MA4}
\begin{tabular}{clllll}
%   & \multicolumn{6}{c}{$T$}  \\ 
$T$\textbackslash$N^{max}$  & \multicolumn{1}{c}{5} & \multicolumn{1}{c}{10} & \multicolumn{1}{c}{20} & \multicolumn{1}{c}{30} & \multicolumn{1}{c}{40}  \\[5pt] 
                                                                                                    % FMA(4)                                                    
150  & 0.312 (0.060) & 0.225 (0.063) & 0.184 (0.060) & 0.170 (0.049) & 0.165 (0.050) \\
300  & 0.206 (0.040) & 0.157 (0.042) & 0.124 (0.028) & 0.115 (0.030) & 0.110 (0.033) \\
450  & 0.167 (0.033) & 0.126 (0.034) & 0.097 (0.022) & 0.092 (0.027) & 0.081 (0.021) \\
600  & 0.137 (0.027) & 0.107 (0.027) & 0.083 (0.017) & 0.077 (0.023) & 0.071 (0.017) \\
900  & 0.115 (0.020) & 0.082 (0.015) & 0.067 (0.016) & 0.061 (0.015) & 0.056 (0.016) \\
1200 & 0.096 (0.019) & 0.072 (0.015) & 0.056 (0.013) & 0.050 (0.012) & 0.047 (0.012)

\end{tabular}
\end{table}

\begin{figure}[hbt!]
\centering
\includegraphics[width=0.8\textwidth]{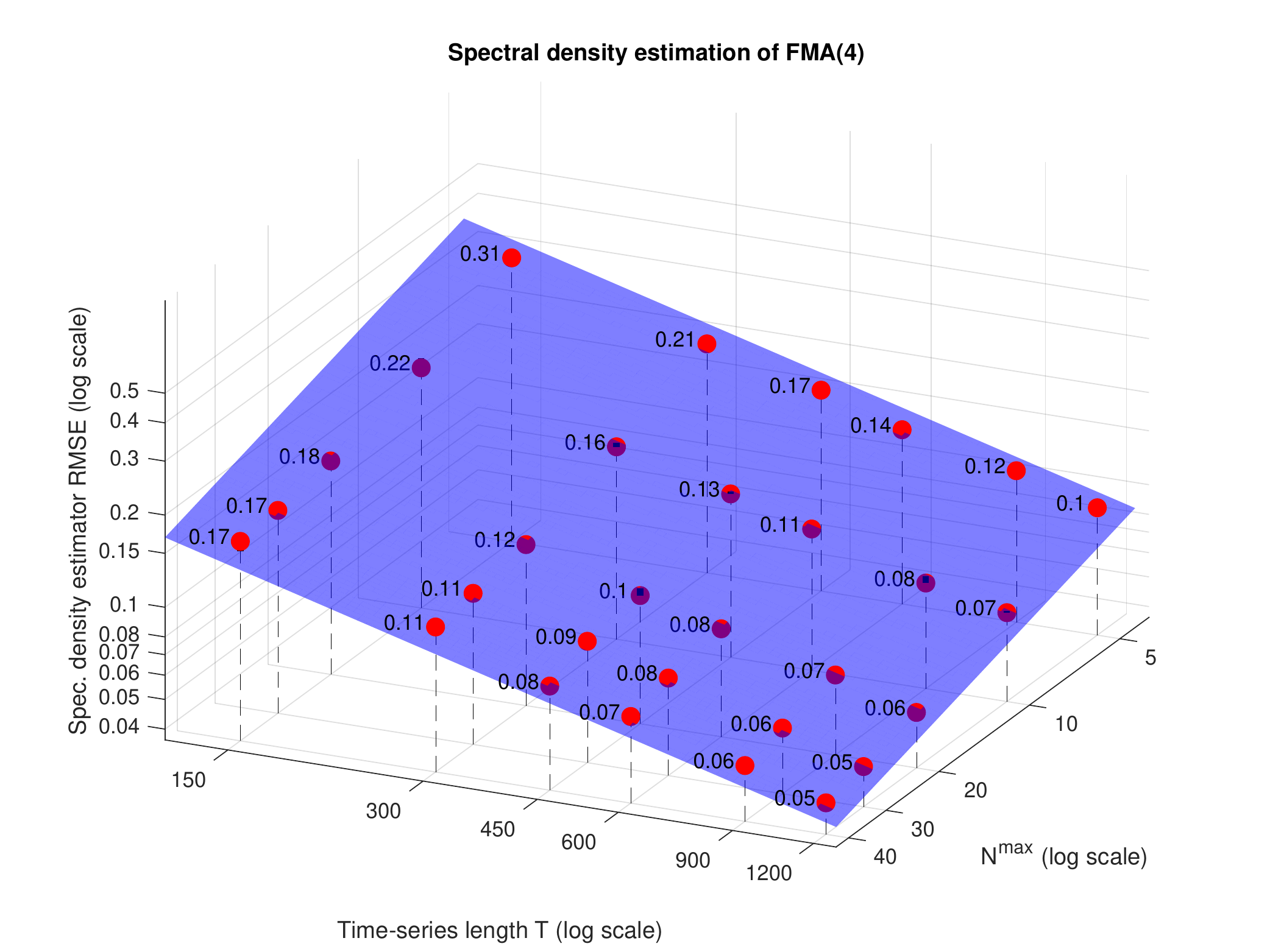}
\caption{The dependence of spectral density estimation relative mean square error (red points with labels of the magnitude of this error) on the sample size parameters $T$ and $N^{max}$. The blue plane is the estimated regression surface in model \eqref{eq:simulations_linear_model}.
}
\label{fig:MA4_spec_density}
\end{figure}

Concerning the results of Table \ref{table:estimation_of_spectral_density_FRO_MA4}, one can raise an interesting design question: 
\begin{quote}
Provided one has a fixed budget for the total number of measurements to be made, should opt to record fewer spatial measurements over a longer time interval (lengthy but sparsely observed time series), or rather record dense spatial measurements over a shorter time period (short but densely observed time series)?
\end{quote}

In order to answer this question we define a simple linear model to asses the dependence of the relative mean square error on the considered sample size parameters $T$ and $N^{max}$. For each of the considered processes we fit the linear model
\begin{equation}
\label{eq:simulations_linear_model}
\log( RMSE(N^{max},T) ) = \beta_0 + \beta_1 \log( N^{max} ) + \beta_2 \log( T) + e
\end{equation}
where $RMSE(N^{max},T)$ is the average relative mean square error for the considered parameters $T$ and $N^{max}$, $(\beta_0,\beta_1,\beta_2)$ are the regression parameters, and $e$ is a homoskedastic model error.

 The least square estimate of \eqref{eq:simulations_linear_model} yields $( \hat{\beta}_0, \hat{\beta}_1, \hat{\beta}_2 ) = (    1.98,   -0.32,   -0.57)$.
%The coefficient of determination $R^2$ of this fit very high ($0.97$).
The coefficient $\hat{\beta}_2$ is larger than $\hat{\beta}_1$ in absolute value, therefore the relative increase of the time-length $T$ has a stronger effect in reducing the relative mean square error of the estimated spectral density than the same relative increase in the number of points per curve. The apparent conclusion is that, in order to estimate the spectral density of a smooth functional time series, the better strategy is to invest in longer time-horizon $T$ rather than denser sampling regime.

%\textcolor{red}{XXX delete this?}
%A more thorough examination of the fitted surface plot in Fig. \ref{fig:MA4_spec_density} provides further insight. The values of the relative mean square error corresponding to $N^{max}=10$ and $N^{max}=20$ are below the fitted regression surface indicating that the average number of points between 5 and 10 seems to be the most favourable (for the considered smooth functional time series). On the other hand, the relative means square error seems to be reaching a plateau when $N^{max}=40$ and more dense sampling might not decrease the error much further.

\subsection{Recovery of Functional Data from Sparse Observations}
\label{subsec:simulations_recovery}

In this section, we examine the performance of the functional recovery procedure proposed in Section \ref{subsec:functional_data_recovery}. We compare the recovery performance of our dynamic predictor \eqref{eq:BLUP_Xs_empirical}, in the following denoted as the \textit{dynamic recovery}, with its static version that relies only on the lag-zero covariance and hence does not exploit the temporal dependence. In the following, we call this predictor the \textit{static recovery}. This static recovery is in fact the predictor \eqref{eq:BLUP_Xs_empirical} with the Bartlett's span parameter $\QBartlett$ set to $1$.

We simulate 100 independent realisations for each of the considered functional moving average processes $\mathbf{FMA(2)}$, $\mathbf{FMA(4)}$, $\mathbf{FMA(8)}$, and the considered functional autoregressive processes $\mathbf{FAR(1)_{0.7}}$, $\mathbf{FAR(1)_{0.9}}$,  (their definitions in Section \ref{subsec:simulations_estimation of spectral density}) and each combination of the sample size parameters $T\in\{150,300,450,600,900,1200\}$ and $N^{max}\in\{5,10,20,30,40\}$.
Again, due to space constraints, we state here the results only for the functional moving average process of order 4, $\mathbf{FMA(4)}$. The results for the other considered processes are stated in Section \ref{appendix:supplementary results}.

For each dataset we run the estimation procedure from Sections \ref{subsec:nonparam_estimation} and \ref{subsec:Spectral Density Kernels Estimation}. The tuning parameters $B_\mu$, $B_R$, and $B_V$ are selected by $K$-fold cross-validation as explained in Section \ref{subsec:selection of B_mu, B_R and B_V}. The parameter $\QBartlett$ is selected again by the rule \eqref{eq:Q_Bartlett decision rule}.

We define the functional recovery (either dynamic or static) relative mean square error as
\begin{equation}
\label{eq:simulations_RMSE}
RMSE = \frac{1}{T} \sum_{t=1}^T \frac{
\int_0^1 \left(\hat{X}_t(x) - X_t(x) \right)^2 \D x
}{ \tr \mathscr{R}_0 }
\end{equation}
where $\hat{X}_t$ is the recovered functional curve at $t=1,\dots,T$, either dynamically or statically, and $X_t$ is the true (unobserved) functional datum.

The key factor contributing to the quality of the functional recovery is the estimate $\hat\sigma^2$ of the additive measurement error variance parameter $\sigma^2$. A very small value of the estimated $\hat\sigma^2$ can lead to an ill-conditioned matrix needed to be inverted in \eqref{eq:BLUP_X_empirical}, thus resulting in a defective recovery of the functional data. Because this circumstance affects the relative mean square error metric, we opt to calculate the median of the relative mean square errors as a better indicator of the typical recovery error instead.

We calculate the \textit{relative gain} as
\begin{equation}
\label{eq:simulations_relative gain}
Relative\text{ }gain = \left( \frac{RMSE(static)}{RMSE(dynamic)} -1\right)*100 \%
\end{equation}
where $RMSE(static)$ is the median relative mean square error of the static recovery and $RMSE(dynamic)$ is the median mean square error of the dynamic recovery.

 {
Table \ref{table:recovery_relative_gain_MA4} summarizes the relative gains of dynamic recovery over the static recovery. Unsurprisingly, the relative gain is strikingly large for sparser designs. This can be explained by the fact that in sparse designs there is not sufficient information to interpolate the functional curves themselves, and the observed data in neighbouring curves are crucial for the recovery of the curves. That being said, it is observed that even when the number of points sampled per curve are as many as 40, the improvement remains  substantial, demonstrating that the new methodology should be preferred over methods designed for the i.i.d. case when dependence is present.}

\begin{table}[hbt]
\caption{Relative gain \eqref{eq:simulations_relative gain} between median relative mean square error of dynamic recovery and median relative mean square error of static recovery. Positive percentage signifies that dynamic recovery has smaller error. Simulations from the functional moving average of order 4, $\mathbf{FMA(4)}$. Each cell of the table is the result of 100 independent simulations
}
\label{table:recovery_relative_gain_MA4}
\centering
\begin{tabular}{crrrrr}
%   & \multicolumn{6}{c}{$T$}                                                \\ 
$T$\textbackslash$N^{max}$  & \multicolumn{1}{c}{5} & \multicolumn{1}{c}{10} & \multicolumn{1}{c}{20} & \multicolumn{1}{c}{30} & \multicolumn{1}{c}{40}   \\[5pt] 
150  & 67 \% & 38 \% & 38 \% & 23 \% & 30 \% \\
300  & 53 \% & 39 \% & 33 \% & 31 \% & 26 \% \\
450  & 52 \% & 45 \% & 38 \% & 30 \% & 24 \% \\
600  & 45 \% & 41 \% & 32 \% & 26 \% & 24 \% \\
900  & 54 \% & 41 \% & 37 \% & 30 \% & 22 \% \\
1200 & 54 \% & 45 \% & 34 \% & 26 \% & 21 \%
\end{tabular}
\end{table}

%%%%%%%%%%%%%%%%%%%%%%%%%%%%%%%%%%%%%%%%%%%%%%%%%%%%%%%%%%%%%%%%%%%%%%%%%%%%%%%%%%%%%%%%%%%%%%%%%%%%%
%%%%%%%%%%%%%%%%%%%%%%%%%%%%%%%%%%%%%%%%%%%%%%%%%%%%%%%%%%%%%%%%%%%%%%%%%%%%%%%%%%%%%%%%%%%%%%%%%%%%%

\section{Data Analysis: Fair-Weather Athmospheric Electricity} \label{sec:data analysis}

The atmosphere is weakly conductive due to the ionization of molecules and this conductivity can be continuously measured by a variable called \emph{atmospheric electricity} \citep{tammet2009joint}. The ionization is the outcome of complicated physical-chemical processes that are subject to the current weather conditions. Since unfair weather conditions affect and alter these processes \citep{israelsson2001variation}, climatologists are interested in analysing the atmospheric electricity variable only under fair weather conditions (the definition of fair weather is given later). The analyses under fair weather conditions are of particular interest because the fair-weather electricity variable is a valuable source of information in global climate research  \citep{tammet2009joint} as well as with regards to air pollution \citep{israelsson2001variation}.

\citet{tammet2009joint} published an open-access database of atmospheric electricity time series accompanied by some meteorological variables. Most of the data come from weather stations across the former Soviet Union states and their data quality is assessed as high \citep{tammet2009joint}. In this article, we analyse the time series of one weather station, namely that measured at the station near Tashkent, Uzbekistan. The atmospheric electricity was recorded between the years 1989 and 1993 in the form of hourly averages. Besides the atmospheric electricity, a number of other meteorological variables were measured, of which we use two: the wind speed and the total cloudiness.

\begin{figure}[]
\centering
\includegraphics[width=\textwidth]{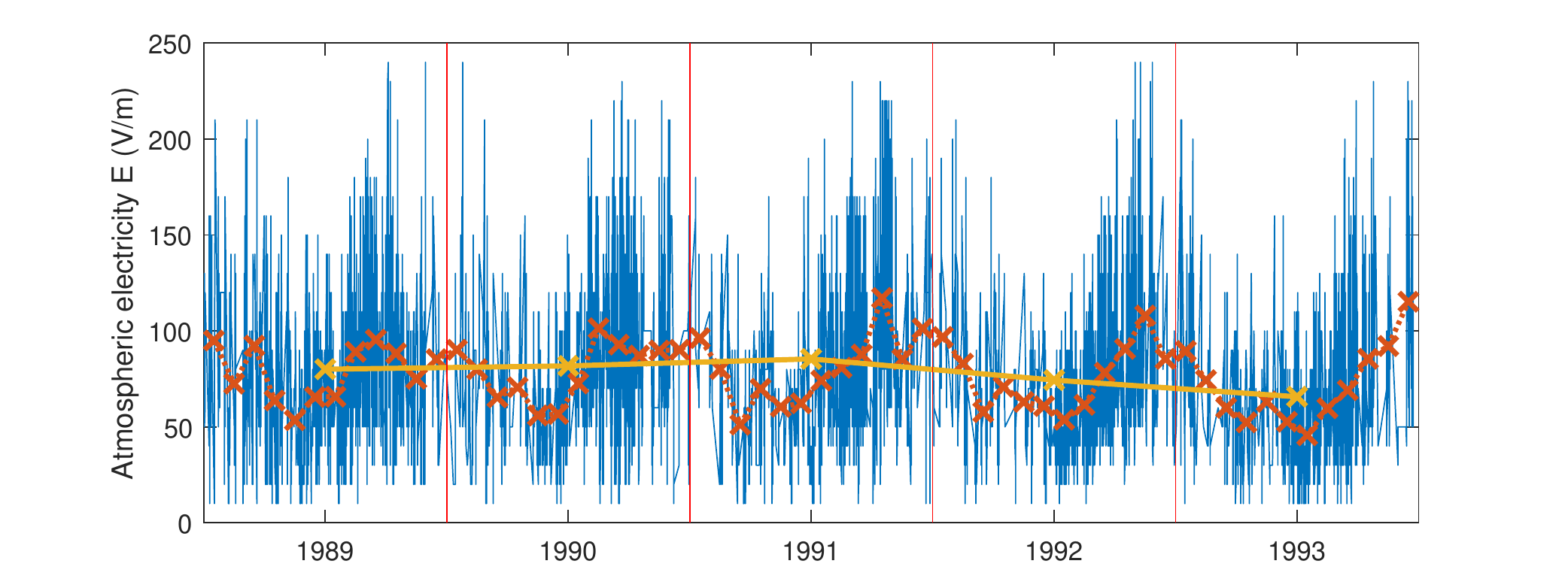}
\caption{Overview of the fair-weather atmospheric electricity time series measured in Tashkent, Uzbekistan. All fair-weather hourly measurements (blue line) accompanied by monthly means (brown crosses, brown dotted line) and yearly means (yellow crosses, yellow solid line).}
\label{fig:full_ts_means}
\end{figure}

The definition of the fair-weather criteria is not simple and can often be relatively subjective \citep{xu2013periodic}. Inspired by criteria in climatology research \citep{xu2013periodic,israelsson2001variation}, we define the weather conditions as fair if the particular hourly measurement satisfies all of the following conditions:
\begin{itemize}
\item the wind speed is less than $20 \, km/h$,
\item the sky is clear (the total cloudiness variable is equal to $0$),
\item the atmospheric electricity $E$ satisfies $0 < E < 250 \, V/m$.
\end{itemize}

Because of the above stated fair-weather criteria (and some genuinely missing data in the database), the resulting fair-weather electricity time series is, in fact, unevenly sampled time series. Nevertheless, we assume there exists an underlying continuous truth, corresponding to the atmospheric electricity if the weather was fair.
The latent process of fair-weather atmospheric electricity is considered smooth and its values are observed only under the fair-weather conditions, possibly with a deviation from the truth (noise). Based on the above discussed natural mechanisms, we justify the assumption that the censoring protocol is independent of the underlying fair-weather atmospheric electricity process.

The underlying fair-weather atmospheric electricity process is a continuous scalar time series. Previous research \citep{hormann2010weakly,hormann2015dynamic,hormann2016detection,aue2015prediction} has demonstrated the usefulness of segmenting a continuous scalar time series into segments of an obvious periodicity, usually days, and thus constructing a functional time series. A key benefit of this practice is the separation of intra-day variability and the temporal dependence across the days while preserving a fully non-parametric model.

We use the same approach in our analysis as well. We segment the (latent) continuous time series into days and consider each day us an unobserved (latent) functional datum defined on $[0,24]$. We place the hourly observations in the middle of the hour interval, i.e. $0.5,1.5,2.5,\dots,23.5$. Because of the above fair-weather criteria, the constructed fair-weather atmospheric electricity time series falls into the sparsely observed functional time series framework defined in Section \ref{subsec:observation scheme}.

Figure \ref{fig:full_ts_means} presents an overview of the considered fair-weather atmospheric electricity time series accompanied by monthly and yearly means. Figure \ref{fig:ts_raw_4_days} provides a zoomed-in perspective into a stretch of data in 4 consecutive days.

\begin{figure}
\centering
\includegraphics[width=0.9\textwidth]{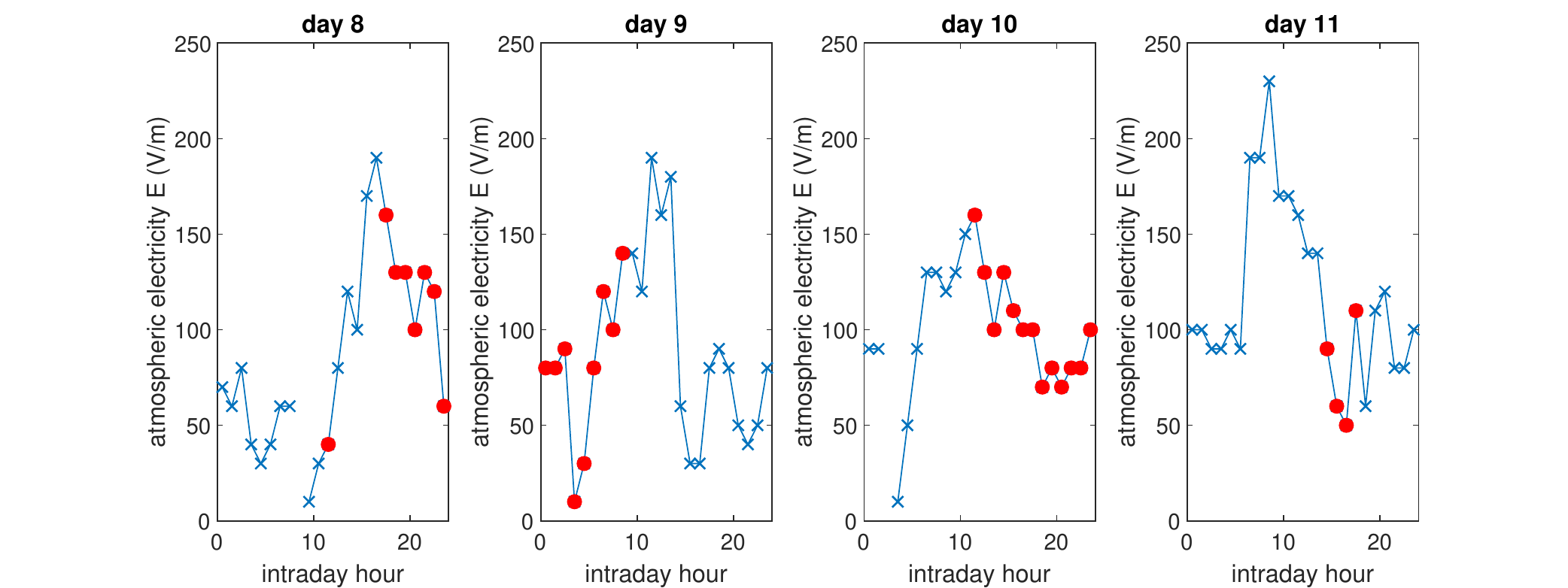}
\caption{Example of atmospheric electricity profiles over 4 consecutive days. The fair-weather atmospheric electricity measurements are highlighted as red points. The unfair-weather measurements (blue crosses) are not used for the analysis. }
\label{fig:ts_raw_4_days}
\end{figure}

In summary, the fair-weather atmospheric electricity functional time series has the following features:
\begin{itemize}
\item the data are recorded over 5 years, therefore the time horizon of the functional time series is $T = 1826$ (days),
\item there are 1118 days have at least 1 fair-weather measurement (61 \%),
\item there are 251 gaps in time series (we define a gap as a stretch of days where there is no measurement within these days) with the average length of 2.8 days,
\item there are 12997 fair-weather measurements in total, i.e. 7.1 on average per day, or 11.6 on average per day among the days with at least one measurement.
\end{itemize}

The statistical question raised is the following. Benefiting from the separation of intra-day variability and temporal dependence across the days, can we fit an interpretable model of the process dynamics? Additionally, we aim to recover the latent functional data, fill in the gaps in the data, remove the noise, and construct confidence bands.

We analyse the fair-weather atmospheric electricity data by the means of Section \ref{sec:model}. Initially, after removing the intra-day dependence by subtracting the estimate $\hat\mu(\cdot)$ we inspect the periodicity identification chart introduced in Section \ref{subsec:periodic behaviour identification}. Specifically, we construct the said chart with $\QBartlett=1000$ and plot the trace of the estimated spectral density operator against frequencies $\omega\in(0,\pi)$. We identify the peaks of this plot as suggesting the presence of periodicities in the corresponding frequencies.

The largest peak in Fig. \ref{fig:periodogram} clearly corresponds to yearly periodicity together with a half-year harmonic. The peak is not entirely at 365 days because of the combination of the following factors: discretisation of the frequency grid, numerical rounding, and most likely the slight smoothing by $\QBartlett = 1000$.

\begin{figure}
\centering
\includegraphics[width=0.9\textwidth]{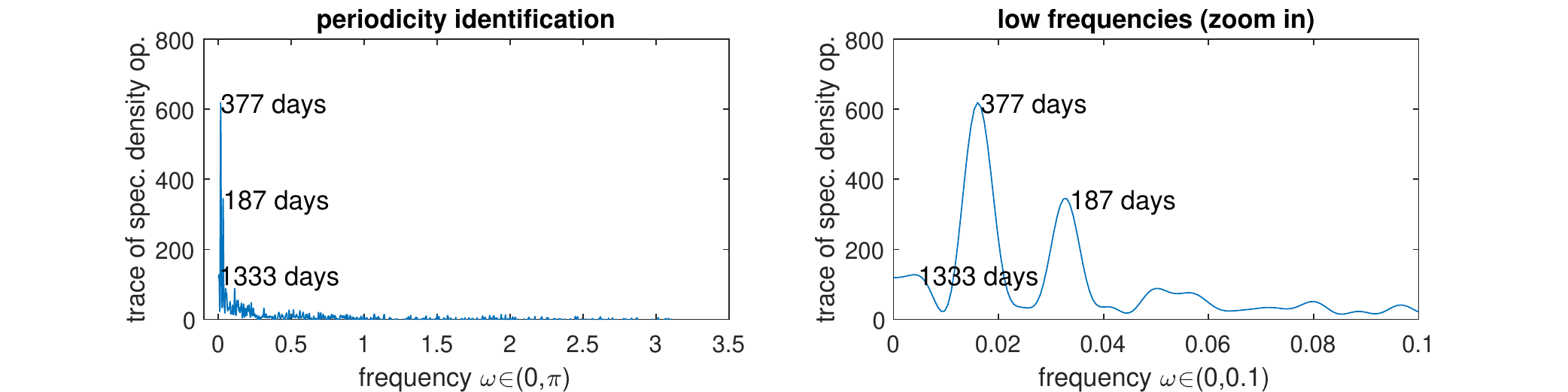}
\caption{\textbf{Left:} The periodicity identification plot with $\QBartlett=1000$. The labels at first 4 peaks convert the frequency into the corresponding periodicity. \textbf{Right:} Zoom-in into low frequencies. 
}
\label{fig:periodogram}
\end{figure}

Once the yearly periodicity is discovered, we opt to model it deterministically, as is usual in (scalar) time series. Thus we propose the model
\begin{equation}
\label{eq:data_analysis_model}
Y_{tj} = \mu(x_{tj}) + s_t + X_t(x_{tj}) + \epsilon_{tj}
\end{equation}
where $Y_{tj}$ are the observed measurements at locations $x_{tj}$, $\mu(\cdot)$ is the intra-day mean, $s_t$ is yearly seasonality adjustment, and the ``residual'' process $X_t(\cdot)$ is a zero-mean stationary weakly-dependent functional time series. The assumptions of an additive relation of $\mu(\cdot)$ and $s_t$ as well as the stationarity of $X_t(\cdot)$ were justified by exploratory analysis.

We fit the model \eqref{eq:data_analysis_model} in the following order. First, we estimate $\mu(\cdot)$ by a local-linear smoother. Nevertheless, we expect the mean function to be periodic and assume $\mu(0)=\mu(24)$. Thus we modify the estimator \eqref{eq:local_LS_for_mu} to measure the distance between $x$ and $x_{tj}$ as if the endpoints of the interval $[0,24]$ were connected. Having estimated $\hat\mu(\cdot)$, we estimate the yearly periodic seasonality adjustment $s_t$ again by a local-linear smoother, again by assuming continuity between first day and last day of the year. The smoothing parameter was chosen by leave-one-year-out cross-validation. Figure \ref{fig:est_mu_st} presents the estimates $\hat{\mu}(\cdot)$ and $\hat{s}_t$. We observe that the intraday mean exhibits two peaks at around 4 a.m. and 3 p.m. The yearly seasonality is almost sinusoidal with low values in the spring and summer and high values in the autumn and winter.

\begin{figure}
\centering
\includegraphics[width=0.9\textwidth]{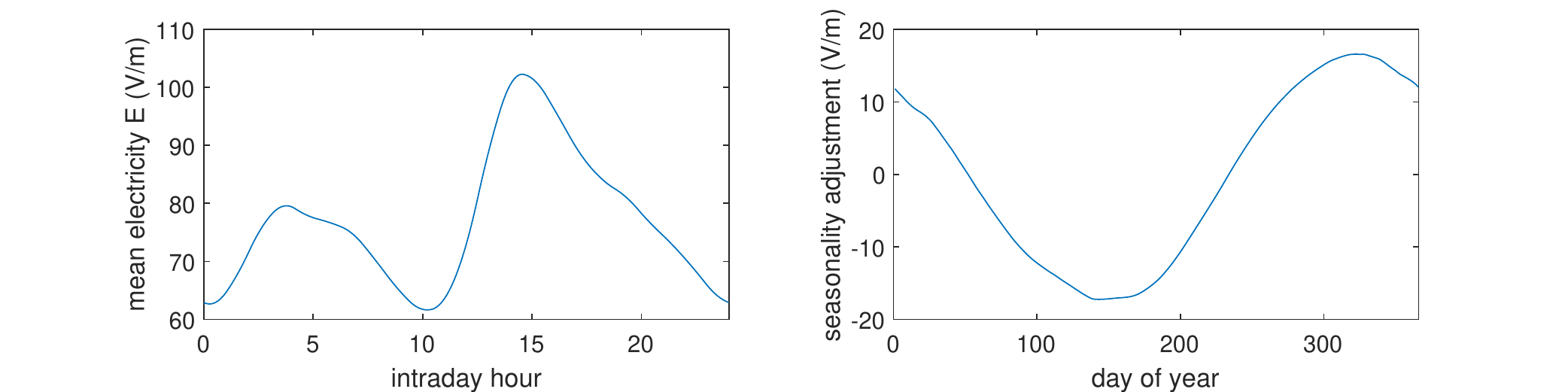}
\caption{\textbf{Left:} The estimated intra-day mean $\hat{\mu}(\cdot)$. \textbf{Right:} the estimated yearly seasonality adjustment $\hat{s}_t$}
\label{fig:est_mu_st}
\end{figure}

Once the first-order structure given by $\mu(\cdot)$ and $s_t$ is estimated, we calculate the raw covariance \eqref{eq:raw_autocovariances} by subtracting both $\hat{\mu}(x)$ and $\hat{s}_t$.
The lag-0 covariance kernel $R_0(\cdot,\cdot)$ is estimated by \eqref{eq:local_LS_cov_lag_0}. 
For the estimation of the components of \eqref{eq:local_LS_est_sigma2}, namely $\hat{V}(\cdot)$ and $\bar{R}_0(\cdot)$, we use the same periodicity adjustment as for $\hat{\mu}(\cdot)$ because we expect the marginal variance (with and without the ridge contamination) to be continuous across midnight.
For illustration and interpretation purposes we estimate also the lag-1 autocovariance $R_1(\cdot,\cdot)$ by \eqref{eq:local_LS_cov_lag_nu}.
Figure \ref{fig:est_covariance_kernels_without_lines} shows the surface plots of these estimates. An interesting element of the estimated lag-0 covariance kernel is the peak at afternoon hours signifying higher marginal variance of the fair-weather atmospheric electricity in the afternoon hours. The estimated lag-0 correlation kernel demonstrates that the observations measured close to each other are highly correlated and the correlation diminishes as the distance grows. The estimated lag-1 autocovariance and autocorrelation kernels show that the correlation between two consecutive days is positive. The lag-1 autocorrelation kernel features a lifted-up surface up to correlation 1 in the eastern corner of the surface plot. The clear interpretation is that the late hours of one day are strongly correlated with early morning hours of the following day.

\begin{figure}
\centering
\includegraphics[width=0.78\textwidth]{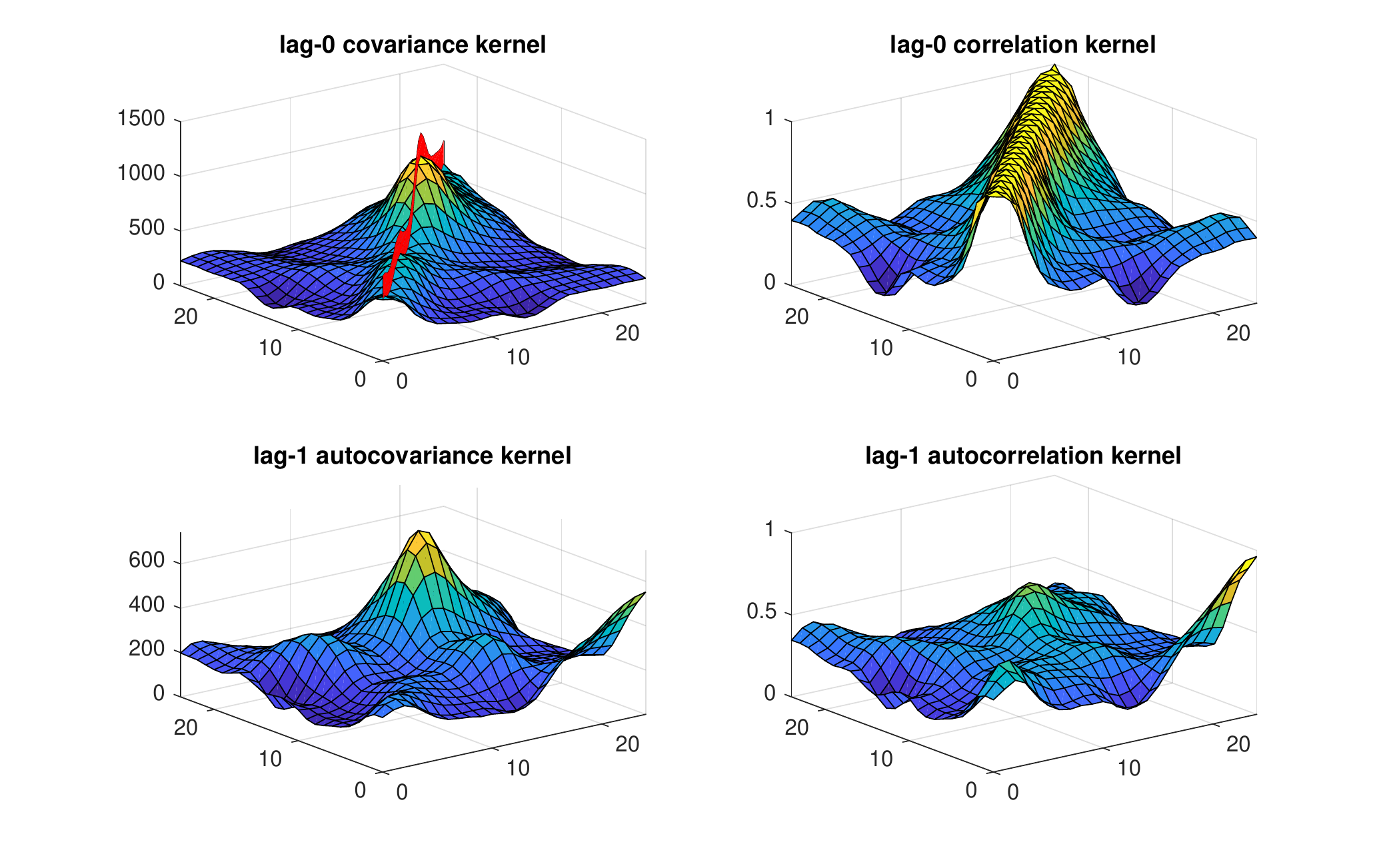}
\caption{
\textbf{Top-left:} $\hat{R}_0(\cdot,\cdot)$, the estimated lag-0 covariance kernel $R_0(\cdot,\cdot)$ (surface), the ridge contamination by the measurement error (red) \textbf{Top-right:} the correlation kernel corresponding to $\hat{R}_0(\cdot,\cdot)$.
\textbf{Bottom-left:} $\hat{R}_1(\cdot,\cdot)$, the estimated lag-0 covariance kernel $R_1(\cdot,\cdot)$.
\textbf{Bottom-right:} the correlation kernel corresponding to $\hat{R}_1(\cdot,\cdot)$.
}
\label{fig:est_covariance_kernels_without_lines}
\end{figure}

\begin{figure}
\centering
\includegraphics[width=0.78\textwidth]{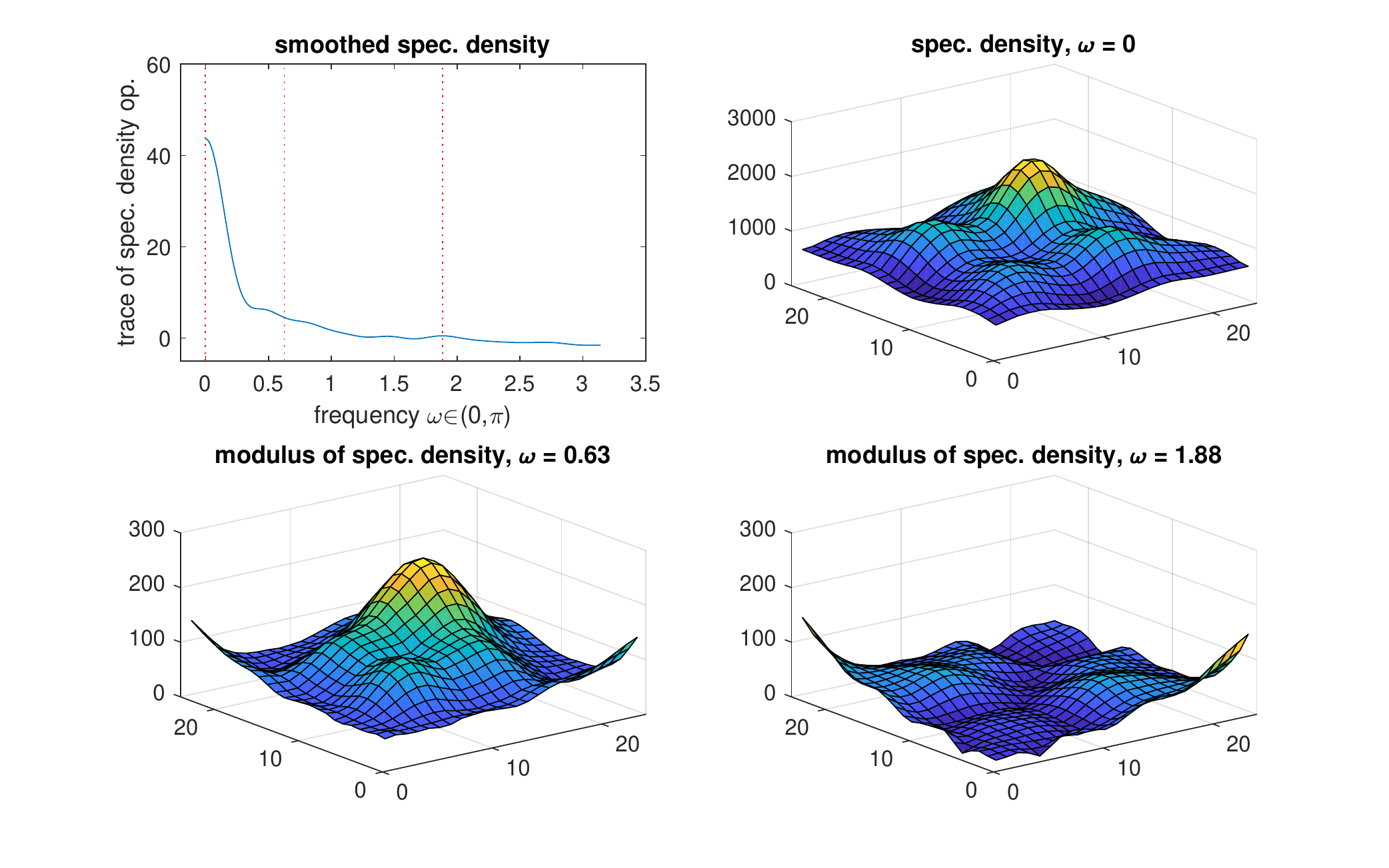}
\caption{\textbf{Top-left:} the traces of estimated spectral density operators with the highlighted frequencies considered in the next plots.
\textbf{Top-right:} the estimated spectral density at frequency $\omega=0$ (it is a real-valued kernel for $\omega=0$).
\textbf{Bottom-left and bottom-right:} the modulus of the complex-valued estimated spectral density at frequencies $\omega=0.63$ and $\omega=1.88$ respectively.
}
\label{fig:specDensity_4_paper}
\end{figure}

In order to estimate the spectral density consistently, we need to select a moderate value of Bartlett's span parameter $\QBartlett$. Plugging in the size of the dataset into the formula 
\eqref{eq:Q_Bartlett decision rule} we set $\QBartlett = 19$. Figure \ref{fig:specDensity_4_paper} presents a few views on the estimated spectral density kernels.

Once the spectral density is estimated, we apply the functional recovery method of Section \ref{subsec:functional_data_recovery} and estimate the unobserved functional data. The method produces estimates of intra-day profiles of fair-weather atmospheric electricity that can be interpreted as predicted atmospheric electricity if the weather was fair at given time, without the modelled noise. As a by-product, the method fills in the gaps in the data (the stretches of days without any measurement). Another output is the construction of confidence bands (under the Gaussianity assumption). Figure \ref{fig:recovery_4_days} presents 4 consecutive days with estimated (noiseless) fair-weather atmospheric electricity together with 95\%-simultaneous confidence bands. It is important to note that these bands are supposed to cover the assumed smooth underlying functional data, not the observed data produced by adding measurement errors to the smooth underlying process.

%\begin{figure}
%\centering
%\includegraphics[width=\textwidth]{figures/recovery_4_days_both.pdf}
%\caption{Fair-weather atmospheric electricity hourly measurements (red points) over 4 consecutive days; \textit{dynamic} functional recovery of the latent smooth fair-weather atmospheric electricity process (blue solid line); 95\%-simultaneous confidence bands for the functional data of the said latent process by dynamic recovery (blue band); \text{static} functional recovery of the latent smooth fair-weather atmospheric electricity process (red dotted line); 95\%-simultaneous confidence bands by the static recovery (yellow band).}
%\label{fig:recovery_4_days}
%\end{figure}

\begin{figure}
\centering
\includegraphics[width=\textwidth]{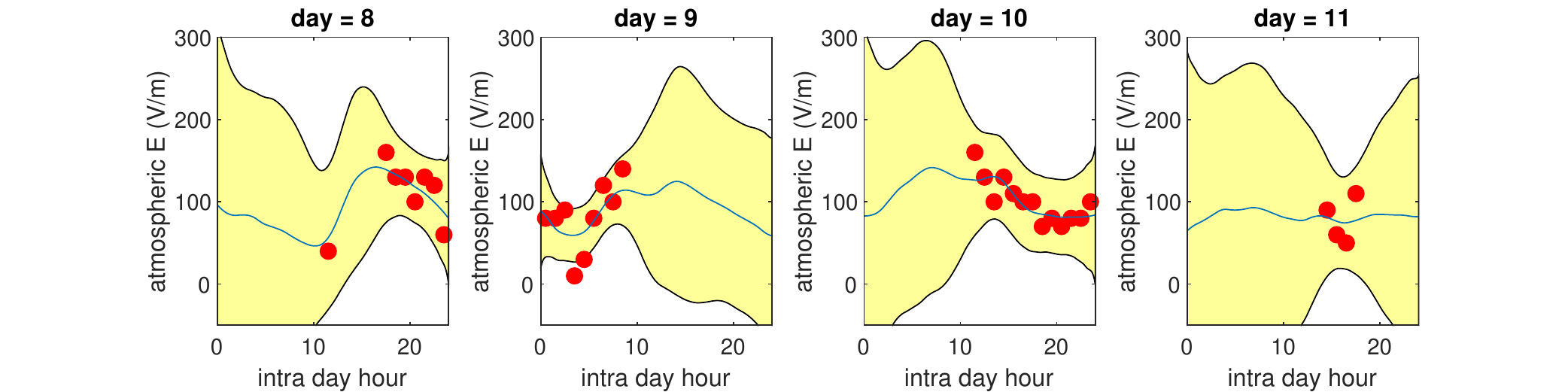}
\caption{Fair-weather atmospheric electricity hourly measurements (red points) over 4 consecutive days; functional recovery of the latent smooth fair-weather atmospheric electricity process (blue); 95\%-simultaneous confidence bands for the functional data of the said latent process (yellow).}
\label{fig:recovery_4_days}
\end{figure}

\appendix
\newpage

\addtocontents{toc}{\protect\setcounter{tocdepth}{1}}

\section{Practical Implementation Concerns}
\label{appendix:practical_implementation_concerns}

\subsection{Selection of bandwidths $B_\mu$, $B_R$, and $B_V$}
\label{subsec:selection of B_mu, B_R and B_V}
Our estimation methodology involves three bandwidth parameters $B_\mu, B_R, B_V$ that need to be selected based on some data-driven criterion. To reduce the computational cost we choose to perform the selection of the parameters in successive fashion.

The selection of a bandwidth parameter in kernel smoothing has been extensively studied in literature for the case of locally polynomial regression. The classical selector by \citet{ruppert1995} calculates the asymptotic mean square error and plugs-in some estimated quantities. However, their methodology applies to the independent case which is distinctly different from the setting of this paper and hence we opt for a cross-validation selection procedure. The selection of the smoothing parameters by cross-validation has already been implemented by \citet{yao2005functional}. Here we use a similar approach.

To further reduce the computational requirements we opt for a K-fold cross-validation strategy instead of the leave-one-curve-out cross-validation originally suggested by \citet{rice1991estimating}. For the K-fold cross-validation, we work with $K=10$ partitions, as follows.
We randomly split the functional curves into $K$ partitions and denote the time indices sets as
$ \mathcal{T}_1, \dots, \mathcal{T}_K $.
For each $k\in\{1,\dots,K\}$, denote $\hat{\mu}^{(-k), B_\mu^0}$ the estimate of the common mean function $\mu$ calculated by the smoother \eqref{eq:local_LS_for_mu} from data without the partition $k$ and using the candidate smoothing parameter $B_\mu^0$. We select the smoothing parameter $B_\mu$ by minimizing the following loss:
%\begin{equation}
%\label{eq:CV_minimisation_B_mu}
%B_\mu = \argmin_{B_\mu^0}
%\frac{1}{K}\sum_{k=1}^K \sum_{t=1}^T \sum_{j=1}^{N_t^{(k)}}
%\left\{
%Y_{tj} - \hat{\mu}^{(-k),B_\mu^0}( x_{tj}^{(k)} ) 
%\right\}^2
%.
%\end{equation}
\begin{equation}
\label{eq:CV_minimisation_B_mu}
B_\mu = \argmin_{B_\mu^0}
\frac{1}{K}\sum_{k=1}^K \sum_{t\in \mathcal{T}_k } \sum_{j=1}^{N_t}
\left\{
Y_{tj} - \hat{\mu}^{(-k),B_\mu^0}( x_{tj} ) 
\right\}^2
.
\end{equation}

Once the smoothing parameter $B_\mu$ is chosen we estimate the function $\hat\mu$ from all data and use it in the second step to select $B_R$ and $B_V$ for smoothing the covariance kernels. We choose these smoothing parameters only while smoothing the lag-zero covariance.   The reason behind this is that we expect the same smoothness for higher order lags and the selection of the parameters on only one covariance kernel reduces the computational cost, which would otherwise become substantial.
We again employ K-fold cross-validation. Denote $\hat{R_0}^{(-k), B_R^0}$ the estimate of $R_0$ obtained by the smoother \eqref{eq:local_LS_cov_lag_0} calculated from the data without the partition $k$ and using the candidate smoothing parameter $B_R^0$.
The smoothing parameters $B_R$ is selected by minimizing the following loss:
\begin{equation}
\label{eq:CV_minimisation_B_R}
B_R = \argmin_{B_R^0} 
\frac{1}{K} \sum_{k=1}^K
\sum_{t\in\mathcal{T}_k}
\sum_{i,j=1}^{N_t}
\left\{
\left(Y_{ti} - \hat{\mu}(x_{ti})\right)
\left(Y_{tj} - \hat{\mu}(x_{tj})\right)
- \hat{R_0}^{(-k),B_R^0}(x_{ti},x_{tj})
\right\}^2.
\end{equation}

To select the smoothing parameter $B_V$, we denote $\hat{V}^{(-k), B_V^0}$ the estimate of the diagonal of $R_0(\cdot,\cdot)$ including the ridge contamination, from the data except the partition $k$ and using the candidate smoothing parameter $B_V^0$. The parameter $B_V$ is selected by minimizing the following loss:
\begin{equation}
\label{eq:CV_minimisation_B_V}
B_V = \argmin_{B_V^0}
\frac{1}{K} \sum_{k=1}^K
\sum_{t\in\mathcal{T}_k}
\sum_{i=1}^{N_t}
\left\{
\left(Y_{ti} - \hat{\mu}(x_{ti})\right)^2
- \hat{V}(x_{ti})^{(-k), B_V^0}
\right\}^2
\end{equation}

% Analogously, denote $\widehat{\sigma^2}^{(-k), B_V^0}$ the estimate of the measurement error obtained by \eqref{eq:local_LS_diagonal} and \eqref{eq:local_LS_est_sigma2} from the data without the partition $k$ and using the candidate smoothing parameter $B_V^0$.
%We select the smoothing parameters $B_R$ and $B_V$ by minimizing the following loss:
%
%\begin{multline}
%\label{eq:CV_minimisation_B_B_and_B_V}
%\left(B_R,B_V\right) = \argmin_{B_R^0, B_V^0}
%\frac{1}{K} \sum_{k=1}^K \sum_{t=1}^T
%\Bigg\{  \sum_{\substack{i,j=1 \\ i\neq j }}^{N_t^{(k)}}
%\left\{
%\left(Y_{ti} - \hat{\mu}(x_{ti}^{(k)})\right)
%\left(Y_{tj} - \hat{\mu}(x_{tj}^{(k)})\right)
%- \hat{R_0}^{(-k),B_R^0}(x_{ti}^{(k)},x_{tj}^{(k)})
%\right\}^2
%+\\+
%\sum_{i=1}^{N_t^{(k)}}
%\left\{
%\left(Y_{ti} - \hat{\mu}(x_{ti}^{(k)})\right)^2
%- \hat{R_0}^{(-k),B_R^0}(x_{ti}^{(k)},x_{tj}^{(k)}) - \widehat{\sigma^2}^{(-k), B_V^0}
%\right\}^2
% \Bigg\}
% .
%\end{multline}

Once the minimizers $B_R$ and $B_V$ have been found, we construct the estimate of the lag-zero covariance kernel $\hat{R_0}$ and the measurement error $\widehat{\sigma^2}$ from the full data. The bandwidth parameter $B_R$ will be used for estimation of the spectral density because we expect the same degree of spatial smoothness for spectral density kernels over all frequencies.

To numerically solve the optimization problems \eqref{eq:CV_minimisation_B_mu}, \eqref{eq:CV_minimisation_B_R}, and \eqref{eq:CV_minimisation_B_V} we use MATLAB's implementation of the Bayesian optimisation algorithm (BayesOpt). A review of this algorithm can be found for example in  \citet{mockus2012bayesian}.

\subsection{Selection of Bartlett's span parameter $\QBartlett$}
\label{subsec:selection of Q_Bartlett}

The selection of the parameter $\QBartlett$, i.e. the number of lags taken into account when estimating the dynamics, is a challenging problem in general. Selection rules for the bandwidth parameter for smoothing in the frequency domain, which is equivalent to Bartlett's estimate as explained in Subsection \ref{subsec:Spectral Density Kernels Estimation}, is reviewed in  \citet{fan2008nonlinear} for the case of one-dimensional time-series. The selection of the parameter $\QBartlett$, or equivalently the bandwidth parameter for frequency domain smoothing, has nevertheless not been explored for the case of functional time-series. Neither  \citet{panaretos2013fourier} nor \citet{hormann2015dynamic} provide data-dependent criteria, but instead rely on a prior choices based on asymptotic considerations.

The selection of the tuning parameter $\QBartlett$ is better studied in a related problem --- the estimation of the long-run covariance, which is in fact the value of the spectral density at frequency $\omega=0$. The long-run covariance can be estimated by the Bartlett's formula \eqref{eq:bartletts_estimate_full_observations} for frequency $\omega=0$. Data adaptive selection procedures for the tuning parameter $\QBartlett$ have been suggested in this context by  \citet{rice2017plug} and \citet{horvath2016adaptive_bandwidth}.

However, it is unclear how to incorporate the sparse sampling scheme to the above-cited rules.
To address this issue, we run a number of numerical experiments, simulating datasets from a couple of smooth functional time-series, and estimating the spectral density with a varying value of the parameter $\QBartlett$. By investigating the estimation error, we propose guidelines on selecting $\QBartlett$ in the form of a rule of thumb. The details on the simulation study are reported in Section \ref{subsec:simulations_estimation of spectral density}, the results are recorded in Section \ref{appendix_results:determination_Q_Bartlett}, and the proposed rule of thumb is stated in formula \eqref{eq:Q_Bartlett decision rule}.

\subsection{Representation of Functional Data}
\label{subsec:basis representation}

In the classical functional data analysis, one typically works with the functional data expressed with respect to a given finite (but possibly large) fixed basis. The usual choice is B-splines, Fourier basis, or wavelets. Throughout this article (in simulations and the data analysis) we choose to work with the B-spline basis of order $3$ because B-splines are efficient in expressing smooth functions  \citep{ramsay2007applied}.

A useful feature of the B-spline basis is the interpolation capability \citep{ramsay2007applied} which we benefit from. The smoother based estimators introduced in Sections \ref{subsec:nonparam_estimation} and \ref{subsec:Spectral Density Kernels Estimation} require to perform the smoothing at every point of $[0,1]$ or $[0,1]^2$. Therefore one has to choose a grid where the smoother is to be calculated.
To mitigate the computational time, we want to avoid executing the smoother on a very dense grid. Therefore we evaluate the smoother on a grid with a moderate number of points. Specifically, we operate with the equidistant grid with $21$ and $21\times 21$ points for functions and 2-dimensional kernels respectively. Once the smoothing estimator is realized on this grid, the functional counterparts as functions on $[0,1]$ and kernels on $[0,1]^2$ are retrieved by the B-spline interpolation. This technique is in contrast to \citet{yao2005functional} who evaluate the smoother on the equidistant grid of size $51\times 51$ and treat the covariance kernel as a $51\times 51$ matrix and the functional data as vectors. Our simulations (not reported here) suggest that these two approaches have essentially the same statistical performance for smooth functional data. Indeed the stochastic estimation error dominates the numerical approximation error of the fully functional quantities. From the implementation point of view, the B-spline interpolation approach shortens the computational time, reduces the dimension of the data to be stored, and directly expresses the functional quantities with respect to a basis. 

Once the smoother-based estimates of the model dynamics expressed in the B-spline basis, we assume that the functional data themselves are expressed within the fixed finite B-spline basis. Of course, the functional data are not directly observed and thus we treat the unknown basis coefficient as latent variables to be retrieved. Using the calculus for functions and operators expressed with respect to a basis \citep{ramsay2007applied}, the functional recovery formulae of Section \ref{subsec:functional_data_recovery} can be rewritten and their evaluation is based on vector and matrix manipulations, albeit in a much lower dimensional setting.

%Fixing an (orthonormal) basis functions $\phi_1,\dots,\phi_n$ in space $\mathcal{H}$, one can approximate a function $f\in\mathcal{H}$ by
%$ f\approx \sum_{j=1}^n \langle f,\phi_j\rangle\phi_j $ and thus store in memory just the vector $(f_1,\dots,f_n)$ of coefficients $f_j := \langle f,\phi_j\rangle$. Likewise, consider a linear operator $A:\mathcal{H}\to\mathcal{H}$. Its output for $f\approx \sum_{j=1}^n \langle f,\phi_j\rangle\phi_j$, and thus the operator itself, can be approximated by
%$A f \approx \sum_{j=1}^n \sum_{k=1}^n f_j \langle A\phi_j, \phi_k \rangle \phi_k $. Therefore we store the approximation of $A$ as a matrix $(a_{jk})_{j,k=1}^n$ with elements $a_{jk}=\langle A\phi_j,\phi_k\rangle$. The calculus with the approximations $(f_j)_{j=1}^n$ and $(a_{jk})_{j,k=1}^n$ is then essentially vector and matrix algebra in Euclidean spaces \citep{ramsay2007applied}.

\subsection{Forecasting}
\label{subsec:forecasting}

A natural next step to consider, and indeed one of the main reasons why one may be interested in recovering the functional time-series dynamics, is that of forecasting. In this section, we comment on how the forecasting problem naturally fits into the functional data recovery framework introduced in Section \ref{subsec:functional_data_recovery}.

Assume that we are given sparse data $\{Y_{tj}: 1\leq j\leq N_t, 1\leq t \leq T\}$ and we wish to forecast the functional datum $X_{T+r}$ for $r\in\mathbb{N}$ as well as to quantify the uncertainty of the forecast.
We define the random element $\mathbb{X}_{T+r} = [X_1,\dots,X_T,X_{T+1},\dots,X_{T+r}] \in\mathcal{H}^{T+r}$. If the forecasts for the intermediate data $X_{T+1},\dots,X_{T+r-1}$ are not of interest, we may delete these elements and naturally alter the explained method below. Nevertheless, we opt to explain the approach for forecasting up to the time $T+r$ simultaneously.

We utilize the notation introduced in Subsection \ref{subsec:functional_data_recovery}.
By extending the formulae \eqref{eq:mathbb_X_mu} and \eqref{eq:mathbb_X_cov} for $t=1,\dots,T+r$ we obtain the law of $\mathbb{X}_{T+r}$, i.e. the joint law of $X_1,\dots,X_{T+r}$, and can calculate their conditional distribution given the observed data $\mathbb{Y}_T$. In particular, by taking $s=T+r$ in the equations \eqref{eq:BLUP_Xs_theoretical}, \eqref{eq:confidence_band_Xs_population}, and \eqref{eq:simul_confidence_band_population} we obtain the forecast, the pointwise confidence band, and the simultaneous confidence band respectively for the functional datum $X_{T+r}$. In practice, we substitute the unknown population level quantities by their empirical estimators. Therefore, by taking $s=T+r$ in the equations \eqref{eq:BLUP_Xs_empirical}, \eqref{eq:confidence_band_Xs_empirical}, and \eqref{eq:simul_confidence_band_empiric} we obtain the forecast, the (asymptotic) pointwise confidence band, and the (asymptotic) simultaneous confidence band for $X_{T+r}$.

%%%%%%%%%%%%%%%%%%%%%%%%%%%%%%%%%%%%%%%%%%%%%%%%%%%%%%%%%%%%%%%%%%%%%%%%%%%%%%%%%%%%%%%%%%%%%%%%%%%%%
%%%%%%%%%%%%%%%%%%%%%%%%%%%%%%%%%%%%%%%%%%%%%%%%%%%%%%%%%%%%%%%%%%%%%%%%%%%%%%%%%%%%%%%%%%%%%%%%%%%%%

\section{Proofs of Formal Statements}
\label{sec:proofs}

\subsection{Proof of Theorem \ref{thm:mean_and_autocov_function} }

We start with the smoother for the common mean function $\mu(\cdot)$.
Its estimator $\hat\mu(x)$, the minimizer of \eqref{eq:local_LS_for_mu}, explicitly:
\begin{equation}\label{eq:explicit_formula_mu}
\hat\mu(x) = \frac{Q_0 S_2 - Q_1 S_1}{S_0 S_2 - S_1^2}, 
\end{equation}
where
$$ S_r = \frac{1}{T} \sum_{t=1}^T \sum_{j=1}^{N_t} \left(\frac{x_{tj}-x}{B_\mu}\right)^r \frac{1}{B_\mu} K\left(\frac{x_{tj}-x}{B_\mu}\right), \qquad r=0,1,2,$$
$$ Q_r = \frac{1}{T} \sum_{t=1}^T \sum_{j=1}^{N_t} \left(\frac{x_{tj}-x}{B_\mu}\right)^r Y_{tj} \frac{1}{B_\mu} K\left(\frac{x_{tj}-x}{B_\mu}\right), \qquad r=0,1.$$
All of the above quantities are functions of $x\in[0,1]$ and all of the operations are to be understood in the pointwise sense, and this includes the division operation. In Lemma \ref{lemma:asymptotics_of_S_for_mu} and Lemma \ref{lemma:asymptotics_of_Q_for_mu}  we determine the asymptotic behaviour of $S_r$ and $Q_r$, respectively.

\begin{lemma}\label{lemma:asymptotics_of_S_for_mu}
Under \ref{assumption:B.1}, \ref{assumption:B.2} and \ref{assumption:B.7},
for $r=0,1,2$
$$ \sup_{x\in[0,1]} \left| S_r - M_{[S_r]} \right| = \Op\left(  \frac{1}{\sqrt{T}B_\mu} +B^2_\mu\right) $$
where $M_{[S_0]}=\Ez{ N}g(x), M_{[S_1]}=0, M_{[S_2]}= \Ez{N}\sigma^2_K g(x)$ and $\sigma^2_K = \int v^2 K(v) \D{v}$.
\end{lemma}
\begin{proof}
We have the usual bias-variance decomposition
$$\Ez{ \sup_{x\in[0,1]} \left| S_r - M_{[S_r]} \right| }\leq
\sup_{x\in[0,1]} \left| \Ez{ S_r } - M_{[S_r]} \right| +
\Ez{ \sup_{x\in[0,1]} \left| \Ez{ S_r } - S_r \right| }.$$
For the bias term, by using the Taylor expansion to order 2 it is easy to show the formulae for $M_{[S_r]}, r=0,1,2$ as well as that
$\Ez{ S_r } = M_{[S_r]} + O(B_\mu^2)$ where the remainder of the Taylor expansion is uniform in $x\in[0,1]$.
Hence
\begin{equation}\label{eq:asymptotics_of_S_for_mu_bias}
\sup_{x\in[0,1]} \left| \Ez{ S_r } - M_{[S_r]} \right| = O(B_\mu^2).
\end{equation}
For the stochastic term, it will be useful to employ the Fourier transform. The inverse Furrier transform of the function $u\mapsto K(u)u^r$ is defined as $\zeta_r(t) = \int e^{-\I ut}K(u)u^r \D{u}$. Therefore we may write
$$
w_{tj} \left( \frac{x_{tj}-x}{B_\mu} \right)^r
= \frac{1}{2\pi B_\mu}\int e^{\I u (x_{tj}-x)/hu_\mu} \left(\frac{x_{tj}-x}{B_\mu}\right)^r \zeta_r(u) \D{u}
=
\frac{1}{2\pi}\int e^{\I v (x_{tj}-x)}(x_{tj}-x)^r \zeta_r(v B_\mu) \D{v} .
$$
Define
\begin{equation}\label{eq:asymptotics_of_S_for_mu_bias_phi}
\phi_{r}(v) = \frac{1}{T }\sum_{t=1}^T \sum_{j=1}^{N_t} e^{\I v x_{tj}}(x_{tj}-x)^r
\end{equation}
and thus we can write
$$ S_r(x) = \frac{1}{2\pi} \int \phi_r(v) e^{-\I x v} \zeta_r(v B_\mu)\D{v} .$$
Thanks to the independence of $\{N_t\}$ and $\{x_{tj}\}$ we can bound the variance of $\phi_{S_r}(x)$
\begin{multline*}
\var(  \phi_{S_r}(x) ) \leq 
\frac{1}{T} \var \left\{  \sum_{j=1}^{N_1} e^{\I v x_{1j} (x_{1j}-x)^r} \right\}
\leq
\frac{1}{T} \Ez{ E\left[ \left\{ \sum_{j=1}^{N_1} e^{\I v x_{1j}}(x_{1j}-x)^r \right\}^2 \mid N_1 \right]}
\leq\\\leq
\frac{1}{T} \Ez{ E\left[ \left\{ \sum_{j=1}^{N_1} \left| e^{\I v x_{1j}} \right|^2 \right\}\left\{ \sum_{j=1}^{N_1} (x_{1j}-x)^{2r} \right\} \mid N_1 \right]}
\leq
\frac{1}{T} E\left[ \Ez{N}  E \left\{ \sum_{j=1}^{N_1} (x_{1j}-x)^{2r} \right\} \mid N_1 \right]
\leq\\\leq
\frac{\Ez{N}}{T} \E\left[ (x_{11}-x)^{2r} \right] \leq \frac{\Ez{N}}{T}.
\end{multline*}
Thus
\begin{multline}
\label{eq:asymptotics_of_S_for_mu_var}
\E\left\{ \sup_x \left| S_r(x) - \Ez{ S_r(x) } \right| \right\} \leq
\frac{1}{2\pi} \int \Ez{\left| \phi_r(v)-\E \phi_{S_r}(v) \right|} |\zeta_r(v B_\mu)| \D{v}
\leq\\\leq
\frac{1}{2\pi} \int \sqrt{\var(  \phi_r(x) ) } |\zeta_r(v B_\mu)| \D{v}
\leq
\frac{ \int |\zeta_r(u)|\D{u}}{2\pi } \frac{\Ez{N}}{\sqrt{T} B_\mu} = O\left(\frac{1}{\sqrt{T}B_\mu}\right)
.
\end{multline}

The proof is concluded by combining \eqref{eq:asymptotics_of_S_for_mu_bias} and \eqref{eq:asymptotics_of_S_for_mu_var}, and by the observation that $\Ez{ |Z_n| } = O(a_n)$ implies $Z_n = \Op(a_n)$ for an arbitrary sequence of random variables $Z_n$ and a sequence of constants $a_n$.
%Also note that $\Op(B_\mu^2)$ is faster than $\Op(1/(\sqrt{T} B_\mu))$.
\end{proof}

\begin{lemma}\label{lemma:asymptotics_of_Q_for_mu}
Under \ref{assumption:B.1} --- \ref{assumption:B.3} and \ref{assumption:B.7},
for $r=0,1$
$$ \sup_{x\in[0,1]} \left| Q_r - M_{[Q_r]} \right| = \Op\left(\frac{1}{\sqrt{T}B_\mu} +B^2_\mu\right) $$
where $M_{[Q_0]} = \Ez{N} \mu(x) g(x)$ and $M_{[Q_1]} = 0$.
\end{lemma}
\begin{proof}
The proof of Lemma \ref{lemma:asymptotics_of_Q_for_mu} follows the same ideas as that of Lemma \ref{lemma:asymptotics_of_S_for_mu}. We use the bias variance decomposition and a Taylor expansion to order 2 to derive the analogous results as in \eqref{eq:asymptotics_of_S_for_mu_bias} as well as the formulae for $M_{[Q_0]}(x)$ and $M_{[Q_1]}(x)$. We then define
\begin{equation}\label{eq:asymptotics_of_S_for_mu_bias_varphi}
\varphi_r(v) = \frac{1}{T} \sum_{t=1}^T \sum_{j=1}^{N_t} e^{\I v x_{tj}} (x_{tj}-x)^r Y_{tj}
\end{equation}
in analogy to \eqref{eq:asymptotics_of_S_for_mu_bias_phi}. Thus we can write
$$ Q_r(x) = \frac{1}{2\pi} \int \varphi_r(v) e^{-\I x v} \zeta_r(v B_\mu) \D{v}. $$
It remains to bound the variance of \eqref{eq:asymptotics_of_S_for_mu_bias_varphi}. However, the temporal dependence among $Y_{tj}$ must be now taken into account.
First of all remark that for an arbitrary stationary time-series $\{Z_t\}$ with a summable autocovariance function $\rho_Z(\cdot)$, one has:
\begin{equation}
\label{eq:proof_arbitrary_stationary_timeseries_variance}
\var\left(\frac{1}{T} \sum_{t=1}^T Z_t \right) = \frac{1}{T} \sum_{h=-T+1}^{T-1} \rho_Z(h) \left(1-\frac{|h|}{T}\right) \leq \frac{1}{T} \sum_{h=-\infty}^\infty |\rho_Z(h)| .
\end{equation}
Define $Z_t = \sum_{j=1}^{N_t} e^{\I v x_{tj}} (x_{tj}-x)^r Y_{tj}$. This sequence of real random variables constitutes a stationary time-series. By conditioning on $N_t$ and $x_{tj}$, and applying the law of total covariance, we can bound the autocovariance of $\{Z_t\}$ by $|\rho_Z(h)| \leq \max_{x,y} |R_h(x,y)|$ for $h\neq0$. For $h=0$, the bound is augmented by $\sigma^2$ due to the measurement error but this changes nothing  on the summability.
The autocovariance function is summable thanks to the assumption \eqref{eq:assumption_stationarity_summability_of_autocovariance} and we conclude that $ \var \varphi_r(v) = O(1/T)$.
By repeating the same steps as in \eqref{eq:asymptotics_of_S_for_mu_var} we obtain
$$ E \left\{ \sup_{x\in[0,1]} \left| S_r(x) - \Ez{ S_r(x) } \right| \right\}  = O\left(\frac{1}{\sqrt{T}B_\mu} \right)$$
which completes the proof.
\end{proof}

\begin{proof}[Proof of the first part of Theorem \ref{thm:mean_and_autocov_function}]
By combining Lemma \ref{lemma:asymptotics_of_S_for_mu}, Lemma \ref{lemma:asymptotics_of_Q_for_mu}, the formula \eqref{eq:explicit_formula_mu}, and the uniform version of Slutsky's theorem, we obtain the rate \eqref{eq:thm:mean_function:mu}.
\end{proof}

Now we turn our attention to the estimation of the lag-$0$ covariance and lag-$h$ autocovariance kernels. We include the proof only for $h\neq 0$. For $h=0$ one has to exclude the diagonal to evade the measurement errors but the proof is essentially the same. It is possible to explicitly express the minimizer to \eqref{eq:local_LS_cov_lag_nu} (cf. \citet{li2010uniform}). The general principles of the explicit formula deviation are also commented on for the case of spectral density estimation in Section \ref{subsec:proof_of_Thm_2}, which uses similar deviation steps as the estimator of lagged autocovariance kernels.
The explicit formula yields
\begin{equation}
\label{eq:autocov_explicit_formula}
 \hat{R}_h(x,y) = \left( 
\mathscr{A}^{(h)}_1 Q_{00}^{(h)} -
\mathscr{A}^{(h)}_2 Q_{10}^{(h)} -
\mathscr{A}^{(h)}_3 Q_{01}^{(h)}
 \right) \left(\mathscr{B}^{(h)}\right)^{-1} ,
\end{equation}
where $|h|<T$ and
\begin{align*}
\mathscr{A}^{(h)}_1 &= S_{20}^{(h)}S_{02}^{(h)}-\left(S_{11}^{(h)}\right)^2, \qquad 
\mathscr{A}^{(h)}_2  = S_{10}^{(h)}S_{02}^{(h)}-S_{01}^{(h)}S_{11}^{(h)}, \\
\mathscr{A}^{(h)}_3  &= S_{01}^{(h)}S_{20}^{(h)}-S_{10}^{(h)}S_{11}^{(h)}, \quad
\mathscr{B}^{(h)} = \mathscr{A}_1^{(h)} S_{00}^{(h)} - \mathscr{A}_2^{(h)} S_{10}^{(h)} - \mathscr{A}_3^{(h)} S_{01}^{(h)}, \\
S_{pq}^{(h)} &= \frac{1}{T-|h|} \sum_{t=\max(1,1-h)}^{\max(T,T-h)}
\stackrel[j\neq k \text{ if } h=0]{}{\sum_{j=1}^{N_{t+h}} \sum_{k=1}^{N_t}}
\left( \frac{x_{t+h,j}-x}{B_R} \right)^p
\left( \frac{x_{tk}-y}{B_R} \right)^q
\times\\
&\times
\frac{1}{B_R^2}
K\left( \frac{x_{t+h,j}-x}{B_R} \right)
K\left( \frac{x_{tk}-y}{B_R} \right), \\
Q_{pq}^{(h)} &= \frac{1}{T-|h|} \sum_{t=\max(1,1-h)}^{\max(T,T-h)}
\stackrel[j\neq k \text{ if } h=0]{}{\sum_{j=1}^{N_{t+h}} \sum_{k=1}^{N_t}}
G_{h,t}(x_{t+h,j}, x_{tk})
\left( \frac{x_{t+h,j}-x}{B_R} \right)^p
\left( \frac{x_{tk}-y}{B_R} \right)^q \times\\
&\times
\frac{1}{B_R^2}
K\left( \frac{x_{t+h,j}-x}{B_R} \right)
K\left( \frac{x_{tk}-y}{B_R} \right).
\end{align*}
All of the above terms are functions of $(x,y) \in [0,1]^2$ and all operations are understood the pointwise sense, including the pointwise inversion of $\left(\mathscr{B}^{(h)}\right)^{-1} = \left(\mathscr{B}^{(h)}(x,y)\right)^{-1}$.

We asses the uniform asymptotic behaviour of $S_{pq}^{(h)}$ and $Q_{pq}^{(h)}$ in Lemma
\ref{lemma:asymptotics_of_S_for_R} and Lemma \ref{lemma:asymptotics_of_Q_for_R}.

\begin{lemma}\label{lemma:asymptotics_of_S_for_R}
Under \ref{assumption:B.1}, \ref{assumption:B.2}, \ref{assumption:B.7} and \ref{assumption:B.8},
\begin{align}
\label{eq:lemma3:rate}
\Ez{ \sup_{x,y\in[0,1]} \left| S_{pq}^{(h)} - \E S_{pq}^{(h)} \right| }
&\leq
U \frac{1}{\sqrt{T-|h|}} \frac{1}{B_R^2} \\
\label{eq:lemma3:rate2}
\sup_{x,y\in[0,1]} \left| \E S_{pq}^{(h)} - M_{[S_{pq}]} \right|
&= O\left( B_R^2 \right)
\end{align}
where the constant $U$ is uniform for $0\leq p+q\leq 2$, $T\in\mathbb{N}$, $|h|<T$, and
\begin{equation}
\label{eq:lemma3:means}
\begin{gathered}
M_{[S_{00}^{(h)}]} = c_h g(x)g(y), \qquad
M_{[S_{01}^{(h)}]} = M_{[S_{10}^{(h)}]} = M_{[S_{11}^{(h)}]} = 0, \\
M_{[S_{20}^{(h)}]} =  M_{[S_{02}^{(h)}]} = c_h g(x)g(y) \sigma^2_K, \qquad
\sigma^2_K = \int v^2 K(v) \D{v},
\end{gathered}
\end{equation}
where $c_h = (\E N)^2$ for $h\neq 0$ and $c_0 = \E\{N(N-1)\}$.
Moreover, the convergence \eqref{eq:lemma3:rate2} is uniform in $h$.
\end{lemma}
\begin{proof}
Note the bias-variance decomposition of the estimation error
\begin{equation}\label{eq:lemma3:bias variance decomposition}
\Ez{ \sup_{x,y\in[0,1]} \left| S_{pq}^{(h)} - M_{[S_{pq}^{(h)}]} \right| } \leq
\Ez{ \sup_{x,y\in[0,1]} \left| S_{pq}^{(h)} - \Ez{ S_{pq}^{(h)} } \right| } +
\sup_{x,y\in[0,1]} \left| \Ez{ S_{pq}^{(h)} } - M_{[S_{pq}^{(h)}]} \right|
\end{equation}
Considering a Taylor expansion of order 2, it is easy to show that the formulae \eqref{eq:lemma3:means} and that the second term of \eqref{eq:lemma3:rate} is of order $O(B_R^2)$ uniformly in $h$ and $T$.

Taking the analogous steps as in the proof of Lemma \ref{lemma:asymptotics_of_S_for_mu} while using the Fourier transform of the function $(u,v) \mapsto K(u)K(v) u^p v^q$, one can prove that the first term on the right-hand side of \eqref{eq:lemma3:bias variance decomposition} are bounded by $1/(T-|h|)$.
\end{proof}

\noindent Now assume that the common mean function $\mu(\cdot)$ is known for the moment.
Thus formally define
\begin{multline*}
\tilde{Q}_{pq}^{(h)} = \frac{1}{T-|h|} \sum_{t=\max(1,1-h)}^{\max(T,T-h)}
\stackrel[j\neq k \text{ if } h=0]{}{\sum_{j=1}^{N_{t+h}} \sum_{k=1}^{N_t}}
\tilde{G}_{h,t}(x_{t+h,j}, x_{tk})
\left( \frac{x_{t+h,j}-x}{B_R} \right)^p
\left( \frac{x_{tk}-y}{B_R} \right)^q
\times\\\times
\frac{1}{B_R^2}
K\left( \frac{x_{t+h,j}-x}{B_R} \right)
K\left( \frac{x_{tk}-y}{B_R} \right)
\end{multline*}
where
\begin{equation}\label{eq:raw_covariance_known_mu}
\tilde{G}_{h,t}( x_{t+h,j}, x_{tk} ) = (Y_{t+h,j} - \mu(x_{t+h,j}))(Y_{tk} - \mu(x_{tk})).
\end{equation}
We analyse the asymptotics of $\tilde{Q}_{pq}^{(h)}$ in Lemma \ref{lemma:asymptotics_of_Q_for_R}.

\begin{lemma}\label{lemma:asymptotics_of_Q_for_R}
Under \ref{assumption:B.1} --- \ref{assumption:B.5} and \ref{assumption:B.8}
\begin{align}
\label{eq:lemma4:rate}
\Ez{ \sup_{x,y\in[0,1]} \left| \tilde{Q}_{pq}^{(h)} - \E \tilde{Q}_{pq}^{(h)} \right| }
&\leq U \frac{1}{\sqrt{T-|h|}}\frac{1}{B_R^2} \\
\label{eq:lemma4:rate2}
\sup_{x,y\in[0,1]} \left| \E \tilde{Q}_{pq}^{(h)} - M_{[Q_{pq}^{(h)}]} \right|
&= O\left( B_R^2 \right)
\end{align}
where the constant $U$ is uniform for $0\leq p+q\leq 2$, $T\in\mathbb{N}$, $|h|<T$, and
\begin{equation}
\label{eq:lemma4:means}
M_{[Q_{00}^{(h)}]} = c_h R_h(x,y) g(x)g(y), \qquad
M_{[Q_{01}^{(h)}]} = M_{[Q_{10}^{(h)}]} = 0,
\end{equation}
where $c_h = (\E N)^2$ for $h\neq 0$ and $c_0 = \E\{N(N-1)\}$.
Moreover, the convergence \eqref{eq:lemma4:rate2} is uniform in $h$.
\end{lemma}
\begin{proof}
Again, the bias-variance decomposition yields
$$
\Ez{ \sup_{x,y\in[0,1]} \left| \tilde{Q}_{pq}^{(h)} - M_{[Q_{pq}^{(h)}]} \right| }
\leq
\Ez{ \sup_{x,y\in[0,1]} \left| \tilde{Q}_{pq}^{(h)} - \Ez{ \tilde{Q}_{pq}^{(h)} } \right| } +
\sup_{x,y\in[0,1]} \left| \Ez{ \tilde{Q}_{pq}^{(h)} } - M_{[Q_{pq}^{(h)}]} \right|
$$
By taking a Taylor expansion of order 2, it is again straightforward to show that the formulae \eqref{eq:lemma4:means} and that the second term of \eqref{eq:lemma4:rate} is of order $O(B_R^2)$ uniformly in $h$ and $T$.

To treat the first term on the right-hand side of \eqref{eq:lemma4:rate}, we define the Fourier transform of the function $(\alpha,\beta) \mapsto K(\alpha)\alpha K(\beta)\beta$ as $\zeta_{pq}(u,v) = \iint e^{-\I (u\alpha+v\beta)} K(\alpha)\alpha^p K(\beta)\beta^q d\alpha d\beta$. Thus we may write
\begin{multline*}
\left( \frac{x_{t+h,j}-x}{B_R} \right)^p
\left( \frac{x_{tk}-y}{B_R} \right)^q
\frac{1}{B_R^2}
K\left( \frac{x_{t+h,j}-x}{B_R} \right)
K\left( \frac{x_{tk}-y}{B_R} \right)
=\\=
\frac{1}{(2\pi)^2 B_R^2} \iint
\exp\left\{ \I \left( \frac{x_{t+h,j}-x}{B_R} \right) u  \right\}
\exp\left\{ \I \left( \frac{x_{tk}-y}{B_R} \right) v  \right\}
\zeta_{pq}(u,v)\D{u}\D{v}
=\\=
\frac{1}{(2\pi)^2} \iint
%\exp\left\{ \I \left( x_{t+h,j}-x \right) \tilde{u}  \right\}
%\exp\left\{ \I \left( x_{tk}-y \right) \tilde{v}  \right\}
e^{ \I ( x_{t+h,j}-y ) \tilde{u} }
e^{ \I ( x_{tk}-y ) \tilde{v} }
\zeta_{pq}(B_R \tilde{u},B_R \tilde{v})d\tilde{u}d\tilde{v}
\end{multline*}
Define
$$ \varphi_{pq}^{(h)} = \varphi_{pq}^{(h)}(u,v,x,y) =
\frac{1}{T-|h| } \sum_{t=\max(1,1-h)}^{\max(T,T-h)}
\stackrel[j\neq k \text{ if } h=0]{}{\sum_{j=1}^{N_{t+h}} \sum_{k=1}^{N_t}}
e^{ \I ( x_{t+h,j}-x ) u }
e^{ \I ( x_{tk}-y ) v }
\tilde{G}_{h,t}(x_{t+h,j}, x_{tk})
$$
and write
$$ \tilde{Q}_{pq}^{(h)} =
\frac{1}{(2\pi)^2} \iint \varphi_{pq}^{(h)} \zeta_{pq}(B_R u,B_R v) \D{u}\D{v}
$$
Analogously to \eqref{eq:asymptotics_of_S_for_mu_var}, it now remains to analyse the variance of $\varphi_{pq}^{(h)}$. Define the following stationary time-series
$$ Z_t^{(h)} =
\stackrel[j\neq k \text{ if } h=0]{}{\sum_{j=1}^{N_{t+h}} \sum_{k=1}^{N_t}}
e^{ \I ( x_{t+h,j}-x ) u }
e^{ \I ( x_{tk}-y ) v }
\tilde{G}_{h,t}(x_{t+h,j}, x_{tk}).
$$
As in the proof of Lemma \ref{lemma:asymptotics_of_Q_for_mu} we want to bound the sum of the autocovariance function $\sum_{\xi\in\mathbb{Z}} |\rho_{Z^{(h)}}(\xi)|$ but the bound must be uniform in $h$. By conditioning on $N_t$ and $x_{tj}$, and applying the law of total covariance, the $\xi$-lag autocovariance $\rho_{Z^{(h)}}(\xi)$ can be bounded by
\begin{multline}\label{eq:autocovZnu bound}
\left| \rho_{Z^{(h)}}(\xi) \right| = \left| \cov( Z_{t+\xi}, Z_t ) \right|
\leq\\\leq
(\Ez{N})^2 
\sup_{x_1,x_2,x_3,x_4 \in [0,1]} \Big|
\cov\Big\{
(X_{t+\xi+h}(x_1) - \mu(x_1)) (X_{t+\xi}(x_2) - \mu(x_2))
(X_{t+h}(x_3) - \mu(x_3))( X_t(x_4) - \mu(x_4))
\Big\}
\Big|
=\\=
(\Ez{N})^2
\sup_{x_1,x_2,x_3,x_4 \in [0,1]} \Big|
\cov\Big\{
(X_{\xi+h}(x_1) - \mu(x_1)) (X_{\xi}(x_2) - \mu(x_2))
(X_{h}(x_3) - \mu(x_3))( X_0(x_4) - \mu(x_4))
\Big\}
\Big|
\end{multline}
for $\xi \notin \{-h,0,h\}$. For $\xi \in \{-h,0,h\}$, the bound is augmented by $\sigma^2$ but this changes nothing as to the summability with respect to $\xi\in\mathbb{Z}$.

Using the formula for the 4-th order cumulant of centred random variables \citep[p. 36]{rosenblatt2012stationary}, we express the covariance on the right-hand side of \eqref{eq:autocovZnu bound} as
\begin{multline}\label{eq:autocovZnu expression cum}
\cov\left(
(X_{\xi+h}(x_1) - \mu(x_1)) (X_{\xi}(x_2) - \mu(x_2)),
(X_{h}(x_3) - \mu(x_1))( X_0(x_4) - \mu(x_1))
\right) =\\=
\cum\left(
	X_{\xi+h}(x_1) - \mu(x_1),
	X_{\xi}(x_2) - \mu(x_2),
	X_{h}(x_3) - \mu(x_3),
	X_0(x_4) - \mu(x_4),
\right)
+\\+
R_\xi(x_1,x_3) R_\xi(x_2,x_4) +
R_{\xi+h}(x_1,x_4) R_{\xi-h}(x_2,x_3).
\end{multline}
Taking the absolute value and the supremum, the sum of \eqref{eq:autocovZnu bound} with respect to $\xi$ is bounded thanks to the fact that the cumulant on the right-hand side of \eqref{eq:autocovZnu expression cum} is summable by \ref{assumption:B.5} and the autocovariances are summable by \eqref{eq:assumption_stationarity_summability_of_autocovariance}. Moreover the sum is bounded uniformly in $h$.

Therefore
$$\var\left( Q_{pq}^{(h)} \right) \leq \frac{1}{T-h} \sum_{\xi\in\mathbb{Z}} |\rho_{Z^{(h)}}(\xi)| \leq U \frac{1}{T-h}$$
where the constant $U$ is independent of $h$. Observing that $\iint  \zeta_{pq}(B_R u,B_R v) \D{u}\D{v} = O(B_R^2)$ concludes the proof of the bound \eqref{eq:lemma4:rate}.
\end{proof}

In the following lemma we modify the previous result for the raw covariances $G_{h,t}$ instead of $\tilde{G}_{h,t}$.

\begin{lemma}\label{lemma:asymptotics_of_Q_for_R_in_P}
Under \ref{assumption:B.1} --- \ref{assumption:B.5}, \ref{assumption:B.7} and \ref{assumption:B.8},
for $h\in\mathbb{Z}$ and $0\leq p+q \leq 2; p,q\in\mathbb{N}_0$
$$ Q_{pq}^{(h)} = M_{[Q_{pq}]} + \Op \left( \frac{1}{\sqrt{T}} \frac{1}{B_R^2} +B_R^2 \right) $$
uniformly in $x,y\in[0,1]$.
\end{lemma}
\begin{proof}
We follow the lines of the discussion at the end of the proof of \citet[Theorem 1]{yao2005functional}. Consider a generic raw covariance $G_{h,t}(x,y) = \left(X_{t+h}(x) - \hat{\mu}(x) \right)\left( X_{t}(y) - \hat{\mu}(y) \right)$ and its counterpart $\tilde{G}_{h,t}(x,y) = \left(X_{t+h}(x) - \mu(x) \right)\left( X_{t}(y) - \mu(y) \right)$.
They can be related to each other by the expansion:
\begin{multline*}
G_{h,t}(x,y) = \tilde{G}_{h,t}(x,y)
+ \left( X_{t+h}(x) - \mu(x) \right)\left( \mu(y) - \hat{\mu}(y) \right)
+\\+
 \left( \mu(x) - \hat{\mu}(x) \right) \left( X_{t}(y) - \mu(y) \right)
+ \left( \mu(x) - \hat{\mu}(x) \right) \left( \mu(y) - \hat{\mu}(y) \right).
\end{multline*}
By \eqref{eq:thm:mean_function:mu}, the difference of $G_{h,t}(x,y)$ and $\tilde{G}_{h,t}(x,y)$ is of order $\Op\left( \frac{1}{\sqrt{T}} \frac{1}{B_\mu} \right)$ which is negligible with respect to the rate $\Op\left( \frac{1}{\sqrt{T}} \frac{1}{B_R^2} \right)$ from Lemma \ref{lemma:asymptotics_of_Q_for_R}.
\end{proof}

\begin{proof}[Proof of the second part of Theorem \ref{thm:mean_and_autocov_function}]
Combining the results of Lemma \ref{lemma:asymptotics_of_Q_for_mu} and Lemma \ref{lemma:asymptotics_of_Q_for_R_in_P}, we obtain the following uniform convergence rates:
\begin{align*}
\mathscr{A}_1^{(h)} &= \left[ c_h g(x) g(y) \sigma^2_K \right]^2 + \Op\left(\frac{1}{\sqrt{T}} \frac{1}{B_R^2} +B_R^2\right), \\
\mathscr{A}_2^{(h)} &= \Op\left(\frac{1}{\sqrt{T}} \frac{1}{B_R^2} +B_R^2\right),\\
\mathscr{A}_3^{(h)} &= \Op\left(\frac{1}{\sqrt{T}} \frac{1}{B_R^2} +B_R^2\right),\\
\mathscr{B}^{(h)} &= \left[ c_h g(x) g(y)\right]^3  \left(\sigma^2_K \right)^2
+ \Op\left(\frac{1}{\sqrt{T}} \frac{1}{B_R^2} +B_R^2\right).
\end{align*}
The numerator of the ratio \eqref{eq:autocov_explicit_formula} exhibits the following uniform convergence
$$
\mathscr{A}^{(h)}_1 Q_{00}^{(h)} -
\mathscr{A}^{(h)}_2 Q_{10}^{(h)} -
\mathscr{A}^{(h)}_3 Q_{01}^{(h)}
= \left[ c_h g(x) g(y)\right]^3  \left(\sigma^2_K \right)^2
R_h(x,y)
+ \Op\left(\frac{1}{\sqrt{T}} \frac{1}{B_R^2} +B_R^2\right) $$ 
and therefore we have proven the convergence rate for the autocovariance kernel estimator
$$ \hat{R}_h(x,y) = R_h(x,y) + \Op\left(\frac{1}{\sqrt{T}} \frac{1}{B_R^2} +B_R^2\right) $$
uniformly in $x,y\in[0,1]$.
\end{proof}

Finally we turn to the estimation of the measurement error variance $\sigma^2$.
The minimizer of the local quadratic smoother \eqref{eq:local_LS_est_sigma2_minimisation_V} can be expressed explicitly as
\begin{equation}
\label{eq:explicit_formula_barR}
\end{equation}
$$ \bar{R}_0(x) = \left( \bar{\mathscr{A}}_1 \bar{Q}_{0} - \bar{\mathscr{A}}_2 \bar{Q}_{1} - \bar{\mathscr{A}}_3 \bar{Q}_{2} \right) \bar{\mathscr{B}}^{-1}$$
where
\begin{align*}
\bar{\mathscr{A}}_1 &= \bar{S}_{2}\bar{S}_{4}-\left(\bar{S}_{3}\right)^2, \qquad 
\bar{\mathscr{A}}_2  = \bar{S}_{1}\bar{S}_{4}-\bar{S}_{2}\bar{S}_{3}, \qquad 
\bar{\mathscr{A}}_3  = \bar{S}_{2}\bar{S}_{2}-\bar{S}_{1}\bar{S}_{3}, \\
\bar{\mathscr{B}} &= \bar{\mathscr{A}}_1 \bar{S}_{0} - \bar{\mathscr{A}}_2 \bar{S}_{1} - \bar{\mathscr{A}}_3 \bar{S}_{2}, \\
\bar{S}_{r} &= \frac{1}{T} \sum_{t=1}^T \sum_{j\neq k}
\left( \frac{P(x_{tj},x_{tk}) - x}{B_R} \right)^r
\frac{1}{B_R^2}
K\left( \frac{x_{tj}-x}{B_R} \right)
K\left( \frac{x_{tk}-x}{B_R} \right), \\
\bar{Q}_{r} &= \frac{1}{T} \sum_{t=1}^T \sum_{j\neq k}
G_{0,t}(x_{tj}, x_{tk})
\left( \frac{P(x_{tj},x_{tk}) - x}{B_R} \right)^r
\frac{1}{B_R^2}
K\left( \frac{x_{tj}-x}{B_R} \right)
K\left( \frac{x_{tk}-x}{B_R} \right). \\
\end{align*}
All of the above quantities are understood as functions of $x\in[0,1]$ and all operations are considered pointwise, including the pointwise inversion $\bar{\mathscr{B}}^{-1} = (\bar{\mathscr{B}}(x))^{-1}$.

\begin{lemma}\label{lemma6}
Under \ref{assumption:B.1}, \ref{assumption:B.2} and \ref{assumption:B.7},
for $r\in {0,1,2,3,4}$
$$ \bar{S}_r(x) = M_{[\bar{S}_r]}(x) + \Op \left( \frac{1}{\sqrt{T} B_V^2}  +B_V^2\right)$$
uniformly in $x\in[0,1]$
where
\begin{align*}
M_{[\bar{S}_0]} &= c_0 g(x)^2, \qquad M_{[\bar{S}_1]} = M_{[\bar{S}_3]} = 0,\\
M_{[\bar{S}_2]} &= \frac{1}{2} c_0 g(x)^2 \sigma^2_K, \qquad \sigma^2_K = \int v^2K(v)\D{v},\\ 
M_{[\bar{S}_4]} &= \frac{1}{8} c_0 g(x)^2 \left(\mu_4^{(K)}+ 3\sigma^2_K\right), \qquad \mu_4^{(K)} = \int v^4 K(v)\D{v}.
\end{align*}
\end{lemma}
\begin{proof}
The proof of Lemma \ref{lemma6} follows in the footsteps of that of Lemma \ref{lemma:asymptotics_of_S_for_R}, and the details are omitted.
\end{proof}

\begin{lemma}\label{lemma7}
Under \ref{assumption:B.1} --- \ref{assumption:B.5} and \ref{assumption:B.7} --- \ref{assumption:B.9},
for $r\in {0,1,2}$
$$ \bar{Q}_r(x) = M_{[\bar{Q}_r]}(x) + \Op \left( \frac{1}{\sqrt{T} B_V^2}  +B_V^2\right)$$
uniformly in $x\in[0,1]$
where
\begin{align*}
M_{[\bar{Q}_0]} &= c_0 R_0(x,x) g(x)^2, \qquad M_{[\bar{Q}_1]} = 0,\\
M_{[\bar{Q}_2]} &= \frac{1}{2} c_0 R_0(x,x) g(x)^2 \sigma^2_K, \qquad \sigma^2_K = \int v^2K(v)\D{v}.
\end{align*}
\end{lemma}
\begin{proof}
The proof of Lemma \ref{lemma7} is analogous to the proofs of Lemma \ref{lemma:asymptotics_of_Q_for_R} and Lemma \ref{lemma:asymptotics_of_Q_for_R_in_P}.
\end{proof}

The following corollary is a direct consequence of  Lemma \ref{lemma6} and Lemma \ref{lemma7}, and the formula \eqref{eq:explicit_formula_barR}.
\begin{corollary}
Under \ref{assumption:B.1} --- \ref{assumption:B.5} and \ref{assumption:B.7} --- \ref{assumption:B.9},
$$ \bar{R}_0(x) = R_0(x,x) + \Op \left( \frac{1}{\sqrt{T} B_V^2}  +B_V^2\right) $$
uniformly in $x\in[0,1]$.
\end{corollary}

Now we turn our attention to the linear smoother on the diagonal \eqref{eq:local_LS_diagonal}.

\begin{lemma}\label{lemma8}
Under \ref{assumption:B.1} --- \ref{assumption:B.5} and \ref{assumption:B.7} --- \ref{assumption:B.9},
$$ \hat{V}(x) = R_0(x,x) + \sigma^2 + \Op\left( \frac{1}{\sqrt{T} B_V} + B_V^2 \right)$$
uniformly in $x\in[0,1]$.
\end{lemma}
\begin{proof}
The proof is similar to the proofs of the above lemmas. An explicit formula for the minimizer of \eqref{eq:local_LS_diagonal} can be found analogously.
\end{proof}

\begin{proof}[Proof of the last part of the Theorem \ref{thm:mean_and_autocov_function}]
Combining Lemma \ref{lemma6}, Lemma \ref{lemma7}, and Lemma \ref{lemma8} yields the rate \eqref{eq:thm:autocov_function:sigma}. See also the proof of \citet[Theorem 3.4]{li2010uniform} where the proof with the local-linear smoothing of the diagonal is written out in detail.
\end{proof}

\subsection{Proof of Theorem \ref{thm:spectral_density}}
\label{subsec:proof_of_Thm_2}

Firstly we comment that the minimizer to \eqref{eq:minimization-spectral-density} and hence the estimator can be expressed explicitly \eqref{eq:spectral_density_estimator_bartlett_smoother} as
\begin{equation}
\label{eq:minimization-spectral-density-explicit}
\hat{f}_\omega(x,y) = \frac{1}{2\pi}\left( 
\mathscr{A}_1 Q_{00}^{\omega} -
\mathscr{A}_2 Q_{10}^{\omega} -
\mathscr{A}_3 Q_{01}^{\omega}
 \right) \mathscr{B}^{-1} ,
 \end{equation}
where
\begin{align*}
\mathscr{A}_1 &= S_{20}S_{02}-S_{11}^2, \qquad 
\mathscr{A}_2  = S_{10}S_{02}-S_{01}S_{11},\qquad
\mathscr{A}_3  = S_{01}S_{20}-S_{10}S_{11}, \\
\mathscr{B} &= \mathscr{A}_1 S_{00} - \mathscr{A}_2 S_{10} - \mathscr{A}_3 S_{01}, \\
S_{pq} &= \frac{1}{\QBartlett} \sum_{h=-\QBartlett}^{\QBartlett}
\frac{W_h}{\widehat{\mathcal{N}_h}}
\sum_{t=\max(1,1-h)}^{\min(T,T-h)} 
\stackrel[j\neq k \text{ if } h=0]{}{\sum_{j=1}^{N_{t+h}} \sum_{k=1}^{N_t}}
\left( \frac{x_{t+h,j}-x}{B_R} \right)^p
\left( \frac{x_{tk}-y}{B_R} \right)^q
\times\\ &\times
\frac{1}{B_R^2}
K\left( \frac{x_{t+h,j}-x}{B_R} \right)
K\left( \frac{x_{tk}-y}{B_R} \right),
\\
Q_{pq}^{\omega} &= \sum_{h=-\QBartlett}^{\QBartlett}
\frac{W_h e^{-\I h\omega}}{\widehat{\mathcal{N}_h}}
\sum_{t=\max(1,1-h)}^{\min(T,T-h)}
\stackrel[j\neq k \text{ if } h=0]{}{\sum_{j=1}^{N_{t+h}} \sum_{k=1}^{N_t}}
G_{h,t}(x_{t+h,j}, x_{tk})
\left( \frac{x_{t+h,j}-x}{B_R} \right)^p
\left( \frac{x_{tk}-y}{B_R} \right)^q
\times \\
&\times
\frac{1}{B_R^2}
K\left( \frac{x_{t+h,j}-x}{B_R} \right)
K\left( \frac{x_{tk}-y}{B_R} \right).
\end{align*}
All of the above quantities are understood as functions of $(x,y)\in[0,1]^2$ and all operations are considered in a pointwise sense, including the pointwise inversion $\mathscr{B}^{-1} = (\mathscr{B}(x,y))^{-1}$.

To see why the minimizer has the form \eqref{eq:minimization-spectral-density-explicit} we simplify the notation of the complex minimisation problem \eqref{eq:minimization-spectral-density} to the following:
$$ \min_{d_0,d_1,d_2} \sum_{j=1}^J \left| A_j - d_0 - d_1 (x_j-x) - d_2 (y_j-y) \right|^2 v_j $$
where $A_j\in\mathbb{C}$ represents the raw covariances multiplied by the complex exponential, and $v_j\geq 0$ are the spatial and Barlett's weights. The sum of squares can be rewritten in the matrix notation as
$$ \min_{d_0,d_1,d_2} \left(\mathbf{A} - \mathbb{X} \mathbf{d}\right)^\dagger \mathbb{V} \left(\mathbf{A} - \mathbb{X} \mathbf{d}\right)$$
where $\dagger$ denotes the complex conjugate, $\mathbf{A} = (A_1,\dots,A_J)\transpose \in\mathbb{C}^J$, $\mathbf{d}=(d_0,d_1,d_3) \in\mathbb{C}^3, \mathbb{V} = \diag(v_1,\dots,v_J) \in\mathbb{R}^{J\times J}$ and
$$ \mathbb{X} =
\begin{pmatrix}
1 & x_1 - x & y_1 - y \\
\vdots & \vdots & \vdots \\
1 & x_J - x & y_J - y
\end{pmatrix} \in \mathbb{R}^{J\times 3}
.$$
Thanks to $\mathbb{X}$ and $\mathbb{V}$ being real, the real and imaginary parts of the minimisation can be separated:
\begin{multline*}
\hat{\mathbf{d}} =
\argmin_{\mathbf{d}} \left(\mathbf{A} - \mathbb{X} \mathbf{d}\right)^\dagger \mathbb{V} \left(\mathbf{A} - \mathbb{X} \mathbf{d}\right)
=\\=
\underbrace{
\left(\argmin_{\Re \mathbf{d}} \left(\Re\mathbf{A} - \mathbb{X} \Re\mathbf{d}\right)\transpose \mathbb{V} \left(\Re\mathbf{A} - \mathbb{X} \Re\mathbf{d}\right)
\right)}_{ \Re\hat{\mathbf{d}} }
+ \I 
\underbrace{
\left(
\argmin_{\Im \mathbf{d}} \left(\Im\mathbf{A} - \mathbb{X} \Im\mathbf{d}\right)\transpose \mathbb{V} \left(\Im\mathbf{A} - \mathbb{X} \Im\mathbf{d}\right)\right).
}_{ \Im\hat{\mathbf{d}} }
\end{multline*}
We solve the above minimisation problems by the classical normal equations formula for the weighted least squares problem and obtain
$$ \hat{\mathbf{d}} = \Re\hat{\mathbf{d}}+\I\Im\hat{\mathbf{d}} =
\left(\mathbb{X}\transpose\mathbb{V}\mathbb{X}\right)^{-1}\mathbb{X}\transpose \mathbb{V} \Re\mathbf{A} 
+\I
\left(\mathbb{X}\transpose\mathbb{V}\mathbb{X}\right)^{-1}\mathbb{X}\transpose \mathbb{V} \Im\mathbf{A} =
\left(\mathbb{X}\transpose\mathbb{V}\mathbb{X}\right)^{-1}\mathbb{X}\transpose \mathbb{V} \mathbf{A}.$$
We can calculate the first element of $\left(\mathbb{X}\transpose\mathbb{V}\mathbb{X}\right)^{-1}\mathbb{X}\transpose \mathbb{V} \mathbf{A}$ by Cram\'er's rule. After switching back to the quadruple summation \eqref{eq:minimization-spectral-density} we arrive at the formula \eqref{eq:minimization-spectral-density-explicit}.

To investigate the asymptotic behaviour of the estimator \eqref{eq:spectral_density_estimator_bartlett_smoother}, we need to analyse the asymptotics of the terms in the formula \eqref{eq:minimization-spectral-density-explicit}. We now assess the asymptotic behaviour of $S_{pq}$ and $Q_{pq}^\omega$.

\begin{lemma}
\label{lemma:asymptotics_of_S_for_F}
Under the assumptions\ref{assumption:B.1}, \ref{assumption:B.2}, and \ref{assumption:B.8}, for any $p,q\in\mathbb{N}_0$, such that $0\leq p+q\leq 2$:
$$ S_{pq} = M_{[S_{pq}]} + \Op\left( \frac{1}{\sqrt{T}} \frac{1}{B_R^2} +B_R^2\right) $$
uniformly in $x,y\in[0,1]$ and where
$$
\begin{gathered}
M_{[S_{00}]} = g(x)g(y), \qquad
M_{[S_{01}]} = M_{[S_{10}]} = M_{[S_{11}]} = 0, \\
M_{[S_{20}]} =  M_{[S_{02}]} = g(x)g(y) \sigma^2_K, \qquad
\sigma^2_K = \int v^2 K(v) \D{v}.
\end{gathered}
$$
\end{lemma}
\begin{proof}
Denote
$$\mathcal{S}^{(pq)}_{htjk} =  
\left( \frac{x_{t+h,j}-x}{B_R} \right)^p
\left( \frac{x_{tk}-y}{B_R} \right)^q
\frac{1}{B_R^2}
K\left( \frac{x_{t+h,j}-x}{B_R} \right)
K\left( \frac{x_{tk}-y}{B_R} \right),
$$
for $h=-L,\dots,L$, $t=1,\dots,T-h$, $j=1,\dots,N_{t+h}$, $k=1,\dots,N_t$ for $j\neq k$ if $h=0$.

Because $L=o(T)$ we may assume (and we do) in the entire proof that $L\leq T/2$.
Noting that $\QBartlett^{-1}\sum_{h=-\QBartlett}^{\QBartlett} W_h = 1$ we start with the decomposition
\begin{multline}
\label{eq:lemma9_proof_eq1}
\left|\frac{1}{\QBartlett}\sum_{h=-\QBartlett}^{\QBartlett} \left(
\frac{W_h}{\widehat{\mathcal{N}_h}}
\sum_{t=\max(1,1-h)}^{\min(T,T-h)} 
\stackrel[j\neq k \text{ if } h=0]{}{\sum_{j=1}^{N_{t+h}} \sum_{k=1}^{N_t}}
\mathcal{S}^{(pq)}_{htjk}
\right) - M_{[S_{pq}]} \right|
\leq\\\leq
\left|\frac{1}{\QBartlett}\sum_{h=-\QBartlett}^{\QBartlett}
\frac{W_h}{\widehat{\mathcal{N}_h}}
\sum_{t=\max(1,1-h)}^{\min(T,T-h)} 
\stackrel[j\neq k \text{ if } h=0]{}{\sum_{j=1}^{N_{t+h}} \sum_{k=1}^{N_t}}
\left(\mathcal{S}^{(pq)}_{htjk}
-M_{[S_{pq}]}
\right) \right|
+
\left|
\frac{1}{\QBartlett}\sum_{h=-\QBartlett}^{\QBartlett}
W_h M_{[S_{pq}]} \left( 1 - \frac{\mathcal{N}_h}{\widehat{\mathcal{N}_h}} \right)
\right|
\leq\\\leq
\left|\frac{1}{\QBartlett}\sum_{h=-\QBartlett}^{\QBartlett}
\frac{W_h}{\widehat{\mathcal{N}_h}}
\sum_{t=\max(1,1-h)}^{\min(T,T-h)} 
\stackrel[j\neq k \text{ if } h=0]{}{\sum_{j=1}^{N_{t+h}} \sum_{k=1}^{N_t}}
\left(\mathcal{S}^{(pq)}_{htjk}
-M_{[S_{pq}]}
\right) \right|
+
\left|\frac{1}{L} M_{[S_{pq}]} \left( 1 - \frac{\mathcal{N_0}}{\widehat{\mathcal{N}_0}} \right)\right|
+\\+
\left|
\frac{1}{\QBartlett}\sum_{h=-\QBartlett, h\neq 0}^{\QBartlett}
W_h M_{[S_{pq}]} \left( 1 - \frac{\mathcal{N}_h}{(T-|h|) (\E N)^2} \right)
\right|
+
\left|
\frac{1}{\QBartlett}\sum_{h=-\QBartlett, h\neq 0}^{\QBartlett}
W_h M_{[S_{pq}]}
\frac{\mathcal{N}_h}{ T-|h| }
\left( \frac{1}{(\bar{N})^2} - \frac{1}{(\E N)^2} \right)
\right|
\end{multline}
where $\mathcal{N}_h = \sum_{t=\min(1,1-h)}^{\max(T,T-h)} N_{t+h}N_t$ for $h\neq 0$ and $\mathcal{N}_0 =\sum_{t=1}^T N_t(N_t-1)$.
The second term on the right-hand side of \eqref{eq:lemma9_proof_eq1} is of order $\Op(L)$. The third term is bounded by bounding the variance $\mathcal{N}_t \leq U T$ for $|h|\leq T/2$ where the constant $U$ is independent of $T$ and $h$ but may depend on the distribution of $N$. Thus the third term is of order $\Op(T^{-1/2})$ thanks to
$$
\left|
\frac{1}{\QBartlett}\sum_{h=-\QBartlett,h\neq 0}^{\QBartlett}
\Ez{ 1 - \frac{\mathcal{N}_h}{(T-|h|) (\E N)^2}  }
\right|
\leq
\frac{1}{\QBartlett}
\sum_{h=-\QBartlett,h\neq 0}^{\QBartlett}
\left\{\var\left( \frac{\mathcal{N}_h}{(T-|h|)(\E N)^2} \right)\right\}^{1/2}
=O(T^{-1/2}).
$$
The fourth term on the right hand side of order $\Op(T^{-1/2})$ because $\bar{N} = \E N + \Op(T^{-1/2})$.

The first term on the right-hand side of \eqref{eq:lemma9_proof_eq1} is decomposed as
\begin{multline}
\label{eq:lemma9_proof_eq2}
\left|\frac{1}{\QBartlett}\sum_{h=-\QBartlett}^{\QBartlett}
\frac{W_h}{\widehat{\mathcal{N}_h}}
\sum_{t=\max(1,1-h)}^{\min(T,T-h)} 
\stackrel[j\neq k \text{ if } h=0]{}{\sum_{j=1}^{N_{t+h}} \sum_{k=1}^{N_t}}
\left(\mathcal{S}^{(pq)}_{htjk}
-M_{[S_{pq}]}
\right) \right|
\leq\\\leq
\frac{1}{\QBartlett}\sum_{h=-\QBartlett}^{\QBartlett}
\frac{W_h}{\widehat{\mathcal{N}_h}}
\sum_{t=\max(1,1-h)}^{\min(T,T-h)} 
\stackrel[j\neq k \text{ if } h=0]{}{\sum_{j=1}^{N_{t+h}} \sum_{k=1}^{N_t}}
\left| \mathcal{S}^{(pq)}_{htjk} - \E \mathcal{S}^{(pq)}_{htjk}\right|
+\\+
\frac{1}{\QBartlett}\sum_{h=-\QBartlett}^{\QBartlett}
\frac{W_h}{\widehat{\mathcal{N}_h}}
\sum_{t=\max(1,1-h)}^{\min(T,T-h)} 
\stackrel[j\neq k \text{ if } h=0]{}{\sum_{j=1}^{N_{t+h}} \sum_{k=1}^{N_t}}
\left| \E \mathcal{S}^{(pq)}_{htjk} - M_{[S_{pq}]}\right|.
\end{multline}

The second term on the right hand side of \eqref{eq:lemma9_proof_eq2} is of order $\Op(B_R^2)$ because $\left| \E \mathcal{S}^{(pq)}_{htjk} - M_{[S_{pq}]}\right| = \Op(B_R^2)$ uniformly.
The first term on the right hand side of \eqref{eq:lemma9_proof_eq2} is treated using similar steps as in the proof of Lemma \ref{lemma:asymptotics_of_S_for_R}, therefore for a constant $U$ independent of $B_R, T$ and $|h|<T/2$,
$$
\Ez{
\sum_{t=\max(1,1-h)}^{\min(T,T-h)} 
\stackrel[j\neq k \text{ if } h=0]{}{\sum_{j=1}^{N_{t+h}} \sum_{k=1}^{N_t}}
\sup_{x,y\in[0,1]}
\left| \mathcal{S}^{(pq)}_{htjk} - \E \mathcal{S}^{(pq)}_{htjk}\right| 
\vert
N_1,\dots,N_T
}
\leq U \frac{\sqrt{\mathcal{N}_h}}{B_R^2}.
$$
The observation that
$$ \frac{1}{L}\sum_{h=-L}^L \frac{\sqrt{\mathcal{N}_h}}{\widehat{\mathcal{N}}_h} = \Op(T^{-1/2})$$
now establishes that the first term on the right hand side of \eqref{eq:lemma9_proof_eq2} is of order $\Op(\frac{1}{\sqrt{T}} \frac{1}{B_R^2} + B_R^2)$.

\end{proof}
%\begin{proof}
%OLD PROOF
%We will use Lemma \ref{lemma:asymptotics_of_S_for_R}. Note that thanks to $\QBartlett = o(T)$ we may assume $\QBartlett \leq T/2$ and thus $1/(T-|h|) \leq 2/T$ which translates all bounds of Lemma \ref{lemma:asymptotics_of_S_for_R} into a common denominator.
%
%Noting that $\QBartlett^{-1}\sum_{h=-\QBartlett}^{\QBartlett} W_h = 1$ we calculate
%\begin{multline*}
%\left|\frac{1}{\QBartlett}\sum_{h=-\QBartlett}^{\QBartlett} \left( \frac{T-|h|}{T} W_h S_{pq}^{(h)} \right) - M_{[S_{pq}]} \right| =
%\left|\frac{1}{\QBartlett}\sum_{h=-\QBartlett}^{\QBartlett} \left( \frac{T-|h|}{T} W_h S_{pq}^{(h)}  - W_h M_{[S_{pq}]} \right) \right|
%\leq\\\leq
%\frac{1}{\QBartlett}\sum_{h=-\QBartlett}^{\QBartlett} \frac{T-|h|}{T} W_h \left| S_{pq}^{(h)} - M_{[S_{pq}]} \right|
%+ \frac{1}{\QBartlett}\sum_{h=-\QBartlett}^{\QBartlett} \frac{|h|}{T} W_h M_{[S_{pq}]}
%\end{multline*}
%Taking the supremum norm and the expectation we bound the first term by $U / (\sqrt{T}B_R^2)$ by Lemma \ref{lemma:asymptotics_of_S_for_R}. The second term is bounded by $\QBartlett M_{[S_{pq}]} /T$ which is a faster rate than the one above.
%\end{proof}

\begin{lemma}
\label{lemma:asymptotics_of_Q_for_F}
For $p,q\in\mathbb{N}_0$, be such that $0\leq p+q\leq 2$, we have
%\textcolor{red}{ARE YOU REFERRING TO THE C(p,q) CONDITIONS IN PANARETOS AND TAVAKOLI? IF SO YOU MUST POINT THIS OUT}
\begin{enumerate}
\item under \ref{assumption:B.1} --- \ref{assumption:B.5}, \ref{assumption:B.7}, \ref{assumption:B.8} and \ref{assumption:B.10}
\begin{align*}
\tilde{Q}_{pq}^\omega &= M_{[Q^\omega_{pq}]} + \op\left( 1\right),\\
Q_{pq}^\omega &= M_{[Q^\omega_{pq}]} + \op\left( 1\right),
\end{align*}
\item under \ref{assumption:B.1} --- \ref{assumption:B.8} and \ref{assumption:B.10}
\begin{align*}
\tilde{Q}_{pq}^\omega &= M_{[Q^\omega_{pq}]} + \Op\left( \QBartlett \frac{1}{\sqrt{T}} \frac{1}{B_R^2}  + \QBartlett B_R^2\right),\\
Q_{pq}^\omega &= M_{[Q^\omega_{pq}]} + \Op\left( \QBartlett \frac{1}{\sqrt{T}} \frac{1}{B_R^2}  + \QBartlett B_R^2\right) ,
\end{align*}
\end{enumerate}
where all convergences are uniformly in $\omega\in[-\pi,\pi]$ and $x,y\in[0,1]$
and
$$
M_{[Q_{00}^\omega]} = 2\pi g(x)g(y) f_\omega(x,y) , \qquad
M_{[Q_{10}^\omega]} = M_{[Q_{01}^\omega]} = 0.
$$
\end{lemma}
\begin{proof}
Analogously to Lemma \ref{lemma:asymptotics_of_Q_for_R}, we first assume that $\mu(\cdot)$ is known. Hence we define
\begin{multline*}
\tilde{Q}_{pq}^\omega =
\sum_{h=-\QBartlett}^{\QBartlett}
\frac{W_h e^{-\I h\omega}}{\widehat{\mathcal{N}_h}}
\sum_{t=\max(1,1-h)}^{\min(T,T-h)}
\stackrel[j\neq k \text{ if } h=0]{}{\sum_{j=1}^{N_{t+h}} \sum_{k=1}^{N_t}}
\tilde{G}_{h,t}(x_{t+h,j}, x_{tk})
\left( \frac{x_{t+h,j}-x}{B_R} \right)^p
\left( \frac{x_{tk}-y}{B_R} \right)^q
\times \\
\times
\frac{1}{B_R^2}
K\left( \frac{x_{t+h,j}-x}{B_R} \right)
K\left( \frac{x_{tk}-y}{B_R} \right).
\end{multline*}

Denote
$M_{[Q_{00,h}]}(x,y) = g(x)g(y) R_h(x,y)$ and $M_{[Q_{10,h}]}(x,y)= M_{[Q_{01,h}^\omega]}(x,y) = 0$.
Further denote
$$\mathcal{Q}^{(pq)}_{htjk} =
\tilde{G}_{h,t}(x_{t+h,j}, x_{tk})
\left( \frac{x_{t+h,j}-x}{B_R} \right)^p
\left( \frac{x_{tk}-y}{B_R} \right)^q
\frac{1}{B_R^2}
K\left( \frac{x_{t+h,j}-x}{B_R} \right)
K\left( \frac{x_{tk}-y}{B_R} \right)
$$
$h=-L,\dots,L$, $t=1,\dots,T-h$, $j=1,\dots,N_{t+h}$, $k=1,\dots,N_t$ for $j\neq k$ if $h=0$,
we can write
\begin{multline}\label{eq:lemma10_proof_eq1}
\left| \tilde{Q}_{pq}^\omega - M_{[Q^\omega_{pq}]} \right|
\leq
\left|
\sum_{h=-\QBartlett}^{\QBartlett}
\frac{W_h e^{-\I h\omega}}{\widehat{\mathcal{N}_h}}
\sum_{t=\max(1,1-h)}^{\min(T,T-h)}
\stackrel[j\neq k \text{ if } h=0]{}{\sum_{j=1}^{N_{t+h}} \sum_{k=1}^{N_t}}
\tilde{\mathcal{Q}}^{(pq)}_{htjk}
-
\sum_{h=-\infty}^\infty M_{[Q_{pq,h}]} e^{-\I h \omega}
\right|
\leq\\\leq
\sum_{h=-\QBartlett}^{\QBartlett}
\frac{W_h}{\widehat{\mathcal{N}_h}}
\sum_{t=\max(1,1-h)}^{\min(T,T-h)}
\stackrel[j\neq k \text{ if } h=0]{}{\sum_{j=1}^{N_{t+h}} \sum_{k=1}^{N_t}}
\left|
\tilde{\mathcal{Q}}^{(pq)}_{htjk}
- M_{[Q_{pq,h}]}
\right|
+\\+
\frac{1}{\QBartlett}
\sum_{h=-\QBartlett}^\QBartlett
|h| \left| M_{[Q_{pq,h}^\omega]} \right|
+
\sum_{|h|>\QBartlett} \left|M_{[Q_{pq,h}]}\right|
\end{multline}
Under the assumption \ref{assumption:B.5}, the second and the third term on the right-hand side of \eqref{eq:lemma10_proof_eq1} converge to zero uniformly in $x,y\in[0,1]$ by Kronecker's lemma. Assuming further the assumption \ref{assumption:B.6}, these terms are in fact of order $O(L^{-1})$ uniformly in $x,y\in[0,1]$.

The first term on the right-hand side of \eqref{eq:lemma10_proof_eq1} is treated similarly as in the proof of Lemma \ref{lemma:asymptotics_of_Q_for_R_in_P}.
\begin{multline}\label{eq:lemma10_proof_eq2}
\sum_{h=-\QBartlett}^{\QBartlett}
\frac{W_h}{\widehat{\mathcal{N}_h}}
\sum_{t=\max(1,1-h)}^{\min(T,T-h)}
\stackrel[j\neq k \text{ if } h=0]{}{\sum_{j=1}^{N_{t+h}} \sum_{k=1}^{N_t}}
\left|
\tilde{\mathcal{Q}}^{(pq)}_{htjk}
- M_{[Q_{pq,h}]}
\right|
\leq\\\leq
\sum_{h=-\QBartlett}^{\QBartlett}
\frac{W_h}{\widehat{\mathcal{N}_h}}
\sum_{t=\max(1,1-h)}^{\min(T,T-h)}
\stackrel[j\neq k \text{ if } h=0]{}{\sum_{j=1}^{N_{t+h}} \sum_{k=1}^{N_t}}
\left|
\tilde{\mathcal{Q}}^{(pq)}_{htjk}
- \E \tilde{\mathcal{Q}}^{(pq)}_{htjk}
\right|
+\\+
\sum_{h=-\QBartlett}^{\QBartlett}
\frac{W_h}{\widehat{\mathcal{N}_h}}
\sum_{t=\max(1,1-h)}^{\min(T,T-h)}
\stackrel[j\neq k \text{ if } h=0]{}{\sum_{j=1}^{N_{t+h}} \sum_{k=1}^{N_t}}
\left|
\E \tilde{\mathcal{Q}}^{(pq)}_{htjk}
- M_{[Q_{pq,h}]}
\right|
\end{multline}
The second term on the right-hand side of \eqref{eq:lemma10_proof_eq2} is of order $O(\QBartlett B_R^2)$ uniformly in $x,y\in[0,1]$.
The first term on the right-hand side of \eqref{eq:lemma10_proof_eq2} is treated analogously as in the proof of Lemma \ref{lemma:asymptotics_of_Q_for_R_in_P}, thus there exists a constant $U$ independent of $B_R,T$ and $|h|<T/2$ such that
$$ \Ez{
\sum_{h=-\QBartlett}^{\QBartlett}
\frac{W_h}{\widehat{\mathcal{N}_h}}
\sum_{t=\max(1,1-h)}^{\min(T,T-h)}
\stackrel[j\neq k \text{ if } h=0]{}{\sum_{j=1}^{N_{t+h}} \sum_{k=1}^{N_t}}
\sup_{x,y\in[0,1]} \left| \tilde{\mathcal{Q}}^{(pq)}_{htjk} - \E \tilde{\mathcal{Q}}^{(pq)}_{htjk} \right|
\vert
N_1,\dots,N_t
}
\leq U \frac{\sqrt{\mathcal{N}_h}}{B_R^2}.
$$
Observing that
$$ \sum_{h=-\QBartlett}^\QBartlett \frac{\sqrt{\mathcal{N}_h}}{\widehat{\mathcal{N}}_h}
= \Op( L T^{-1/2} )
$$
concludes the rates $\op(1)$, and $\Op( L T^{-1/2} B_R^{-2})$ under the assumption \ref{assumption:B.6}.

The proof is completed by the repetition of the steps in the proof of Lemma \ref{lemma:asymptotics_of_Q_for_R_in_P}, switching to the $\Op$ notation and noting that the difference between $\tilde{Q}_{pq}^\omega$ and $Q_{pq}^\omega$ is asymptotically negligible.
\end{proof}

\begin{proof}[Proof of Theorem \ref{thm:spectral_density}]

Combining the above derived results in lemmas \ref{lemma:asymptotics_of_S_for_F} and \ref{lemma:asymptotics_of_Q_for_F} we are ready to establish the asymptotic behaviour of the terms that enter the formula \eqref{eq:minimization-spectral-density-explicit}.
\begin{align*}
\mathscr{A}_1 &= \left[ g(x) g(y) \sigma^2_K \right]^2 + \Op\left(\frac{1}{\sqrt{T}} \frac{1}{B_R^2} + B_R^2\right), \\
\mathscr{A}_2 &= \Op\left(\frac{1}{\sqrt{T}} \frac{1}{B_R^2} + B_R^2\right),\\
\mathscr{A}_3 &= \Op\left(\frac{1}{\sqrt{T}} \frac{1}{B_R^2} + B_R^2\right),\\
\mathscr{B} &= \left[ g(x) g(y)\right]^3  \left(\sigma^2_K \right)^2
+ \Op\left(\frac{1}{\sqrt{T}} \frac{1}{B_R^2} + B_R^2\right), \\
Q_{00} &= 2\pi g(x)g(y) f_\omega(x,y) + \op\left( 1 \right), \\
Q_{10} &= \op\left( 1 \right), \\
Q_{01} &= \op\left( 1 \right)
\end{align*}
uniformly in $\omega\in[-\pi,\pi]$ and $x,y\in[0,1]$.
Finally, the numerator of \eqref{eq:minimization-spectral-density-explicit} is
$$ \mathscr{A}_1 Q_{00}^{\omega} -
\mathscr{A}_2 Q_{10}^{\omega} -
\mathscr{A}_3 Q_{01}^{\omega}
=
2\pi f_\omega(x,y) \left[ g(x) g(y)\right]^3\left(\sigma^2_K \right)^2 + \op(1)
$$
uniformly in $\omega\in[-\pi,\pi]$ and $x,y\in[0,1]$ which completes the proof of consistency. Under the assumption \ref{assumption:B.6} we replace $\op(1)$ by $\Op( \QBartlett / (\sqrt{T} B_R^2) + \QBartlett B_R^2)$.

\end{proof}

\subsection{Proof of Corollary \ref{corollary:sup_consistency}}

\begin{proof}[Proof of Corollary \ref{corollary:sup_consistency}]
Note that for $h\in\mathbb{Z}$ and $x,y\in[0,1]$:
$$ \tilde{R}_h(x,y) - R_h(x,y) = \int_{-\pi}^\pi \left\{ \tilde{f}_\omega (x,y) - f_\omega (x,y)\right\} e^{\I h\omega} \D{\omega}.$$
Therefore
$$ \sup_{h\in\mathbb{Z}}\sup_{x,y\in[0,1]} \left| \tilde{R}_h(x,y) - R_h(x,y) \right|
\leq 2\pi\sup_{\omega\in [-\pi,\pi]}\sup_{x,y\in[0,1]}\left| \tilde{f}_\omega (x,y) - f_\omega (x,y)\right|
= \op(1).$$
Assuming further \ref{assumption:B.6}, proving the statement \eqref{eq:corollary:sup_consistency:statement2} is analogous to the previous line.
\end{proof}

\subsection{Proof of Theorem \ref{thm:strong_mixing_mu_R}}
\begin{proof}[Proof of Theorem \ref{thm:strong_mixing_mu_R} ]

{We begin with the estimation of the mean function $\mu(\cdot)$.} We are going to make use of the result in \citet[Thm 10]{hansen2008uniform}.
Define the two-dimensional time-series $\{\tilde{Y}_i,\tilde{X}_i\}_i$ composed of the sparse observations and their observation locations according to the observation scheme \eqref{eq:observation_scheme}
\begin{align*}
\left(\tilde{Y}_1, \tilde{Y}_2, \dots \right) &=
\left( Y_{1,1},\dots,Y_{1,N_1},Y_{2,1},\dots,Y_{2,N_2},Y_{3,1},\dots \right), \\
\left(\tilde{X}_1, \tilde{X}_2, \dots \right) &=
\left( x_{1,1},\dots,x_{1,N_1},x_{2,1},\dots,x_{2,N_2},x_{3,1},\dots \right).
\end{align*}
Under the assumptions \ref{assumption:D1} --- \ref{assumption:D9}, the time-series $\{\tilde{Y}_i,\tilde{X}_i\}_i$ is strictly stationary and strongly mixing, with mixing coefficients $\tilde{\alpha}(m)\leq \tilde{A}m^{-\beta}$, and satisfies the conditions (2) --- (7) of \citet{hansen2008uniform}. Indeed:
\begin{enumerate}
\item[(3)]
$$ \E |\tilde{Y}_1|^s \leq 2^s \left( \E |X_1(x_{11})|^s + \E |\epsilon_{1,1}|^s \right) <\infty$$
\item[(6)]
$$ \sup_{x\in[0,1]} \Ez{ |\tilde{Y}_1|^s \vert \tilde{X} = x} g(x) \leq B_1 B_3 <\infty $$
\item[(7)]
\begin{multline*}
\sup_{x,x'\in[0,1]} \Ez{ |\tilde{Y}_1 \tilde{Y}_j| \vert \tilde{X}_1 = x, \tilde{X}_j=x' } g(x)g(x') \leq\\\leq
\left\{
\sup_{h\in\mathbb{Z}}\sup_{x,x'\in[0,1]}\left[\E \left| X_h(x)X_0(x') \right| \right] +
2\sup_{x\in[0,1]}\left[ \E|X_0(x)| \E|\epsilon_{1,1}| \right] + 
\sigma^2\right\}g(x)g(x') <\infty
\end{multline*}
\end{enumerate}
The conditions (10) --- (13) of \citet{hansen2008uniform} are also satisfied taking $q=\infty$ and $c_n=1$.
Therefore all conditions of \citet[Thm 10]{hansen2008uniform} are satisfied. Noting that the length $n=n(T)$ of the time-series $\{\tilde{Y}_i,\tilde{X}_i\}_i$ is asymptotically of the same order as $T$ of the functional time series $\{X_t(\cdot)\}$, formally $n=n(T) \asymp T, \quad T\to\infty$, yields
$$ \sup_{x\in[0,1]} \left| \hat{\mu}(x) - \mu(x) \right| =
\Op\left( \sqrt{\frac{\log T}{T B_\mu}} + B_\mu^2 \right).$$

\noindent
{Next we turn to the estimation of the lag-$h$ autocovariance kernels $R_h(\cdot,\cdot)$.} Fix $h\in\mathbb{Z}$. For simplicity consider $h\neq 0$. The proof for $h=0$ is essentially the same, only the diagonal ``raw'' covariances must be removed.
For the moment assume that the mean function $\mu(\cdot)$ is known and we shall work with the ``raw'' covariances $\tilde{G}_{h,t}(x_{t+h,j},x_{tk})$ as defined in \eqref{eq:raw_covariance_known_mu}.
Similarly as in the first part of this proof, define now the three dimensional time-series $\{\tilde{Y}_i,\tilde{X}_i\}_i$ composed of the ``raw'' covariances and their locations
\begin{align*}
\left(\tilde{Y}_1, \tilde{Y}_2, \dots \right) &=
\left( \tilde{G}_{h,1}(x_{1+h,1},x_{1,1}),\dots,\tilde{G}_{h,1}(x_{1+h,N_{1+h}},x_{1,N_t}),
\tilde{G}_{2,1}(x_{2+h,1},x_{2,1}),\dots \right), \\
\left(\tilde{X}_1, \tilde{X}_2, \dots \right) &=
\left(
\left[
\begin{array}{c}
x_{1+h,1} \\ x_{1,1}
\end{array} 
\right],
\dots,
\left[
\begin{array}{c}
x_{1+h,N_{t+h}} \\ x_{1,N_t}
\end{array} 
\right],
\left[
\begin{array}{c}
x_{2+h,1} \\ x_{2,1}
\end{array} 
\right],
\dots
\right).
\end{align*}
We are again going to make us of \citet[Thm 10]{hansen2008uniform}.
Under the assumptions \ref{assumption:D1} --- \ref{assumption:D13}
it is easy to verify (analogously as in the first part of this proof) that
the time series $\{\tilde{Y}_i,\tilde{X}_i\}_i$ satisfies the conditions (2) --- (7) of \citet{hansen2008uniform}. The conditions (10) --- (13) also follow directly from our assumptions. 
It remains to repeat the discussion as in the proof of Lemma \ref{lemma:asymptotics_of_Q_for_R_in_P} to conclude that the difference between $\tilde{G}_{h,t}(x_{t+h,j},x_{tk})$ and $G_{h,t}(x_{t+h,j},x_{tk})$ is asymptotically negligible with respect to the rate bellow.

Therefore by \citet[Thm 10]{hansen2008uniform}, for fixed $h\in\mathbb{Z}$,
$$ \sup_{x,y\in[0,1]} \left| \hat{R}_h(x,y) - R_h(x,y) \right| =
\Op\left( \sqrt{\frac{\log T}{T B_R^2}} + B_R^2\right).$$

\end{proof}

\subsection{Proof of Theorem \ref{thm:strong_mixing_F}}

The proof of Theorem \ref{thm:strong_mixing_F} is more involved. Rather than make direct use of, we shall need to modify the proof techniques of \citet{hansen2008uniform} in order to construct our proof. We express the spectral density kernel estimator \eqref{eq:spectral_density_estimator_bartlett_smoother} in a similar way as in the proof of Theorem \ref{thm:spectral_density}.

\begin{equation}\label{eq:proof_strong_mixing_eq0.5}
\hat{f}_\omega(x,y) = \frac{1}{2\pi}\left( 
\mathscr{A}_1 Q_{00}^{\omega} -
\mathscr{A}_2 Q_{10}^{\omega} -
\mathscr{A}_3 Q_{01}^{\omega}
 \right) \mathscr{B}^{-1} ,
\end{equation}
where
\begin{align*}
\mathscr{A}_1 &= S_{20}S_{02}-S_{11}^2, \qquad 
\mathscr{A}_2  = S_{10}S_{02}-S_{01}S_{11},\qquad
\mathscr{A}_3  = S_{01}S_{20}-S_{10}S_{11}, \\
\mathscr{B} &= \mathscr{A}_1 S_{00} - \mathscr{A}_2 S_{10} - \mathscr{A}_3 S_{01}, \\
S_{pq} &= \frac{1}{\QBartlett} \sum_{h=-\QBartlett}^{\QBartlett}
\frac{W_h \mathcal{N}_h}{\widehat{\mathcal{N}_h}}
S_{pq}^{(h)}, \\
S_{pq}^{(h)} &= \frac{1}{\mathcal{N}_h B_R^2} 
\sum_{t=\max(1,1-h)}^{\min(T,T-h)} 
\stackrel[j\neq k \text{ if } h=0]{}{\sum_{j=1}^{N_{t+h}} \sum_{k=1}^{N_t}}
\left( \frac{x_{t+h,j}-x}{B_R} \right)^p
\left( \frac{x_{tk}-y}{B_R} \right)^q
K\left( \frac{x_{t+h,j}-x}{B_R} \right)
K\left( \frac{x_{tk}-y}{B_R} \right),
\\
Q_{pq}^{\omega} &= \sum_{h=-\QBartlett}^{\QBartlett}
\frac{W_h e^{-\I h\omega} \mathcal{N}_h}{\widehat{\mathcal{N}_h}}
Q_{pq}^{(h)}, \\
Q_{pq}^{(h)} &=
\frac{1}{\mathcal{N}_h B_R^2}
\sum_{t=\max(1,1-h)}^{\min(T,T-h)}
\stackrel[j\neq k \text{ if } h=0]{}{\sum_{j=1}^{N_{t+h}} \sum_{k=1}^{N_t}}
\tilde{G}_{h,t}(x_{t+h,j}, x_{tk})
\left( \frac{x_{t+h,j}-x}{B_R} \right)^p
\left( \frac{x_{tk}-y}{B_R} \right)^q
\times \\
&\times
K\left( \frac{x_{t+h,j}-x}{B_R} \right)
K\left( \frac{x_{tk}-y}{B_R} \right).
\end{align*}
All of the above quantities are understood as functions of $(x,y)\in[0,1]^2$ and all operations are considered in a pointwise sense, including the pointwise inversion $\mathscr{B}^{-1} = (\mathscr{B}(x,y))^{-1}$.

Similarly as in the proof of Theorem \ref{thm:strong_mixing_mu_R} define for $h\in\mathbb{Z}$,
\begin{align*}
\left(\tilde{Y}_1^{h,r}, \tilde{Y}_2^{h,r}, \dots \right) &=
\left(
\left\{\tilde{G}_{h,1}(x_{1+h,1},x_{1,1}) \right\}^r,\dots,
\left\{\tilde{G}_{h,1}(x_{1+h,N_{1+h}},x_{1,N_t}) \right\}^r,
\left\{\tilde{G}_{2,1}(x_{2+h,1},x_{2,1}) \right\}^r ,\dots \right), \\
\left(\tilde{X}_1^{h}, \tilde{X}_2^{h}, \dots \right) &=
\left(
\left[
\begin{array}{c}
x_{1+h,1} \\ x_{1,1}
\end{array} 
\right],
\dots,
\left[
\begin{array}{c}
x_{1+h,N_{t+h}} \\ x_{1,N_t}
\end{array} 
\right],
\left[
\begin{array}{c}
x_{2+h,1} \\ x_{2,1}
\end{array} 
\right],
\dots
\right).
\end{align*}

Let $k(\cdot) : \mathbb{R}^2\to\mathbb{R}$ be a function satisfying the assumption \ref{assumption:C}
and denote for $r=0,1$ and $h\in\mathbb{Z}$ define
\begin{align}
\label{eq:proof_strong_mixing_eq0.74}
\hat{\Psi}^{h,r}(x,y) &= \frac{1}{\mathcal{N}_h B_R^2}
\sum_{t=\max(1,1-h)}^{\min(T,T-h)}
\stackrel[j\neq k \text{ if } h=0]{}{\sum_{j=1}^{N_{t+h}} \sum_{k=1}^{N_t}}
\left\{\tilde{G}_{h,t}(x_{t+h,j}, x_{tk})\right\}^r
k\left( \frac{x_{t+h,j}-x}{B_R} , \frac{x_{tk}-y}{B_R} \right) \\
\label{eq:proof_strong_mixing_eq0.75}
&=
\frac{1}{\mathcal{N}_h B_R^2}
\sum_{i=1}^{\mathcal{N}_h}
\tilde{Y}_i^{h,r}
k\left( \frac{\tilde{X}_i^{h} - (x,y)}{B_R} \right) \\
\label{eq:proof_strong_mixing_eq0.76}
&=
\frac{1}{\mathcal{N}_h B_R^2}
\sum_{i=1}^{\mathcal{N}_h}
Z_i^{h,r}(x,y)
\end{align}
where we are denoting
\begin{equation}
\label{eq:proof_strong_mixing_eq1}
Z_i^{h,r}(x,y) = \tilde{Y}_i k\left( \left(\tilde{X}_i^{h,r} - (x,y)\right)/B_R \right).
\end{equation}

\begin{lemma}\label{lemma12}
Under the assumptions \ref{assumption:D1} --- \ref{assumption:D13},
$$
\var( \hat{\Psi}^{h,r}(x,y) \vert \mathcal{N}_h ) \leq \frac{\Theta}{\mathcal{N}_h B_R^2}
$$
for $\mathcal{N}_h > 0$
and where the constant $\Theta$ is uniform in $h\in\mathbb{Z}$, $x,y\in[0,1]$, $r=0,1$.
\end{lemma}
\begin{proof}
Note that the sequence $\{ Z_i^{h,r}(x,y) \}_i$ is a stationary scalar time-series and denote its autocovariance function as $\rho_{Z_i^{h,r}(x,y)}(\xi)$ for lag $\xi$. Therefore we have the bound 
\eqref{eq:proof_arbitrary_stationary_timeseries_variance}. Conditioning on $\mathcal{N}_h$ yields
\begin{equation}\label{eq:proof_strong_mixing_eq2}
\var\left( \frac{1}{\mathcal{N}_h} \sum_{i=1}^{\mathcal{N}_h} Z_i^{h,r}(x,y) \vert \mathcal{N}_h\right)
\leq \frac{1}{\mathcal{N}_h} \sum_{\xi=-\infty}^\infty \left| \rho_{Z_i^{h,r}(x,y)}(\xi) \right|.
\end{equation}
The sum on the right hand side of \eqref{eq:proof_strong_mixing_eq2} can be bounded by
\begin{multline}
\label{eq:proof_strong_mixing_eq3}
\sum_{\xi=-\infty}^\infty \left| \rho_{Z_i^{h,r}(x,y)}(\xi) \right|
\leq\\\leq
\left(N^{max}\right)^2
\sum_{\xi=-\infty}^\infty
\sup_{x_1,x_2,x_3,x_4\in[0,1]}
\Bigg|
\cum\left( X_{\xi+h}(x_1),X_{\xi}(x_2),X_{h}(x_3),X_{0}(x_4) \right)
+\\+
R_\xi(x_1,x_2)R_\xi(x_3,x_4) + R_{\xi+h}(x_1,x_4)R_{\xi-h}(x_2,x_3)
\Bigg|.
\end{multline}
The bound \eqref{eq:proof_strong_mixing_eq3} is uniform in $h$ and constitutes the constant $\Theta$ in the statement of Lemma \ref{lemma12}.
\end{proof}

The key tool for our proof is an exponential-type inequality for strongly mixing random sequences. This inequality was given by \citet[Thm 2.1]{liebscher1996strong}, whose result was derived from \citet[Thm 5]{rio1995functional}.
\begin{lemma}[Liebscher/Rio]\label{lemma13}
Let $Z_i$ be a stationary zero-mean real-valued process such that $|Z_i| \leq b$, with strong mixing coefficients $\alpha_m$. Then for each positive integer $m\leq n$ and $\epsilon$ such that $m<\epsilon b/4$
$$ \Prob\left( \left| \sum_{i=1}^n Z_i \right| > \epsilon \right) \leq
4 \exp\left(
-\frac{\epsilon^2}{
	64\frac{n\sigma_m^2}{m} +
	\frac{8}{3} \epsilon m b
}
\right)
+
4\frac{n}{m}\alpha_m,
$$
where $\sigma_m^2 = \E\left( \sum_{i=1}^m Z_i \right)^2$
\end{lemma}

\begin{lemma}\label{lemma16}
Under the assumptions \ref{assumption:D1} --- \ref{assumption:D11} and \ref{assumption:D13} --- \ref{assumption:D16}
\begin{equation}\label{eq:lemma16_eq0.5}
\sup_{\omega\in[-\pi,\pi]} \sup_{x,y\in[0,1]} \left| Q_{pq}^{\omega} - M_{[Q_{pq}^\omega]} \right| =
\op(1),
\end{equation}
and assuming further assumption \ref{assumption:B.6},
\begin{equation}\label{eq:lemma16_eq1}
\sup_{\omega\in[-\pi,\pi]} \sup_{x,y\in[0,1]} \left| Q_{pq}^{\omega} - M_{[Q_{pq}^\omega]} \right| = \Op\left( \QBartlett \sqrt{\frac{\log T}{T B_R^2}} + B_R^2 \right)
\end{equation}
where
$$
M_{[Q_{00}^\omega]} = 2\pi g(x)g(y) f_\omega(x,y) , \qquad
M_{[Q_{10}^\omega]} = M_{[Q_{01}^\omega]} = 0.
$$
\end{lemma}
\begin{proof}
Denote
$M_{[Q_{00,h}]}(x,y) = g(x)g(y) R_h(x,y)$ and $M_{[Q_{10,h}]}(x,y)= M_{[Q_{01,h}^\omega]}(x,y) = 0$.
Similarly as in the proof of Lemma \ref{lemma:asymptotics_of_Q_for_F}, decompose
\begin{multline}\label{eq:lemma16_eq2}
\left| Q_{pq}^\omega - M_{[Q_{pq}^\omega]} \right|
\leq
\left|
\sum_{h=-\QBartlett}^\QBartlett
W_h e^{-\I h\omega} \frac{\mathcal{N}_h}{\hat{\mathcal{N}}_h} Q_{pq}^{(h)}
-
\sum_{h=-\infty}^\infty
M_{[Q_{pq,h}]} e^{-\I h\omega}
\right|
\leq\\\leq
\sum_{h=-\QBartlett}^\QBartlett
W_h \left| Q_{pq}^{(h)} \right|
\left| \frac{\mathcal{N}_h}{\hat{\mathcal{N}}_h} -1\right|
+
\sum_{h=-\QBartlett}^\QBartlett
W_h
\left| Q_{pq}^{(h)} - \E Q_{pq}^{(h)} \right|
+
\sum_{h=-\QBartlett}^\QBartlett
W_h
\left| \E Q_{pq}^{(h)} - M_{[Q_{pq,h}]} \right|
+\\+
\frac{1}{\QBartlett}
\sum_{h=-\QBartlett}^\QBartlett
|h| \left| M_{[Q_{pq,h}]} \right|
+
\sum_{|h|\geq \QBartlett}
\left| M_{[Q_{pq,h}]} \right|.
\end{multline}
Under the assumption \ref{assumption:D14}, the last two terms on the right-hand side of \eqref{eq:lemma16_eq2} converge to zero uniformly in $x,y\in[0,1]$ by Kronecker's lemma. Assuming further the assumption \ref{assumption:B.6}, these terms are in fact of order $O(L^{-1})$ uniformly in $x,y\in[0,1]$.
The first term on the right-hand side of \eqref{eq:lemma16_eq2} is of order $\Op(\QBartlett T^{-1/2})$ uniformly in $x,y\in[0,1]$.
The bias term, third term on the right-hand side of \eqref{eq:lemma16_eq2}, is of order $\Op( \QBartlett B_R^2 )$ which is shown exactly as in the proof of Lemma \ref{lemma:asymptotics_of_Q_for_R}.

It remains to treat the second term on the right-hand side of \eqref{eq:lemma16_eq2}, for which we start with the observation
\begin{equation}\label{eq:lemma16_eq3}
\sup_{\omega\in[-\pi,\pi]} \sup_{x,y\in[0,1]}
\sum_{h=-\QBartlett}^\QBartlett
W_h
\left| Q_{pq}^{(h)} - \E Q_{pq}^{(h)} \right|
\leq
\sum_{h=-\QBartlett}^\QBartlett
\sup_{x,y\in[0,1]}
\left| Q_{pq}^{(h)} - \E Q_{pq}^{(h)} \right|.
\end{equation}
Denote $a_T = \left(\log T / (T B_R^2) \right)^{-1/2}$.
To show the order $\Op(\QBartlett a_T)$ of the right-hand side of \eqref{eq:lemma16_eq3} we investigate the probabilities for some $M>0$
\begin{multline}\label{eq:lemma16_eq4}
\Prob\left(
\sum_{h=-\QBartlett}^\QBartlett
\sup_{x,y\in[0,1]}
\left| Q_{pq}^{(h)} - \E Q_{pq}^{(h)} \right| > M \QBartlett a_T
\right)
\leq
\sum_{h=-\QBartlett}^\QBartlett
\Prob\left(
\sup_{x,y\in[0,1]}
\left| Q_{pq}^{(h)} - \E Q_{pq}^{(h)} \right| > \frac{M \QBartlett a_T}{2\QBartlett+1}
\right)
\leq\\\leq
\sum_{h=-\QBartlett}^\QBartlett
\Prob\left(
\sup_{x,y\in[0,1]}
\left| Q_{pq}^{(h)} - \E Q_{pq}^{(h)} \right| > \frac{1}{3}M a_T
\right)
\end{multline}
We bound the probabilities on the right-hand side of \eqref{eq:lemma16_eq4} using the proof techniques presented in \citet[Thm 2]{hansen2008uniform}.
For the simplification of the notation and the proof we shall assume that the numbers of observation locations are deterministic and constant,
\begin{equation}\label{eq:lemma14_eq1}
N_1=\dots=N_T = N^{max} \equiv N \geq 2.
\end{equation}
Without this assumption, all bounds must be conditioned on these counts and the unconditional statements follow from the fact that $(1/T) \mathcal{N}_h = (\E N)^2 + \Op(T^{-1/2})$ for $h\neq 0$ and $(1/T) \mathcal{N}_0 = (\E\{ N(N-1)\}) + \Op(T^{-1/2})$ where the convergences are uniform in $|h|<T/3$.
Under the technical assumption \eqref{eq:lemma14_eq1}, $\mathcal{N}_h = (T-|h|)N^2$ for $h\neq 0$ and $\mathcal{N}_0 = T N (N-1)$.

From the assumption \ref{assumption:D15} we may take $T$ to be sufficiently large so that
\begin{equation}\label{eq:lemma14_eq1.5}
\QBartlett \leq \frac{1}{2}\sqrt{\frac{\log T}{T B_R^2}}^{-\frac{s-2}{s-1}}.
\end{equation}

Our proof follows essentially the same steps \citet[Thm 2]{hansen2008uniform}, the only difference is that we need to keep track of the uniformity in $h$ and adjust the convergence rate for the growing $\QBartlett$. 

Using the notation \eqref{eq:proof_strong_mixing_eq0.74}, \eqref{eq:proof_strong_mixing_eq0.75}, and \eqref{eq:proof_strong_mixing_eq0.76} rewrite  $Q_{pq}^{(h)}$ as
\begin{align*}
Q_{pq}^{(h)}(\tilde{x}) &= \frac{1}{\mathcal{N}_h B_R^2}
\sum_{i=1}^{\mathcal{N}_h}
\tilde{Y}^{h,1}_i
k\left(
\frac{\tilde{X}_i^h -\tilde{x}}{B_R}
\right) \\
&= \frac{1}{\mathcal{N}_h B_R^2}
\sum_{i=1}^{\mathcal{N}_h}
\tilde{Z}^{h,1}_i(\tilde{x})
\end{align*}
where $k(u,v) = u^p v^q K(u)K(v)$.

The proof consists of three steps. Firstly we replace $ \tilde{Y}_i^{h,r} $ with the truncated process $ \tilde{Y}_i^{h,r} \mathbb{1}_{[ |\tilde{Y}_i^{h,r}|\leq \tau_{T} ]}$ where $\tau_T = a_T^{-1/(s-1)}$. Secondly, we replace the supremum over $\tilde{x}\equiv(x,y)\in[0,1]$ with a maximisation over a finite $N_{g}$-point grid. And finally, with the help of the exponential inequality from Lemma \ref{lemma13} we bound the remainder.

Define
$$ R^{h,r}(\tilde{x}) = \hat{\Psi}^{h,r}(\tilde{x}) - \frac{1}{\Nh B_R^2}
\sum_{i=1}^{\Nh}
Z_i^{h,r}(\tilde{x})
\mathbb{1}_{\left[ \tilde{Y}_i \leq \tau_T \right]}.
$$
Following the same steps as in the proof of \citet[Thm 2]{hansen2008uniform}, we bound
$$ \left| \Ez{ R^{h,r}(\tilde{x}) } \right| = \Op\left( \tau_T^{-(s-1)} \right) = \Op(a_T)$$
uniformly in $|h|<T/3$.

Thus replacing $ \tilde{Y}_i $ with $ \tilde{Y}_i \mathbb{1}_{[ |\tilde{Y}_i| \leq \tau_{T} ]}$ yields only an error of order $\Op(a_T)$ and we therefore assume for the rest of the proof that $\tilde{Y}_i \leq \tau_T$.

The second step of the proof introduces a discretization of the square $[0,1]^2$ which can be covered by a regular grid of $N_g = 2 B_R^{-2} a_T^{-2}$ points such that for each $(x,y)\in[0,1]^2$, the closest grid point $\tilde{x}_j \equiv (x_j,y_j)$ is at a distance of at most $B_R a_T$ distance. Denote this discretization as $A_j \subset [0,1]^2, j=1,\dots,N_g$.

Thanks to the assumption \ref{assumption:D11}, for all $\tilde{x}_1,\tilde{x}_2 \in [0,1]^2$
satisfying $\|\tilde{x}_1 - \tilde{x}_2\| \leq \delta \leq \tilde{L}$, we have the bound
\begin{equation}\label{eq:lemma14_eq2}
\left| k(\tilde{x}_1) - k(\tilde{x}_2) \right| \leq \delta k^*(\tilde{x}_1)
\end{equation}
where $k^* : \mathbb{R}^2\to\mathbb{R}$ is a bounded integrable function. Indeed, if $k(\cdot)$ satisfies the compact support condition of \ref{assumption:C} and is Lipschitz then $k^*(u) = \Lambda_1 1_{[ \|u\|\leq 2\tilde{L} ]}$. If on the other hand $k(u)$ satisfies the differentiability condition of \ref{assumption:C}, then we may put $k^*(u) = \Lambda_1 1_{[\|u\|\leq 2\tilde{L}]} + \|u-\tilde{L}\|^{-\eta}$.

The inequality \eqref{eq:lemma14_eq2} implies that if $a_T \leq \tilde{L}$ then for $\tilde{x}\in A_j$ we have
$\| \tilde{x}- \tilde{x}_j \|/B_R \leq a_T$ and, for $T$ large enough such that $a_T\leq\tilde{L}$,
$$ \left|
k\left(\frac{\tilde{x}-\tilde{X}_i^h}{B_R}\right)
-
k\left(\frac{\tilde{x}_j-\tilde{X}_i^h}{B_R}\right)
\right|
\leq 
a_T k^*\left( \frac{\tilde{x}_j-X_i}{B_R} \right)
. $$

Define
$$ \tilde{\Psi}^{h,r} (\tilde{x}) =
\frac{1}{\mathcal{N}_h B_R^2}
\sum_{i=1}^{\mathcal{N}_h}
\tilde{Y}_i^{h,r}
k^*\left( \frac{\tilde{x}-\tilde{X}_i^h}{B_R} \right),
$$
that is, a modification of $\hat{\Psi}^{h,r}$ where $k(\cdot)$ is replaced by $k^*(\cdot)$.
Note that by the assumptions \ref{assumption:D4} and \ref{assumption:D10},
$\E \left| \tilde{\Psi}^{h,r}(\tilde{x}) \right|$
is bounded uniformly in $h\in\mathbb{Z}$ and $r=0,1$.
Following the steps in the proof of \citet[Thm 2]{hansen2008uniform}, we conclude that
$$ \sup_{\tilde{x}\in A_j}
\left| \hat{\Psi}^{h,r}(\tilde{x}) - \E \hat{\Psi}^{h,r}(\tilde{x}) \right|
\leq
\left| \hat{\Psi}^{h,r}(\tilde{x}_j) - \E\hat{\Psi}^{h,r}(\tilde{x}_j) \right|
+
\left| \tilde{\Psi}^{h,r}(\tilde{x}_j) - \E\tilde{\Psi}^{h,r}(\tilde{x}_j) \right|
+ 2 a_T M,
$$
for $ M > \E| \tilde{\Psi}^{h,r}(\tilde{x})|$, and
\begin{align}
\Prob\Bigg( \sup_{\tilde{x}\in[0,1]^2} &\left| \hat{\Psi}^{h,r}(x) - \E \hat{\Psi}^{h,r}(x) \right| > 3  M a_T \Bigg)
\leq \nonumber \\
\label{eq:proof_strong_mixing_eq3.5}
&\leq
N_g \max_{j=1,\dots,N_g} \Prob\left(
	\left| |\hat{\Psi}^{h,r}(\tilde{x}_j) - \E\hat{\Psi}^{h,r}(\tilde{x}_j) | \right| >  M a_T
\right)
+\\&+
\label{eq:proof_strong_mixing_eq4}
N_g \max_{j=1,\dots,N_g} \Prob\left(
	\left| |\tilde{\Psi}^{h,r}(\tilde{x}_j) - \E\tilde{\Psi}^{h,r}(\tilde{x}_j) | \right| >  M a_T
\right)
\end{align}
The terms \eqref{eq:proof_strong_mixing_eq3.5} and \eqref{eq:proof_strong_mixing_eq4} are bounded likewise because the only difference between them is the presence of $k(\cdot)$ and $k^*(\cdot)$. Next we show how to bound \eqref{eq:proof_strong_mixing_eq3.5}.

By the definition \eqref{eq:proof_strong_mixing_eq1} of $Z_i^{h,r}(\tilde{x})$ we notice that  $|Z_i^{h,r}(\tilde{x})| \leq \tau_T \bar{K} \equiv b_T$ because $|\tilde{Y}_i^{h,r}| \leq \tau_T$ and $\left| k((\tilde{x}-\tilde{X}_i^h)/B_R) \right| \leq \bar{k}$ where $\bar{k}$ is the upper bound of the bounded function $k(\cdot)$. Therefore, by Lemma \ref{lemma12}, for $m$ sufficiently large we have, uniformly in $|h|<m/3$,
$$ \sup_{\tilde{x}\in[0,1]^2} \E\left( \sum_{i=1}^{m} Z_i^{h,r}(\tilde{x}) \right)^2 \leq \Theta m B_R^2.$$

Put $m = (a_T \tau_T)^{-1}$ and we conclude that $m<T$ and $m<\epsilon b_T/4$ for $\epsilon=Ma_T T B_R^2$ for $T$ sufficiently large. Therefore by Lemma \ref{lemma13} for any $\tilde{x}\in[0,1]$
\begin{align*}
\Prob\Big( \Big| \hat{\Psi}^{h,r}(\tilde{x}) &-\E\hat{\Psi}^{h,r}(\tilde{x}) \Big| > M a_T \Big)
=
\Prob\left( \left| \sum_{i=1}^{\mathcal{N}_h} Z_i^{h,r}(\tilde{x}) \right| > M a_T \mathcal{N}_h B_R^2 \right)
\leq\\ &\leq
4\exp\left( - \frac{M^2 a_T^2 T^2 B_R^2 }{64 \Theta \mathcal{N}_h B_R^2 + 6 \bar{k} M T B_R^2 }\right)
+
4 \frac{\mathcal{N}_h}{m}\alpha_m
\leq\\&\leq
4 \exp \left( - \frac{M^2 \log T}{64 \left(N^{max}\right)^2 \Theta + 6 \bar{k} M} \right)
+
4 \left(N^{max}\right)^2 T \tilde{A} (m-|h|)^{-\beta} m^{-1}
\leq\\&\leq
4 T^{ -M/\left(64 \left(N^{max}\right)^2 + \bar{k}\right) } + 4\left(N^{max}\right)^2 \tilde{A}T \left( \frac{1}{2} m \right)^{-\beta} m^{-1}
\leq\\&\leq
4 T^{ -M/\left(64 \left(N^{max}\right)^2 + \bar{k}\right) } + 4(2^\beta)\left(N^{max}\right)^2 \tilde{A}T a_T^{1+\beta} \tau_T^{1+\beta}
\end{align*}
where the second inequality comes from the fact that the time-series $\{Z_i^{h,r}(\tilde{x})\}$ is strong mixing with coefficients $\alpha_m \leq \tilde{A} (m-|h|)^{-\beta}$ for $m\geq |h|$, the third inequality is due to \eqref{eq:lemma14_eq1.5}, and the final one by taking $M>\Theta$.
Since $N_g \leq 2 B_R^{-2} a_T^{-2}$ we have from the above inequality and \eqref{eq:proof_strong_mixing_eq3.5} and \eqref{eq:proof_strong_mixing_eq4} that
\begin{equation}\label{eq:proof_strong_mixing_eq5}
\Prob\left( \sup_{\tilde{x}\in[0,1]^2} \left| \hat{\Psi}^{h,r}(\tilde{x}) - \E \hat{\Psi}^{h,r}(\tilde{x}) \right| > 3 M a_T \right)
\leq O\left( C_{1,T} \right) + O\left( C_{2,T} \right)
\end{equation}
where
\begin{align*}
C_{1,T} &= B_R^{-2} a_T^{-2} T^{-M/(64+6\bar{k})} \\
C_{2,T} &= B_R^{-2} T a_T^{-1+\beta}\tau_T^{1+\beta}.
\end{align*}
Returning to the inequalities \eqref{eq:lemma16_eq3} and \eqref{eq:lemma16_eq4}, we conclude that
\begin{equation}\label{eq:proof_strong_mixing_eq6}
\Prob\left(
\sum_{h=-\QBartlett}^\QBartlett
\sup_{x,y\in[0,1]}
\left| Q_{pq}^{(h)} - \E Q_{pq}^{(h)} \right|
> M \QBartlett a_T
\right)
\leq \QBartlett \left[ O(C_{1,T}) + O(C_{2,T}) \right]
\end{equation}

Assumption \ref{assumption:D12} implies that $(\log T)B_R^{-2} = o(T^\theta)$ and therefore also $B_R^{-2} = o\left( T^\theta \right)$ and $a_T = ((\log T)B_R^{-2}T^{-1})^{1/2} = o( T^{-(1-\theta)/2} )$.
For $M$ sufficiently large and by the assumptions \ref{assumption:D15} and \ref{assumption:D16}
\begin{align*}
\QBartlett C_{1,n} &= o\left( T^{\theta_F+(1-\theta_F) - M/\left(64 (N^{max})^2+6\bar{k} \right) + (1-\theta_F)(s-2)/(s-1)/2 } \right) = o(1), \\
\QBartlett C_{2,n} &= o\left( T^{\theta_F+1-(1-\theta_F)\left[ 1+\beta-2-(1+\beta)/(s-1) - (s-2)/(s-1) \right]/2 } \right) = o(1).
\end{align*}
Thus \eqref{eq:proof_strong_mixing_eq6} is of order $o(1)$ and we conclude, together with the rates
of the other terms of \eqref{eq:lemma16_eq2}, the rates
\eqref{eq:lemma16_eq0.5} and \eqref{eq:lemma16_eq1}.

\end{proof}

\begin{lemma}\label{lemma15}
Under the assumptions \ref{assumption:D1} --- \ref{assumption:D11} and \ref{assumption:D13} --- \ref{assumption:D16},
\begin{equation}\label{eq:lemma15_eq1}
\sup_{x,y\in[0,1]} \left| S_{pq} - M_{[S_{pq}]} \right| = \Op\left( \sqrt{\frac{\log T}{T B_R^2}} + B_R^2 \right).
\end{equation}
\end{lemma}
\begin{proof}
We decompose the estimation error as follows:
\begin{multline}\label{eq:lemma15_eq2}
\left| S_{pq} - M_{[S_{pq}]} \right| =
\left| \frac{1}{\QBartlett} \sum_{|h|<\QBartlett} W_h\left(\frac{\mathcal{N}_h}{\hat{\mathcal{N}}_h} S_{pq}^{(h)} - M_{[S_{pq}]} \right)\right|
\leq\\\leq
\frac{1}{\QBartlett} \sum_{|h|<\QBartlett} W_h \left| \left(\frac{\mathcal{N}_h}{\hat{\mathcal{N}}_h}-1\right)S_{pq}^{(h)} \right|
+
\frac{1}{\QBartlett} \sum_{|h|<\QBartlett} W_h
\left|
S_{pq}^{(h)} - \E S_{pq}^{(h)}
\right|
+
\frac{1}{\QBartlett} \sum_{|h|<\QBartlett} W_h
\left|
\E S_{pq}^{(h)} - M_{[S_{pq}]}
\right|
\end{multline}

The first term on the right hand side of \eqref{eq:lemma15_eq2} is of order $O(T^{-1/2})$, uniformly in $x,y\in[0,1]$, because $(1/T)\mathcal{N}_h = c_h + \Op(T^{-1/2})$ and $(1/T)\hat{\mathcal{N}}_h = c_h + \Op(T^{-1/2})$ uniformly in $|h|\leq\QBartlett$. 

The third term on the right hand side of \eqref{eq:lemma15_eq2} is of order $O(B_R^2)$, uniformly in $x,y\in[0,1]$. This is shown identically as in the proof of Lemma \ref{lemma:asymptotics_of_S_for_R}.

The second term on the right hand side of order
$$S_{pq} = \E S_{pq} + \Op\left( \sqrt{\frac{\log T}{T B_R^2}} \right)$$
uniformly in $x,y\in[0,1]$ and $|h|<\QBartlett$. This is shown analogously as the proof of Lemma \ref{lemma16}. The difference is that the normalising factor $1/\QBartlett$ improves the rate to $(\log T /(T B_R^2))^{1/2}$ as opposed to $\QBartlett (\log T /(T B_R^2))^{1/2}$ as in Lemma \ref{lemma16}.
\end{proof}

\begin{proof}[Proof of Theorem \ref{thm:strong_mixing_F}]
We start with assuming that the mean function $\mu(\cdot)$ is known. Combining the results of Lemmas \ref{lemma16} and \ref{lemma15}, and the formula \eqref{eq:proof_strong_mixing_eq0.5} provides the rate
\eqref{eq:thm:strong_mixing_F_rate1}, and the rate \eqref{eq:thm:strong_mixing_F_rate2} if \ref{assumption:B.6} is assumed.

The proof is completed by the discussion that the difference between the ``raw'' covariances with and without $\mu(\cdot)$ is negligible.
\end{proof}

\color{black}

\subsection{Proof of Theorem \ref{thm:correctness_of_kriging}}

The following lemma ensures the convergence of 
$\hat{\boldsymbolmu}_{X_s|\mathbb{Y}_S}$ and $\hat{\mathbbSigma}_{X_s|\mathbb{Y}_S}$
 to their population level counterparts
\eqref{eq:conditional_distribution_Xs_mu_Sigma}.
We investigate the convergence without the Gaussianity assumption.
\begin{lemma}\label{lemma:convergence_4_correctness_of_kriging}
Under the assumptions \ref{assumption:B.1} --- \ref{assumption:B.5} and \ref{assumption:B.7} --- \ref{assumption:B.10},
$$ \sup_{x\in[0,1]} \left|
\hat{\boldsymbolmu}_{X_s|\mathbb{Y}_S}(x) -
\boldsymbolmu_{X_s|\mathbb{Y}_S}(x)
\right| = \op(1) \qquad \text{as}\quad T\to\infty,$$
$$ \sup_{x,y\in[0,1]} \left|
\hat{\mathbbSigma}_{X_s|\mathbb{Y}_S}(x,y) -
\mathbbSigma_{X_s|\mathbb{Y}_S}(x,y)
\right| = \op(1) \qquad \text{as}\quad T\to\infty.$$
\end{lemma}
\begin{proof}
We start with $\hat{\boldsymbolmu}_{X_s|\mathbb{Y}_S}$. Decompose the difference as
\begin{multline}
\label{eq:proof_of_lemma_thm3_decomp}
\left|
\hat{\boldsymbolmu}_{X_s|\mathbb{Y}_S} -
\boldsymbolmu_{X_s|\mathbb{Y}_S}
\right|
\leq
\underbrace{\left| \hat{\mu}(x) - \mu(x) \right|}_{J_1}
+\\+
\underbrace{
\left|\left[
P_s \hat{\mathbbSigma}_S \mathbb{H}_S^*
\left\{
\left( \mathbb{H}_S\hat{\mathbbSigma}_S\mathbb{H}_S^* + \hat\sigma^2 I_{\mathcal{N}_1^T} \right)^{-1}
-
\left( \mathbb{H}_S\mathbbSigma_S\mathbb{H}_S^* + \sigma^2 I_{\mathcal{N}_1^T} \right)^{-1}
\right\}
\left( \mathbb{Y}_S - \mathbb{H}_S \widehat{\boldsymbolmu}_S\right)
\right](x)
\right|}_{J_2}
+\\+
\underbrace{\left|\left\{
\left( P_s \hat{\mathbbSigma}_S \mathbb{H}_S^* - P_s \mathbbSigma_S \mathbb{H}_S^*\right)
\left( \mathbb{H}_S\mathbbSigma_S\mathbb{H}_S^* + \sigma^2 I_{\mathcal{N}_1^T} \right)^{-1}
\left( \mathbb{Y}_S - \mathbb{H}_S \widehat{\boldsymbolmu}_S\right)
\right\}
(x)\right|}_{J_3}.
\end{multline}

The first term $J_1$ on the right-hand side of \eqref{eq:proof_of_lemma_thm3_decomp} tends to zero, uniformly in $x$, as $T\to\infty$ by Theorem \ref{thm:mean_and_autocov_function}.
The second term $J_2$ and the third term $J_3$ can be rewritten as
\begin{multline*}
J_2 = \left|\left[
P_s \hat{\mathbbSigma}_S \mathbb{H}_S^*
\left\{
\left( \mathbb{H}_S\hat{\mathbbSigma}_S\mathbb{H}_S^* + \hat\sigma^2 I_{\mathcal{N}_1^T} \right)^{-1}
-
\left( \mathbb{H}_S\mathbbSigma_S\mathbb{H}_S^* + \sigma^2 I_{\mathcal{N}_1^T} \right)^{-1}
\right\}
\left( \mathbb{Y}_S - \mathbb{H}_S \widehat{\boldsymbolmu}_S\right)
\right](x)
\right|
=\\=
\left|
\widehat{\cov}( X_s(x), \mathbb{Y}_S )
^*
\left(
\var(\mathbb{Y}_S) ^{-1}
-
\widehat{\var}(\mathbb{Y}_S)^{-1}
\right)
\left( \mathbb{Y}_S - \mathbb{H}_S \widehat{\boldsymbolmu}_S\right)
\right|
\end{multline*}
\begin{multline*}
J_3 = \left|\left\{
\left( P_s \hat{\mathbbSigma}_S \mathbb{H}_S^* - P_s \mathbbSigma_S \mathbb{H}_S^*\right)
\left( \mathbb{H}_S\mathbbSigma_S\mathbb{H}_S^* + \sigma^2 I_{\mathcal{N}_1^T} \right)^{-1}
\left( \mathbb{Y}_S - \mathbb{H}_S \hat{\boldsymbolmu}_S\right)
\right\}
(x)\right|
=\\=
\left|
\left\{
	\widehat{\cov}( X_s(x), \mathbb{Y}_S )
	-
	\cov( X_s(x), \mathbb{Y}_S )
\right\}
^*
\left(
\var(\mathbb{Y}_S) ^{-1}
\right)
\left( \mathbb{Y}_S - \mathbb{H}_S \hat{\boldsymbolmu}_S\right)
\right|
\end{multline*}
where $ \cov( X_s(x), \mathbb{Y}_S ) $ is a random vector in $\mathbb{R}^{\mathcal{N}_1^S}$ whose elements are of the form $\{ R_{h_k}(x, x_{t_k,j_k}) \}_{k=1}^{N_S}$ for some lags $h_k$ and locations $x_{t_k,j_k}$ and $\var( \mathbb{Y}_S )$ is a random matrix in $\mathbb{R}^{\mathcal{N}_1^S\times \mathcal{N}_1^S}$ whose elements are of the form $\{ R_{t_{k'}-t_k}(x_{t_k,j_k}, x_{t_{k'},j_{k'}}) \}_{k,k'=1}^{ \mathcal{N}_1^S }$. The terms $\widehat{\cov}( X_s(x), \mathbb{Y}_S )$ and $\widehat{\var}(\mathbb{Y}_S)^{-1}$ are defined using the estimated autocovariance kernels.

To treat the term $J_2$ note that $\widehat{\var}(\mathbb{Y}_S)^{-1} - \var(\mathbb{Y}_S) ^{-1} \to 0$ as $T\to\infty$ by Corollary \ref{corollary:sup_consistency}. The term $\left( \mathbb{Y}_S - \mathbb{H}_S \widehat{\boldsymbolmu}_S\right)$ is bounded as $T\to\infty$ thanks to the convergence $\hat\mu \to \mu$. The term $\widehat{\cov}( X_s(x), \mathbb{Y}_S )$ is bounded uniformly in $x$ due to its convergence to $\cov( X_s(x), \mathbb{Y}_S )$, uniformly in $x$, by Corollary \ref{corollary:sup_consistency}.

The term $J_3$ is treated similarly. $\widehat{\cov}( X_s(x), \mathbb{Y}_S ) - \cov( X_s(x), \mathbb{Y}_S ) \to 0$, uniformly in $x$, by Corollary \ref{corollary:sup_consistency}. The formula for the variance
$\hat{\mathbbSigma}_{X_s|\mathbb{Y}_S}(x,y)$
can be written as
$$\hat{\mathbbSigma}_{X_s|\mathbb{Y}_S}(x,y)
= \hat{R}_0(x,y) -
\widehat{\cov}( X_s(x), \mathbb{Y}_S ) ^*
\widehat{\var}(\mathbb{Y}_S) ^{-1}
\widehat{\cov}( X_s(y), \mathbb{Y}_S ).
$$
Its convergence, uniform in $(x,y)\in[0,1]^2$, is treated similarly as above by Corollary \ref{corollary:sup_consistency}.

\end{proof}

\begin{proof}[Proof of Theorem \ref{thm:correctness_of_kriging}]
The first statement of Lemma \ref{lemma:convergence_4_correctness_of_kriging} is the statement of Theorem \ref{thm:correctness_of_kriging}.
\end{proof}

\subsection{Proof of Theorem \ref{thm:correctness_of_bands}}

\begin{proof}[Proof of Theorem \ref{thm:correctness_of_bands} ]
We start with the pointwise confidence band. Fix $x\in[0,1]$. From \ref{assumption:A.1} and the conditional distribution
$$ \frac{X_s(x) - \boldsymbolmu_{X_s|\mathbb{Y}_S}(x) }{\sqrt{\mathbbSigma_{X_s|\mathbb{Y}_S}(x,x)}} \sim N(0,1).$$
Therefore $$ \Prob\left\{ \left| X_s(x) - \boldsymbolmu_{X_s|\mathbb{Y}_S}(x)  \right| \leq \Phi^{-1}\left(1-\alpha/2\right) \sqrt{\mathbbSigma_{X_s|\mathbb{Y}_S}(x,x)} \right\} = 1-\alpha .$$
By Lemma \ref{lemma:convergence_4_correctness_of_kriging},
$$ \frac{X_s(x) - \hat{\boldsymbolmu}_{X_s|\mathbb{Y}_S}(x) }{\sqrt{\hat{\mathbbSigma}_{X_s|\mathbb{Y}_S}(x,x)}} \stackrel{d}{\to} N(0,1)$$
and thus
$$ \Prob\left\{ \left| X_s(x) - \hat{\boldsymbolmu}_{X_s|\mathbb{Y}_S}(x)  \right| \leq \Phi^{-1}\left(1-\alpha/2\right) \sqrt{\hat{\mathbbSigma}_{X_s|\mathbb{Y}_S}(x,x)} \right\}  \to 1-\alpha .$$

Now we turn our attention to the simultaneous confidence band. By the definition of the conditional distribution
$$  X_s - \boldsymbolmu_{X_s|\mathbb{Y}_S} \sim N(0, \mathbbSigma_{X_s|\mathbb{Y}_S} ).$$
By the definition of the simultaneous confidence bands \citep{degras2011simultaneous}, which was reviewed in Section \ref{subsec:functional_data_recovery},
$$  \Prob\left\{
\forall x \in [0,1]:
\left|
X_s(x) - \boldsymbolmu_{X_s|\mathbb{Y}_S}(x)
\right|
\leq
z_{\alpha,\rho}
\sqrt{\mathbbSigma_{X_s|\mathbb{Y}_S}(x,x)} \right\} = 1-\alpha.$$
Define the correlation kernel $\rho_{X_s|\mathbb{Y}_T} (x,y)$ as in \eqref{eq:simul_bands_correlation_kernel}. Assume for simplicity of the proof that $\rho_{X_s|\mathbb{Y}_T} (x,x) > 0$ for all $x\in[0,1]$. Then
$$
\frac{X_s(\cdot) - \boldsymbolmu_{X_s|\mathbb{Y}_S}(\cdot) }{ \sqrt{\mathbbSigma_{X_s|\mathbb{Y}_S}(\cdot,\cdot)} }
\sim N\left(0, \rho_{X_s|\mathbb{Y}_T} \right) $$
where the square root and the division is understood pointwise. Denote $W_{\rho}$ the law of $\sup_{x\in[0,1]} |Z_{\rho}|$ where $Z_{\rho} \sim N(0,\rho)$. Then
$$
\sup_{x\in[0,1]}
\left|
\frac{X_s(x) - \boldsymbolmu_{X_s|\mathbb{Y}_S}(x) }{ \sqrt{\mathbbSigma_{X_s|\mathbb{Y}_S}(x,x)} }
\right|
\sim
W_{ \rho_{X_s|\mathbb{Y}_T} } $$
By Lemma \ref{lemma:convergence_4_correctness_of_kriging},
$$
\sup_{x\in[0,1]}
\left|
\frac{X_s(x) - \hat{\boldsymbolmu}_{X_s|\mathbb{Y}_S}(x) }{ \sqrt{\hat{\mathbbSigma}_{X_s|\mathbb{Y}_S}(x,x)} }
\right|
\stackrel{d}{\to}
W_{ \rho_{X_s|\mathbb{Y}_T} } .$$

Note also that if $\rho_n \to \rho$ uniformly then $N(0,\rho_n) \to N(0,\rho)$ weakly, $W_{\rho_n} \to W_\rho$ weakly and therefore $z_{\alpha,\rho_n} \to z_{\alpha, \rho}$.

$$
\Prob\left\{ \sup_{x\in[0,1]}
\left|
\frac{X_s(x) - \hat{\boldsymbolmu}_{X_s|\mathbb{Y}_S}(x) }{ \sqrt{\hat{\mathbbSigma}_{X_s|\mathbb{Y}_S}(x,x)} }
\right|
\leq z_{\alpha,\hat{\rho}}
\right\} 
=
\Prob\left\{ \sup_{x\in[0,1]}
\left|
\frac{X_s(x) - \hat{\boldsymbolmu}_{X_s|\mathbb{Y}_S}(x) }{ \sqrt{\hat{\mathbbSigma}_{X_s|\mathbb{Y}_S}(x,x)} }
\right|
\frac{z_{\alpha,\rho}}{z_{\alpha,\hat{\rho}}}
\leq z_{\alpha,\rho}
\right\} \to 1-\alpha.
$$

\end{proof}

\subsection{Proof of Proposition \ref{prop:spec_density_FMA} and Proposition \ref{prop:spec_density_FAR}}

\begin{proof}[Proof of Proposition \ref{prop:spec_density_FMA}]
The formula \eqref{eq:spectral_density_FMA} is verified by calculating the autocovariance operators of the functional moving average process, which are non-zero only for a finite number of lags.

The assumptions \ref{assumption:B.3}, \ref{assumption:B.4}, \ref{assumption:B.5}, \ref{assumption:B.6} are easily verified by the smoothness of the kernels and the exponential decay of the norm of the autocovariance operators. Verifying the condition \eqref{eq:assumption_stationarity_summability_of_autocovariance} in the supremum sense yields the existence of the spectral density in the kernel sense \eqref{eq:spectral_density_kernel_operator}.
\end{proof}

\begin{proof}[Proof of Proposition \ref{prop:spec_density_FAR}]
The existence, the uniqueness, and the stationarity is treated by \citet{bosq2012linear}. The Gaussianity is also immediate. We now verify the formula \eqref{eq:spectral_density_FAR}. We can write the inversions on the right-hand side of \eqref{eq:spectral_density_FAR} as a Neumann series:
\begin{equation}
\label{eq:proof_spec_FAR_eq1}
(I - \mathcal{A}  e^{-\I\omega})^{-1} \mathcal{S} (I - \mathcal{A}\transpose  e^{\I\omega})^{-1}
=
\left( \sum_{j=0}^\infty \mathcal{A}^j e^{-\I\omega j} \right)
\mathcal{S}
\left( \sum_{j=0}^\infty \left(\mathcal{A}^j\right)\transpose e^{\I\omega j} \right).
\end{equation}
Fix $h\geq 0$. Expanding the sums on the right-hand side of \eqref{eq:proof_spec_FAR_eq1}, in order to obtain the term with $e^{-\I\omega h}$ one has to sum up
\begin{equation}
\label{eq:proof_spec_FAR_eq2}
\sum_{j=0} \mathcal{A}^{h+j} S \left(\mathcal{A}^j\right)\transpose e^{-\I\omega h} = \mathcal{A}^h \mathscr{R}_0 e^{-\I\omega h} = \mathscr{R}_h e^{-\I\omega h}
\end{equation}
where $\mathscr{R}_0 = \sum_{j=0}^\infty \mathcal{A}^j \mathcal{S}\left(\mathcal{A}^j\right)\transpose$ is the lag-0 covariance operator of the process \citep{bosq2012linear}.
Checking the analogue of \eqref{eq:proof_spec_FAR_eq2} for $h<0$ yields the formula  \eqref{eq:spectral_density_FAR}. The discussion of the assumptions is analogous to the functional moving average process.
\end{proof}

%%%%%%%%%%%%%%%%%%%%%%%%%%%%%%%%%%%%%%%%%%%%%%%%%%%%%%%%%%%%%%%%%%%%%%%%%%%%%%%%%%%%%%%%%%%%%%%%%%%%%
%%%%%%%%%%%%%%%%%%%%%%%%%%%%%%%%%%%%%%%%%%%%%%%%%%%%%%%%%%%%%%%%%%%%%%%%%%%%%%%%%%%%%%%%%%%%%%%%%%%%%

\section{Supplementary Results of Numerical Experiments}
\label{appendix:supplementary results}

\subsection{Determination of the Optimal Parameter $\QBartlett$}
\label{appendix_results:determination_Q_Bartlett}

We run a simulation study across the considered functional moving average processes $\mathbf{FMA(2)}$, $\mathbf{FMA(4)}$, and $\mathbf{FMA(8)}$, and the functional autoregressive processes $\mathbf{FAR(1)_{0.7}}$ and $\mathbf{FAR(1)_{0.9}}$. For their definitions refer to Subsection \ref{subsec:simulation_setting}. We simulated 25 independent realizations of each of the process for each pair of the considered sample size parameters $T\in\{150,300,450,600,900,1200\}$ and $N^{max}\in\{5,10,20,30,40\}$.
For each realization we selected the bandwidth parameters $B_\mu$, $B_R$, and $B_V$ for smoothing estimators by the K-fold cross-validation suggested in Section \ref{subsec:selection of B_mu, B_R and B_V}. Then we estimated the spectral density by the estimator \eqref{eq:spectral_density_estimator_bartlett_smoother} with varying value of Bartlett's span parameter $\QBartlett$ to identify what value is the optimal for the estimation of the spectral density with respect to the relative mean square error \eqref{eq:RMSE_definition}.
First five parts of Table \ref{table:best_Q_simulations} presents the optimal values of $\QBartlett$ for the considered processes and the considered sample sizes.

The optimal value of $\QBartlett$ depends on the dynamics of the functional time-series quite substantially. Especially striking is the case of the autoregressive process $\mathbf{FAR(1)_{0.9}}$ which features a higher degree of temporal dependence than the other processes. Observing the results in the first five parts of Table \ref{table:best_Q_simulations} we suggested the selection rule  \eqref{eq:Q_Bartlett decision rule} as a compromise among the considered processes.

The bottom-right part of Table \ref{table:best_Q_simulations} presents the evaluations of the rule   \eqref{eq:Q_Bartlett decision rule} for the considered sample sizes. For the evaluation we consider the average number of points per curve $\bar{N}$ to be set to the expectation of the number of points $N^{max}/2$.

\begin{table}[hbt!]
\caption{
The best $\QBartlett$ to minimize the relative mean square error \eqref{eq:RMSE_definition} of the spectral density estimation for the functional moving average processes $\mathbf{FMA(2)}$, $\mathbf{FMA(4)}$, and $\mathbf{FMA(8)}$, and the functional autoregressive processes $\mathbf{FAR(1)_{0.7}}$ and $\mathbf{FAR(1)_{0.9}}$. The table in the bottom-right corner presents the output of the selection rule \eqref{eq:Q_Bartlett decision rule}
}
\label{table:best_Q_simulations}
\centering
\begin{tabular}{crrrrrccrrrrr}
  \multicolumn{6}{c}{Best $\QBartlett$ for $\mathbf{FMA(2)}$} 
&& \multicolumn{6}{c}{Best $\QBartlett$ for $\mathbf{FAR(1)_{0.7}}$} \\[3pt]
$N^{max}$\textbackslash$T$   & 5 & 10 & 20   & 30 & 40 &
\multicolumn{1}{p{1cm}}{ } &
$N^{max}$\textbackslash$T$   & 5 & 10 & 20 & 30 & 40 \\[5pt]
150  & 5 & 5  & 6  & 6  & 6  &  & 150  & 4 & 5  & 6  & 7  & 7  \\
300  & 5 & 6  & 7  & 7  & 7  &  & 300  & 5 & 7  & 8  & 8  & 8  \\
450  & 6 & 7  & 8  & 8  & 8  &  & 450  & 7 & 8  & 8  & 10 & 10 \\
600  & 6 & 7  & 8  & 9  & 9  &  & 600  & 7 & 8  & 10 & 10 & 11 \\
900  & 7 & 8  & 9  & 10 & 10 &  & 900  & 8 & 10 & 11 & 11 & 13 \\
1200 & 7 & 9  & 10 & 10 & 11 &  & 1200 & 9 & 11 & 12 & 12 & 13 \\[25pt]

  \multicolumn{6}{c}{Best $\QBartlett$ for $\mathbf{FMA(4)}$}
&&\multicolumn{6}{c}{Best $\QBartlett$ for $\mathbf{FAR(1)_{0.9}}$} \\[3pt]
$N^{max}$\textbackslash$T$   & 5 & 10 & 20   & 30 & 40 &
\multicolumn{1}{p{1cm}}{ } &
$N^{max}$\textbackslash$T$ & 5 & 10 & 20  & 30 & 40 
\\[5pt]
150  & 7  & 8  & 8  & 9  & 9  &  & 150  & 19 & 21 & 23 & 20 & 20 \\
300  & 9  & 10 & 11 & 12 & 12 &  & 300  & 21 & 23 & 25 & 27 & 29 \\
450  & 9  & 11 & 11 & 12 & 13 &  & 450  & 26 & 31 & 36 & 30 & 30 \\
600  & 10 & 12 & 12 & 13 & 14 &  & 600  & 30 & 34 & 33 & 37 & 39 \\
900  & 12 & 13 & 15 & 15 & 16 &  & 900  & 33 & 35 & 41 & 43 & 40 \\
1200 & 13 & 14 & 16 & 17 & 17 &  & 1200 & 40 & 42 & 40 & 44 & 48 \\[25pt]

  \multicolumn{6}{c}{Best $\QBartlett$ for $\mathbf{FMA(8)}$}
&& \multicolumn{6}{c}{Selected $\QBartlett$ by \eqref{eq:Q_Bartlett decision rule} } \\[3pt]
$N^{max}$\textbackslash$T$   & 5 & 10 & 20   & 30 & 40 &
\multicolumn{1}{p{1cm}}{ } &
$N^{max}$\textbackslash$T$ & 5 & 10 & 20           & 30 & 40  \\[5pt]
150  & 13 & 14 & 14 & 14 & 13 &  & 150  & 6  & 7  & 9  & 10 & 11 \\
300  & 14 & 16 & 17 & 18 & 20 &  & 300  & 8  & 10 & 11 & 13 & 14 \\
450  & 16 & 18 & 19 & 19 & 21 &  & 450  & 9  & 11 & 13 & 15 & 16 \\
600  & 19 & 20 & 20 & 21 & 22 &  & 600  & 10 & 12 & 14 & 16 & 17 \\
900  & 19 & 23 & 24 & 25 & 24 &  & 900  & 12 & 14 & 17 & 19 & 20 \\
1200 & 21 & 25 & 26 & 27 & 25 &  & 1200 & 13 & 15 & 18 & 20 & 22

\end{tabular}
\end{table}

\newpage
\subsection{Spectral Density Estimation}
\label{appendix_results:estimation_of_spectral_density}

Table \ref{table:estimation_of_spectral_density_FRO_all} states the average relative mean square error \eqref{eq:RMSE_definition} for the considered functional moving average processes $\mathbf{FMA(2)}$,$\mathbf{FMA(4)}$, $\mathbf{FMA(8)}$, and the functional autoregressive processes $\mathbf{FAR(1)_{0.7}}$, $\mathbf{FAR(1)_{0.9}}$. The results for the functional moving average process of order 4, $\mathbf{FMA(4)}$, were already stated in Table \ref{table:estimation_of_spectral_density_FRO_MA4} in Section \ref{subsec:simulations_estimation of spectral density} without the standard deviations. Figure \ref{fig:4processes_spec_density} displays the fitted regression surface for the model \eqref{eq:simulations_linear_model} for the functional moving average processes $\mathbf{FMA(2)}$,$\mathbf{FMA(4)}$, $\mathbf{FMA(8)}$, and the functional autoregressive processes $\mathbf{FAR(1)_{0.7}}$, $\mathbf{FAR(1)_{0.9}}$.

The fitted regression surfaces have coefficients
$(\hat\beta_0,\hat\beta_1,\hat\beta_2)$ are
$(    1.85,   -0.31,   -0.54)$,
$(    2.37,   -0.34,   -0.61)$,
$(    2.12,   -0.32, -0.56)$,
and $(    2.26,   -0.24,   -0.49)$
for the functional moving average processes $\mathbf{FMA(2)}$, $\mathbf{FMA(8)}$, and the functional autoregressive processes $\mathbf{FAR(1)_{0.7}}$, $\mathbf{FAR(1)_{0.9}}$ respectively.
Therefore the conclusion of higher time-length preference of Section \ref{subsec:simulations_estimation of spectral density} remains valid.

\begin{table}[hbt!]
\caption{Average relative mean square errors (defined in \eqref{eq:RMSE_definition}) of the spectral density estimators for the considered functional time-series. The numbers in parentheses are the standard deviations of the relative mean square error. Each cell of the table (each error and its standard deviation) is the result of 100 independent simulations. The Bartlett's span parameter $\QBartlett$ was selected by the rule \eqref{eq:Q_Bartlett decision rule}
}
\label{table:estimation_of_spectral_density_FRO_all}
\centering
\begin{tabular}{cclllll}
%& $T$ \textbackslash $N^{max}$
\multicolumn{2}{r}{$T$ \textbackslash $N^{max}$}
& \multicolumn{1}{c}{5} & \multicolumn{1}{c}{10} & \multicolumn{1}{c}{20} & \multicolumn{1}{c}{30} & \multicolumn{1}{c}{40}  \\[5pt]
																											    \multirow{5}{*}{\rotatebox[origin=c]{90}{ $\mathbf{FMA(2)}$ }}	
 & 150  & 0.313 (0.076) & 0.235 (0.061) & 0.205 (0.053) & 0.173 (0.036) & 0.172 (0.049) \\
 & 300  & 0.211 (0.048) & 0.162 (0.031) & 0.134 (0.030) & 0.124 (0.029) & 0.119 (0.026) \\
 & 450  & 0.179 (0.038) & 0.131 (0.023) & 0.102 (0.020) & 0.098 (0.020) & 0.095 (0.021) \\
 & 600  & 0.157 (0.029) & 0.113 (0.023) & 0.088 (0.016) & 0.083 (0.018) & 0.077 (0.016) \\
 & 900  & 0.121 (0.019) & 0.088 (0.014) & 0.071 (0.012) & 0.065 (0.011) & 0.062 (0.010) \\
 & 1200 & 0.108 (0.020) & 0.079 (0.012) & 0.063 (0.011) & 0.056 (0.010) & 0.054 (0.009) \\[5pt]
																				
	    \multirow{5}{*}{\rotatebox[origin=c]{90}{ $\mathbf{FMA(4)}$ }}		
 & 150  & 0.312 (0.060) & 0.225 (0.063) & 0.184 (0.060) & 0.170 (0.049) & 0.165 (0.050) \\
 & 300  & 0.206 (0.040) & 0.157 (0.042) & 0.124 (0.028) & 0.115 (0.030) & 0.110 (0.033) \\
 & 450  & 0.167 (0.033) & 0.126 (0.034) & 0.097 (0.022) & 0.092 (0.027) & 0.081 (0.021) \\
 & 600  & 0.137 (0.027) & 0.107 (0.027) & 0.083 (0.017) & 0.077 (0.023) & 0.071 (0.017) \\
 & 900  & 0.115 (0.020) & 0.082 (0.015) & 0.067 (0.016) & 0.061 (0.015) & 0.056 (0.016) \\
 & 1200 & 0.096 (0.019) & 0.072 (0.015) & 0.056 (0.013) & 0.050 (0.012) & 0.047 (0.012) \\[5pt]
																				
	    \multirow{5}{*}{\rotatebox[origin=c]{90}{ $\mathbf{FMA(8)}$ }}		
 & 150  & 0.352 (0.071) & 0.263 (0.064) & 0.213 (0.064) & 0.188 (0.074) & 0.178 (0.069) \\
 & 300  & 0.253 (0.055) & 0.170 (0.043) & 0.143 (0.050) & 0.129 (0.053) & 0.127 (0.053) \\
 & 450  & 0.176 (0.048) & 0.148 (0.049) & 0.114 (0.044) & 0.091 (0.031) & 0.086 (0.043) \\
 & 600  & 0.159 (0.041) & 0.123 (0.039) & 0.093 (0.036) & 0.080 (0.031) & 0.081 (0.036) \\
 & 900  & 0.128 (0.030) & 0.098 (0.030) & 0.074 (0.029) & 0.062 (0.023) & 0.060 (0.026) \\
 & 1200 & 0.101 (0.026) & 0.071 (0.023) & 0.055 (0.020) & 0.049 (0.017) & 0.051 (0.018) \\[5pt]
																				
	\multirow{5}{*}{\rotatebox[origin=c]{90}{ $\mathbf{FAR(1)_{0.7}}$ }}									 & 150  & 0.359 (0.082) & 0.289 (0.070) & 0.232 (0.069) & 0.211 (0.064) & 0.213 (0.066) \\
 & 300  & 0.257 (0.067) & 0.195 (0.048) & 0.154 (0.044) & 0.142 (0.042) & 0.138 (0.042) \\
 & 450  & 0.212 (0.041) & 0.155 (0.037) & 0.123 (0.032) & 0.114 (0.031) & 0.111 (0.030) \\
 & 600  & 0.187 (0.047) & 0.129 (0.029) & 0.108 (0.024) & 0.100 (0.026) & 0.090 (0.025) \\
 & 900  & 0.147 (0.031) & 0.107 (0.022) & 0.084 (0.019) & 0.075 (0.018) & 0.069 (0.020) \\
 & 1200 & 0.125 (0.022) & 0.094 (0.019) & 0.073 (0.018) & 0.063 (0.015) & 0.060 (0.015) \\[5pt]
																				
	\multirow{5}{*}{\rotatebox[origin=c]{90}{ $\mathbf{FAR(1)_{0.9}}$ }}	
 & 150  & 0.564 (0.097) & 0.466 (0.112) & 0.460 (0.117) & 0.454 (0.149) & 0.399 (0.135) \\
 & 300  & 0.433 (0.075) & 0.372 (0.102) & 0.334 (0.101) & 0.272 (0.098) & 0.291 (0.113) \\
 & 450  & 0.374 (0.074) & 0.324 (0.078) & 0.283 (0.092) & 0.239 (0.081) & 0.216 (0.077) \\
 & 600  & 0.305 (0.068) & 0.272 (0.062) & 0.216 (0.074) & 0.216 (0.083) & 0.192 (0.073) \\
 & 900  & 0.282 (0.054) & 0.227 (0.068) & 0.179 (0.061) & 0.165 (0.072) & 0.146 (0.061) \\
 & 1200 & 0.241 (0.061) & 0.194 (0.058) & 0.152 (0.059) & 0.137 (0.059) & 0.125 (0.059)

\end{tabular}
\end{table}

\begin{figure}
\centering
\includegraphics[width=\textwidth]{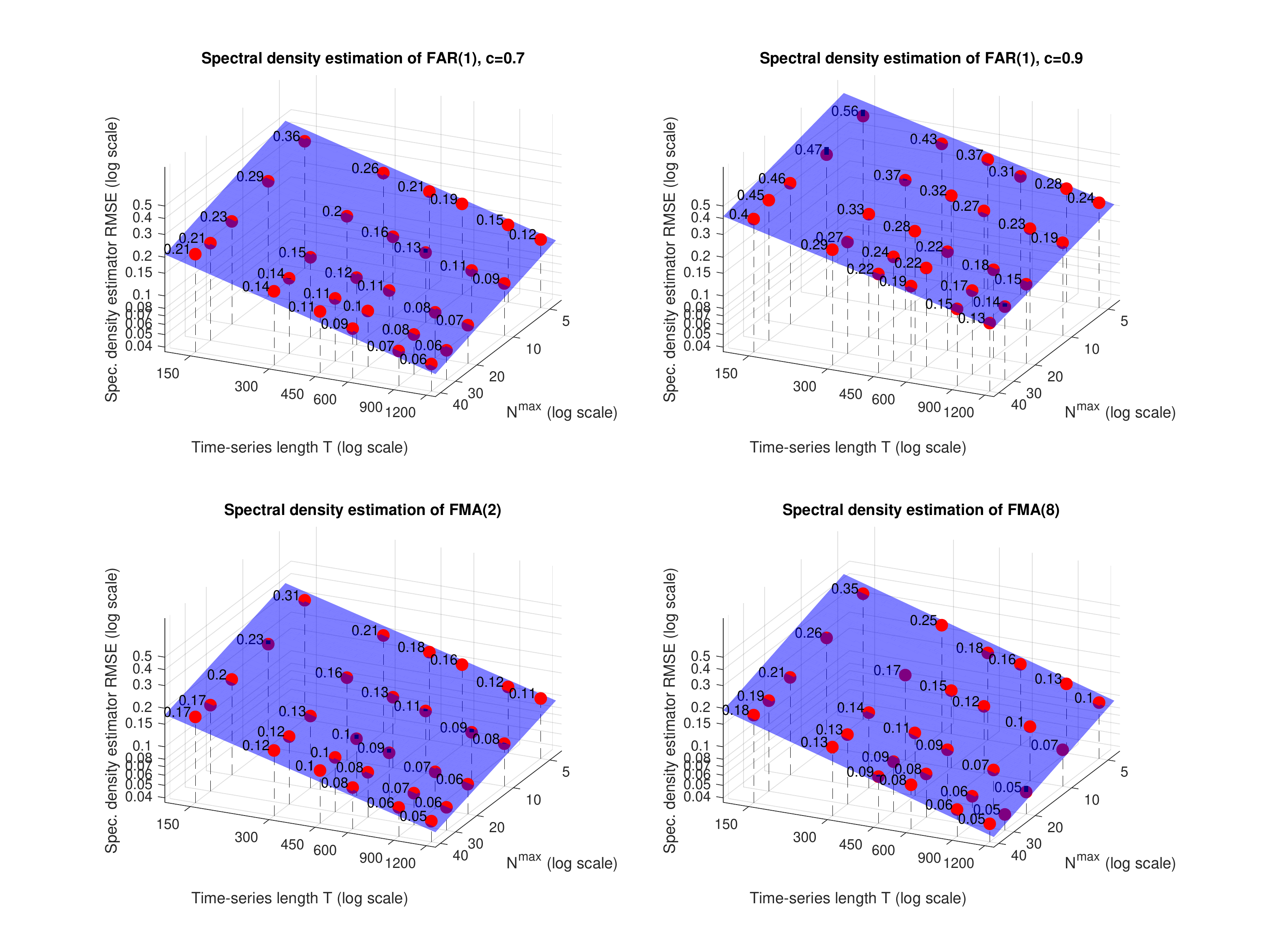}
\caption{The dependence of spectral density estimation relative mean square error (red points) of $\mathbf{FAR(1)_{0.9}}$ on the sample size parameters $T$ and $N^{max}$. The blue plane is the estimated regression surface in model \eqref{eq:simulations_linear_model}.}
\label{fig:4processes_spec_density}
\end{figure}

%\clearpage
\subsection{Functional Data Recovery}
\label{appendix_results:recovery_of_functional_data}

Table~\ref{table:kriging_FMA_2}, Table~\ref{table:kriging_FMA_4}, Table~\ref{table:kriging_FMA_8}, Table~\ref{table:kriging_FAR_07}, and Table~\ref{table:kriging_FAR_09} summarize the performance of dynamic and static recovery methods. Because of the reasons explain in Section \ref{subsec:simulations_recovery}, the relative means square error is sensitive to poor estimation of the measurement error variance parameter $\sigma^2$. Therefore we take into account only those simulations where $\hat{\sigma}> 0.05$ and calculate the median relative mean square error and the corresponding inter-quartile range instead of the mean of the errors and their standard deviation.

The column \textit{Relative gain} of Table \ref{table:kriging_FMA_4} for the functional moving average process of order 4, $\mathbf{FMA(4)}$, corresponds to the data in Table \ref{table:recovery_relative_gain_MA4}.

\begin{table}[!htbp]
\caption{
Median relative mean square error \eqref{eq:simulations_RMSE} of the dynamic and static recovery and the relative gain \eqref{eq:simulations_relative gain} between them. Each row of the table is result of 100 independent simulations of the functional moving average process $\mathbf{FMA(2)}$
}
\label{table:kriging_FMA_2}
\begin{tabular}{ccccc}
$T$   & $N^{max}$ &
\multicolumn{1}{p{3.5cm}}{Median dynamic recovery relative mean square error (inter-quartile range)}&
\multicolumn{1}{p{3.5cm}}{Median static recovery relative mean square error (inter-quartile range)}&
\multicolumn{1}{p{1.0cm}}{Relative gain}													
\\[25pt]
     & 5  & 0.323 (0.067) & 0.514 (0.067) & 59 \% \\
     & 10 & 0.193 (0.054) & 0.263 (0.054) & 37 \% \\
150  & 20 & 0.105 (0.041) & 0.140 (0.041) & 34 \% \\
     & 30 & 0.076 (0.034) & 0.091 (0.034) & 19 \% \\
     & 40 & 0.055 (0.025) & 0.072 (0.025) & 31 \% \\[5pt]
     & 5  & 0.289 (0.042) & 0.440 (0.042) & 52 \% \\
     & 10 & 0.167 (0.034) & 0.240 (0.034) & 44 \% \\
300  & 20 & 0.094 (0.031) & 0.124 (0.031) & 32 \% \\
     & 30 & 0.068 (0.016) & 0.081 (0.016) & 20 \% \\
     & 40 & 0.054 (0.015) & 0.064 (0.015) & 18 \% \\[5pt]
     & 5  & 0.274 (0.038) & 0.426 (0.038) & 55 \% \\
     & 10 & 0.161 (0.029) & 0.226 (0.029) & 40 \% \\
450  & 20 & 0.090 (0.024) & 0.118 (0.024) & 31 \% \\
     & 30 & 0.062 (0.015) & 0.076 (0.015) & 23 \% \\
     & 40 & 0.051 (0.010) & 0.062 (0.010) & 21 \% \\[5pt]
     & 5  & 0.274 (0.035) & 0.395 (0.035) & 44 \% \\
     & 10 & 0.153 (0.030) & 0.217 (0.030) & 41 \% \\
600  & 20 & 0.085 (0.013) & 0.111 (0.013) & 31 \% \\
     & 30 & 0.062 (0.011) & 0.076 (0.011) & 24 \% \\
     & 40 & 0.049 (0.010) & 0.059 (0.010) & 20 \% \\[5pt]
     & 5  & 0.259 (0.029) & 0.376 (0.029) & 45 \% \\
     & 10 & 0.147 (0.023) & 0.206 (0.023) & 41 \% \\
900  & 20 & 0.084 (0.017) & 0.110 (0.017) & 31 \% \\
     & 30 & 0.061 (0.008) & 0.074 (0.008) & 22 \% \\
     & 40 & 0.048 (0.006) & 0.057 (0.006) & 20 \% \\[5pt]
     & 5  & 0.251 (0.024) & 0.368 (0.024) & 46 \% \\
     & 10 & 0.142 (0.016) & 0.199 (0.016) & 40 \% \\
1200 & 20 & 0.083 (0.012) & 0.107 (0.012) & 29 \% \\
     & 30 & 0.059 (0.008) & 0.073 (0.008) & 24 \% \\
     & 40 & 0.047 (0.007) & 0.055 (0.007) & 17 \%
\end{tabular}
\end{table}

\begin{table}[!htbp]
\caption{
Median relative mean square error \eqref{eq:simulations_RMSE} of the dynamic and static recovery and the relative gain \eqref{eq:simulations_relative gain} between them. Each row of the table is result of 100 independent simulations of the functional moving average process $\mathbf{FMA(4)}$
}
\label{table:kriging_FMA_4}
\begin{tabular}{ccccc}
$T$   & $N^{max}$ &
\multicolumn{1}{p{3.5cm}}{Median dynamic recovery relative mean square error (inter-quartile range)}&
\multicolumn{1}{p{3.5cm}}{Median static recovery relative mean square error (inter-quartile range)}&
\multicolumn{1}{p{1.0cm}}{Relative gain}
\\[25pt]
     & 5  & 0.321 (0.089) & 0.537 (0.089) & 67 \% \\
     & 10 & 0.190 (0.053) & 0.263 (0.053) & 38 \% \\
150  & 20 & 0.101 (0.034) & 0.139 (0.034) & 38 \% \\
     & 30 & 0.075 (0.035) & 0.092 (0.035) & 23 \% \\
     & 40 & 0.051 (0.017) & 0.067 (0.017) & 30 \% \\[5pt]
     & 5  & 0.284 (0.058) & 0.435 (0.058) & 53 \% \\
     & 10 & 0.169 (0.038) & 0.235 (0.038) & 39 \% \\
300  & 20 & 0.091 (0.026) & 0.120 (0.026) & 33 \% \\
     & 30 & 0.063 (0.019) & 0.082 (0.019) & 31 \% \\
     & 40 & 0.049 (0.014) & 0.062 (0.014) & 26 \% \\[5pt]
     & 5  & 0.267 (0.040) & 0.405 (0.040) & 52 \% \\
     & 10 & 0.157 (0.040) & 0.228 (0.040) & 45 \% \\
450  & 20 & 0.083 (0.018) & 0.115 (0.018) & 38 \% \\
     & 30 & 0.060 (0.012) & 0.078 (0.012) & 30 \% \\
     & 40 & 0.047 (0.011) & 0.059 (0.011) & 24 \% \\[5pt]
     & 5  & 0.260 (0.041) & 0.378 (0.041) & 45 \% \\
     & 10 & 0.149 (0.030) & 0.211 (0.030) & 41 \% \\
600  & 20 & 0.084 (0.018) & 0.110 (0.018) & 32 \% \\
     & 30 & 0.060 (0.011) & 0.076 (0.011) & 26 \% \\
     & 40 & 0.048 (0.009) & 0.059 (0.009) & 24 \% \\[5pt]
     & 5  & 0.239 (0.031) & 0.367 (0.031) & 54 \% \\
     & 10 & 0.141 (0.018) & 0.199 (0.018) & 41 \% \\
900  & 20 & 0.077 (0.011) & 0.105 (0.011) & 37 \% \\
     & 30 & 0.058 (0.011) & 0.075 (0.011) & 30 \% \\
     & 40 & 0.047 (0.007) & 0.057 (0.007) & 22 \% \\[5pt]
     & 5  & 0.232 (0.022) & 0.357 (0.022) & 54 \% \\
     & 10 & 0.135 (0.013) & 0.195 (0.013) & 45 \% \\
1200 & 20 & 0.079 (0.009) & 0.105 (0.009) & 34 \% \\
     & 30 & 0.057 (0.008) & 0.071 (0.008) & 26 \% \\
     & 40 & 0.046 (0.006) & 0.055 (0.006) & 21 \%
\end{tabular}
\end{table}

\begin{table}[!htbp]
\caption{
Median relative mean square error \eqref{eq:simulations_RMSE} of the dynamic and static recovery and the relative gain \eqref{eq:simulations_relative gain} between them. Each row of the table is result of 100 independent simulations of the functional moving average process $\mathbf{FMA(8)}$
}
\label{table:kriging_FMA_8}
\begin{tabular}{ccccc}
$T$   & $N^{max}$ &
\multicolumn{1}{p{3.5cm}}{Median dynamic recovery relative mean square error (inter-quartile range)}&
\multicolumn{1}{p{3.5cm}}{Median static recovery relative mean square error (inter-quartile range)}&
\multicolumn{1}{p{1.0cm}}{Relative gain}
\\[25pt]
     & 5  & 0.294 (0.072) & 0.454 (0.072) & 55 \% \\
     & 10 & 0.182 (0.083) & 0.248 (0.083) & 36 \% \\
150  & 20 & 0.094 (0.050) & 0.126 (0.050) & 34 \% \\
     & 30 & 0.064 (0.026) & 0.086 (0.026) & 34 \% \\
     & 40 & 0.050 (0.021) & 0.063 (0.021) & 26 \% \\[5pt]
     & 5  & 0.264 (0.060) & 0.393 (0.060) & 49 \% \\
     & 10 & 0.145 (0.043) & 0.210 (0.043) & 45 \% \\
300  & 20 & 0.086 (0.026) & 0.111 (0.026) & 30 \% \\
     & 30 & 0.059 (0.017) & 0.075 (0.017) & 28 \% \\
     & 40 & 0.048 (0.017) & 0.056 (0.017) & 18 \% \\[5pt]
     & 5  & 0.241 (0.056) & 0.365 (0.056) & 51 \% \\
     & 10 & 0.137 (0.028) & 0.188 (0.028) & 37 \% \\
450  & 20 & 0.080 (0.016) & 0.104 (0.016) & 31 \% \\
     & 30 & 0.058 (0.014) & 0.071 (0.014) & 22 \% \\
     & 40 & 0.046 (0.011) & 0.054 (0.011) & 20 \% \\[5pt]
     & 5  & 0.220 (0.037) & 0.341 (0.037) & 55 \% \\
     & 10 & 0.132 (0.023) & 0.191 (0.023) & 45 \% \\
600  & 20 & 0.076 (0.014) & 0.103 (0.014) & 35 \% \\
     & 30 & 0.057 (0.013) & 0.068 (0.013) & 19 \% \\
     & 40 & 0.044 (0.010) & 0.052 (0.010) & 19 \% \\[5pt]
     & 5  & 0.205 (0.029) & 0.320 (0.029) & 56 \% \\
     & 10 & 0.126 (0.025) & 0.185 (0.025) & 47 \% \\
900  & 20 & 0.073 (0.011) & 0.097 (0.011) & 34 \% \\
     & 30 & 0.053 (0.012) & 0.067 (0.012) & 26 \% \\
     & 40 & 0.042 (0.008) & 0.051 (0.008) & 22 \% \\[5pt]
     & 5  & 0.204 (0.024) & 0.316 (0.024) & 54 \% \\
     & 10 & 0.122 (0.011) & 0.175 (0.011) & 44 \% \\
1200 & 20 & 0.072 (0.009) & 0.093 (0.009) & 31 \% \\
     & 30 & 0.052 (0.009) & 0.065 (0.009) & 24 \% \\
     & 40 & 0.040 (0.005) & 0.050 (0.005) & 24 \%
\end{tabular}
\end{table}

\begin{table}[!htbp]
\caption{
Median relative mean square error \eqref{eq:simulations_RMSE} of the dynamic and static recovery and the relative gain \eqref{eq:simulations_relative gain} between them. Each row of the table is result of 100 independent simulations of the functional autoregressive process $\mathbf{FAR(1)_{0.7}}$
}
\label{table:kriging_FAR_07}
\begin{tabular}{ccccc}
$T$   & $N^{max}$ &
\multicolumn{1}{p{3.5cm}}{Median dynamic recovery relative mean square error (inter-quartile range)}&
\multicolumn{1}{p{3.5cm}}{Median static recovery relative mean square error (inter-quartile range)}&
\multicolumn{1}{p{1.0cm}}{Relative gain}
\\[25pt]
     & 5  & 0.318 (0.072) & 0.480 (0.072) & 51 \% \\
     & 10 & 0.185 (0.046) & 0.256 (0.046) & 39 \% \\
150  & 20 & 0.107 (0.037) & 0.122 (0.037) & 14 \% \\
     & 30 & 0.073 (0.028) & 0.084 (0.028) & 16 \% \\
     & 40 & 0.056 (0.023) & 0.059 (0.023) & 4 \%  \\[5pt]
     & 5  & 0.282 (0.049) & 0.403 (0.049) & 43 \% \\
     & 10 & 0.163 (0.030) & 0.209 (0.030) & 29 \% \\
300  & 20 & 0.092 (0.029) & 0.111 (0.029) & 21 \% \\
     & 30 & 0.065 (0.021) & 0.069 (0.021) & 6 \%  \\
     & 40 & 0.052 (0.013) & 0.054 (0.013) & 4 \%  \\[5pt]
     & 5  & 0.265 (0.044) & 0.368 (0.044) & 39 \% \\
     & 10 & 0.157 (0.027) & 0.200 (0.027) & 28 \% \\
450  & 20 & 0.088 (0.018) & 0.104 (0.018) & 18 \% \\
     & 30 & 0.065 (0.016) & 0.070 (0.016) & 7 \%  \\
     & 40 & 0.049 (0.012) & 0.051 (0.012) & 4 \%  \\[5pt]
     & 5  & 0.257 (0.031) & 0.350 (0.031) & 36 \% \\
     & 10 & 0.141 (0.027) & 0.184 (0.027) & 30 \% \\
600  & 20 & 0.089 (0.020) & 0.100 (0.020) & 12 \% \\
     & 30 & 0.061 (0.011) & 0.066 (0.011) & 7 \%  \\
     & 40 & 0.050 (0.009) & 0.052 (0.009) & 5 \%  \\[5pt]
     & 5  & 0.245 (0.032) & 0.335 (0.032) & 37 \% \\
     & 10 & 0.142 (0.022) & 0.181 (0.022) & 27 \% \\
900  & 20 & 0.081 (0.013) & 0.093 (0.013) & 15 \% \\
     & 30 & 0.060 (0.013) & 0.064 (0.013) & 8 \%  \\
     & 40 & 0.047 (0.008) & 0.049 (0.008) & 5 \%  \\[5pt]
     & 5  & 0.238 (0.025) & 0.321 (0.025) & 35 \% \\
     & 10 & 0.139 (0.018) & 0.178 (0.018) & 28 \% \\
1200 & 20 & 0.080 (0.010) & 0.093 (0.010) & 17 \% \\
     & 30 & 0.059 (0.008) & 0.064 (0.008) & 8 \%  \\
     & 40 & 0.047 (0.006) & 0.048 (0.006) & 3 \%
\end{tabular}
\end{table}

\begin{table}[!htbp]
\caption{
Median relative mean square error \eqref{eq:simulations_RMSE} of the dynamic and static recovery and the relative gain \eqref{eq:simulations_relative gain} between them. Each row of the table is result of 100 independent simulations of the functional autoregressive process $\mathbf{FAR(1)_{0.9}}$
}
\label{table:kriging_FAR_09}
\begin{tabular}{ccccc}
$T$   & $N^{max}$ &
\multicolumn{1}{p{3.5cm}}{Median dynamic recovery relative mean square error (inter-quartile range)}&
\multicolumn{1}{p{3.5cm}}{Median static recovery relative mean square error (inter-quartile range)}&
\multicolumn{1}{p{1.0cm}}{Relative gain}
\\[25pt]
     & 5  & 0.297 (0.098) & 0.514 (0.098) & 73 \% \\
     & 10 & 0.181 (0.047) & 0.266 (0.047) & 47 \% \\
150  & 20 & 0.089 (0.029) & 0.129 (0.029) & 46 \% \\
     & 30 & 0.062 (0.022) & 0.083 (0.022) & 34 \% \\
     & 40 & 0.052 (0.026) & 0.064 (0.026) & 23 \% \\[5pt]
     & 5  & 0.255 (0.056) & 0.434 (0.056) & 70 \% \\
     & 10 & 0.162 (0.040) & 0.231 (0.040) & 43 \% \\
300  & 20 & 0.086 (0.025) & 0.118 (0.025) & 37 \% \\
     & 30 & 0.061 (0.019) & 0.079 (0.019) & 30 \% \\
     & 40 & 0.050 (0.014) & 0.057 (0.014) & 15 \% \\[5pt]
     & 5  & 0.242 (0.049) & 0.402 (0.049) & 66 \% \\
     & 10 & 0.146 (0.032) & 0.212 (0.032) & 45 \% \\
450  & 20 & 0.083 (0.021) & 0.109 (0.021) & 32 \% \\
     & 30 & 0.060 (0.015) & 0.069 (0.015) & 16 \% \\
     & 40 & 0.046 (0.010) & 0.055 (0.010) & 18 \% \\[5pt]
     & 5  & 0.239 (0.035) & 0.380 (0.035) & 59 \% \\
     & 10 & 0.142 (0.033) & 0.205 (0.033) & 45 \% \\
600  & 20 & 0.083 (0.018) & 0.104 (0.018) & 25 \% \\
     & 30 & 0.056 (0.011) & 0.068 (0.011) & 22 \% \\
     & 40 & 0.045 (0.009) & 0.053 (0.009) & 18 \% \\[5pt]
     & 5  & 0.226 (0.028) & 0.362 (0.028) & 60 \% \\
     & 10 & 0.133 (0.023) & 0.191 (0.023) & 43 \% \\
900  & 20 & 0.077 (0.017) & 0.100 (0.017) & 29 \% \\
     & 30 & 0.056 (0.009) & 0.066 (0.009) & 18 \% \\
     & 40 & 0.045 (0.009) & 0.052 (0.009) & 15 \% \\[5pt]
     & 5  & 0.218 (0.018) & 0.346 (0.018) & 59 \% \\
     & 10 & 0.131 (0.021) & 0.188 (0.021) & 44 \% \\
1200 & 20 & 0.075 (0.009) & 0.097 (0.009) & 29 \% \\
     & 30 & 0.054 (0.007) & 0.067 (0.007) & 24 \% \\
     & 40 & 0.044 (0.006) & 0.051 (0.006) & 16 \%
\end{tabular}
\end{table}

%%%%%%%%%%%%%%%%%%%%%%%%%%%%%%%%%%%%%%%%%%%%%%%%%%%%%%%%%%%%%%%%%%%%%%%%%%%%%%%%%%%%%%%%%%%%%%%%%%%%%
%%%%%%%%%%%%%%%%%%%%%%%%%%%%%%%%%%%%%%%%%%%%%%%%%%%%%%%%%%%%%%%%%%%%%%%%%%%%%%%%%%%%%%%%%%%%%%%%%%%%%

\clearpage

%%%%%%%%%%%%%%%%%%%%%%%%%%%%%%%%%%%%%%%%%%%%%%%%%%%%%%%%%%%%%%%%%%%%%%%%%%%%%%%%%%%%%%%%%%%%%%%%%%%%%
%%%%%%%%%%%%%%%%%%%%%%%%%%%%%%%%%%%%%%%%%%%%%%%%%%%%%%%%%%%%%%%%%%%%%%%%%%%%%%%%%%%%%%%%%%%%%%%%%%%%%

\bibliographystyle{biometrika}
\bibliography{biblio}

\end{document}